\documentclass[11pt]{article}
\usepackage{color}
\usepackage{graphicx}
\usepackage{amsmath}
\usepackage{amssymb}
\usepackage{mathrsfs}
\usepackage[numbers,sort&compress]{natbib}
\usepackage{amsthm}
\numberwithin{equation}{section}
\usepackage{pgf}
\usepackage{CJK}
\usepackage{graphicx}
\usepackage{tikz}
\usepackage{stmaryrd}
\usepackage{graphicx}
\usepackage{hyperref}
\usepackage[latin1]{inputenc}
\usepackage[T1]{fontenc}
\usepackage[english]{babel}
\usepackage{listings}
\usepackage{xcolor,mathrsfs,url}
\usepackage{amssymb}
\usepackage{amsmath}
\usepackage{ifthen}
\newtheorem{theorem}{Theorem}
\newtheorem{lemma}{Lemma}
\newtheorem{corollary}{Corollary}
\newtheorem{Proposition}{Proposition}
\newtheorem{RHP}{RHP}

\usepackage {caption}
\usepackage{float}
\usepackage{subfigure}
\hypersetup{hypertex=true,
            colorlinks=true,
            linkcolor=blue,
            anchorcolor=blue,
            citecolor=red}

\begin{document}

\title{ On the long-time asymptotics   of the  modified Camassa-Holm equation with step-like initial data}
\author{Yiling YANG$^1$\thanks{\ Email address: 19110180006@fudan.edu.cn } \  , Gaozhan LI$^1$\thanks{\ Email address: 20110180004@fudan.edu.cn } \   and  \  Engui FAN$^{1}$\thanks{\ Corresponding author and email address: faneg@fudan.edu.cn } }
\footnotetext[1]{ \  School of Mathematical Sciences  and Key Laboratory   for Nonlinear Science, Fudan   University, Shanghai 200433, P.R. China.}

\date{ }

\maketitle
\begin{abstract}
\baselineskip=18pt

We   study the long time asymptotic behavior for   the Cauchy problem   of the modified Camassa-Holm (mCH) equation with step-like initial data
\begin{align}
&m_{t}+\left(m\left(u^{2}-u_{x}^{2}\right)\right)_{x}=0, \quad m=u-u_{xx}, \nonumber \\
&u(x,0)=u_0(x)\to \left\{ \begin{array}{ll}
		A_1,   &\  x\to+\infty,\\[5pt]
	A_2,   &\  x\to-\infty,
	\end{array}\right.\nonumber
\end{align}
where  $A_1$ and $A_2$ are two positive constants.
 Our  main technical  tool is the representation of  the Cauchy problem
  with an associated matrix Riemann-Hilbert (RH) problem
and the consequent asymptotic analysis of this RH problem.   Based on the spectral analysis of the  Lax pair associated with the   mCH  equation and scattering matrix,  the solution of the  step-like initial problem  is  characterized   via the solution of a  RH problem  in the new scale $(y,t)$. 
 We adopt double  coordinates $(\xi, c)$ to  divide   the  half-plane  $\{ (\xi,c): \xi \in \mathbb{R}, \ c> 0, \ \xi=y/t\}$
   into four   asymptotic regions.
 Further using  the Deift-Zhou steepest descent method,
  we derive   different long time asymptotic expansion of the solution $u(y,t)$   in  different space-time   regions
   by the different choice of g-function.
The corresponding leading asymptotic  approximations  are given    with    the slow/fast  decay step-like background
wave in  genus-0 regions  and  elliptic  waves  in genus-2 regions.
The second term of the asymptotics is characterized by Airy function or parabolic cylinder model.
 Their residual error order is $\mathcal{O}(t^{-1})$ or $\mathcal{O}(t^{-2})$ respectively.\\[6pt]

\end{abstract}

\baselineskip=18pt

\tableofcontents

\section {Introduction}
\quad
The Camassa and Holm (CH) equation
\begin{align}
	&m_t+ (um )_x+  u_x m=0, \quad m=u-u_{x x}   \nonumber
\end{align}	
  was
first introduced by Camassa and Holm   in  \cite{Holm1} as a
model for shallow water waves, but it already appeared earlier in a list by Fuchssteiner
and Fokas \cite{Fuchssteiner1}.
The CH equation   has   attracted considerable interest and been studied extensively  due to their   rich mathematical structure and
 remarkable properties, such as   peakon  solutions,  bi-Hamiltonian,   algebro-geometric solutions,   local and
global well-posedness of the Cauchy problem     \cite{Constantin1,Constantin2,Fritz1,Qiao3, Monvel22, Minakov,
Eckhardt0, Eckhardt1, Eckhardt2,CH,CH1}.

It is observed that all nonlinear terms in the CH equation are quadratic. Over the last few years, various modifications and generalizations of the CH equation have been introduced.
For example,  Novikov  applied the  symmetry approach   to classify integrable equations of the form
$$(1-\partial_x^2) u_t=F(u, u_x, u_{xx}, \cdots) $$
into two integrable CH-type equations with cubic  nonlinearity     \cite{Novikov1}.
One is the well-known   mCH equation
 \begin{align}
	&m_{t}+\left(m\left(u^{2}-u_{x}^{2}\right)\right)_{x} =0, \quad m=u-u_{x x},\label{mcho}
\end{align}	
 and   another one is  called the  Novikov equation
 \begin{align}
&m_{t}+\left(m_xu+3mu_x\right)u=0,   \quad m=u-u_{x x}.
\end{align}
In an equivalent form, the mCH equation was given by Fokas \cite{Fokas}, Fuchssteiner \cite{BF1996}, Olver and
Rosenau \cite{PP1996}  and Qiao \cite{Qiao}, where the equation was derived from the two-dimensional
Euler system and   the M/W-shape solitons and peakon/cuspon solutions were presented. So the mCH equation (\ref{mcho})
   is also   referred to as the Fokas-Olver-Rosenau-Qiao equation  \cite{THFan},
but is mostly known as the  mCH  equation. The  mCH equations have non-smooth solitons (called peakons) as solutions \cite{Qiao,CS1,CS2}.
 The  stability and orbital stability  of peakons  for the mCH equation were further shown by Qu and Liu \cite{qu9,qu10}.
The well-posedness for  the  Cauchy problem  of the mCH equation (\ref{mcho}) was  studied  \cite{McLachlan,qu4,qu5}.
The local well-posedness and  the precise blow-up phenomena for the Cauchy problem of the mCH equation were   discussed \cite{McLachlan2,Qiao5}.
The  wave-breaking and peakons for the mCH equation were   investigated by   Gui,   Liu,   Olver  and   Qu   \cite{qu7}.
The  algebro-geometric quasiperiodic solutions were constructed by using algebro-geometric method \cite{THFan}.
With the aid of reciprocal transformation,  Backlund transformation  and
nonlinear superposition formula  for the mCH equation were given \cite{WLM}.  The  local well-posedness for classical solutions and global weak
solutions to the mCH equation (\ref{mcho}) were considered  in Lagrangian coordinates \cite{Gao2018}.
Applying the scaling transformation and taking parameter limit,  the  mCH  equation (\ref{mcho}) can  reduce  a   short pulse equation \cite{sp1,RHPsp,wellD}.

 Note that the soliton-type solutions of the mCH equation (\ref{mcho}) vanishing at infinity  are weak solutions in the form of peaked waves,
which are orbitally stable \cite{qu9,qu10}.
 On the other hand, adding to the original  mCH equation  (\ref{mcho})  a linear dispersion term $\kappa u_x$ with $\kappa >0$ leads to a form of the mCH equation \cite{Qiao2,qu8,Matsuno2}
\begin{align}\label{mch}
	&m_{t}+\left(m\left(u^{2}-u_{x}^{2}\right)\right)_{x}+\kappa u_{x}=0, \quad m=u-u_{x x},
\end{align}	
where $\kappa$ characterizes the effect of the linear dispersion.  It can be shown that  the    mCH equation  (\ref{mcho}) on a nonzero background  or  the mCH equation  (\ref{mch})
   with decaying  initial data
\begin{align}\label{mch1}
	&u(x, 0)=u_{0}(x), \quad x \in \mathbb{R},\  t>0
\end{align}	
may  support   smooth soliton solutions \cite{Matsuno2,Ivanov,Matsuno1}.   The mCH equation  (\ref{mch}) admits a Lax pair and  its  smooth dark soliton
solutions   were obtained   by the method of inverse scattering  transformation method \cite{Qiao2,Ivanov}.
  By using a reciprocal transformation and the  Hirota bilinear method,
Matsuno obtained the smooth bright multisoliton solutions for the  mCH equation (\ref{mch}) \cite{Matsuno2,Matsuno1}.
 Boutet de Monvel, Karpenko and Shepelsky first developed a RH approach to deal with  the mCH equation (\ref{mch}) with nonzero boundary
conditions \cite{Mon}.  They further present the results of
the asymptotic analysis in the solitonless case for the two space-time regions  $3/4 <\xi <1, \ 1< \xi <3$ with $\xi=y/t$ \cite{Mon2}.   Xu and Fan  applied Deift-Zhou
steepest decedent method to obtain  long-time asymptotic behavior of (\ref{mch})  with Schwartz initial value  \cite{Xurhp}.
 Very recently  we applied the  $\bar{\partial}$-steepest decedent method to  obtain its long-time asymptotics to the mCH equation (\ref{mch})
 with weighted Sobolev initial value \cite{YYLmch}.

The implementation of the rigorous asymptotic analysis to step-like Cauchy problems
for integrable equations started in the papers \cite{Monvel2009,Buck2007}, which extended the methods from Deift, Venakides, and Zhou \cite{gfunction}.  Since then, problems with step-like initial data  have
also been considered for a variety of integrable systems such as  the KdV equation \cite{KDV2012,KDV2016}, the focusing and defocusing NLS equations \cite{Biondini1,Biondini2,Monvel2011,MonvelCMP1,MonvelCMP2,lenells,jenkins},  the modified KdV equation \cite{MinmKDV1,MinmKDV2,MinmKDV3,MinmKDV4,MinmKDV5} and Camassa-Holm equation\cite{MinCH} among many
others. A wide range of important physical phenomena manifest themselves in the behavior of solutions of such problems for large times, e.g., collisionless and dispersive
shock waves \cite{NLS2018}, rarefaction waves \cite{jenkins}, modulated waves\cite{Deift1991}, elliptic waves\cite{Monvel2011} and so on.  The main feature in the long-time behavior
that distinguishes step-like initial conditions from decaying initial conditions is the
formation of an oscillatory region that connects the different behavior at $x\to\pm\infty$ of the
solution. These oscillatory regions are typically described by elliptic or hyperelliptic
modulated waves.

Very  recently,   Karpenko,  Shepelsky  and   Teschl  develop the RH formalism to the mCH equation  (\ref{mcho}) with   step-like initial data
  \begin{align}
  u(x,0)=u_0(x)\to \left\{ \begin{array}{ll}
		A_1,   &\  x\to+\infty,\\[5pt]
	A_2,   &\  x\to-\infty,
	\end{array}\right. \label{step-1}
\end{align}
where  $u(x, t)$ sufficiently fast approaches its large-$x$ limits.
A representation for the solution of this
problem was given in terms of the solution of an associated RH problem \cite{KSG2022}.

In this paper, we are interested in the long-time asymptotics   of the  mCH  equation (\ref{mcho})  with such  step-like initial data (\ref{step-1}).
For convenience to study the long time asymptotic behavior  of  the
  mCH  equation (\ref{mcho}),   we  properly deal with this  initial value.
Without loss of generality,  we  assume that $A_1<A_2$ in our initial value.
Further noting  that if   $u(x,t)$ is a solution of  the  mCH  equation (\ref{mcho}),
 then for an any constant $c\not =0$,  the function  $ c  u(x,-c^2t)$  also is a solution of the mCH equation (\ref{mcho}).
 So we can use   an  equivalent  scale transformation to $u(x,t)$ such that   $A_2=1$ and  $A_1=1/c$ with a constant $c>1$.
Finally  we consider  the following  step-like initial value in our paper
  \begin{align}
  	u(x,0)=u_0(x)\to \left\{ \begin{array}{ll}
  		1/c,   &\  x\to+\infty,\\[5pt]
  		 1,   &\  x\to-\infty.
  	\end{array}\right.\label{initial}
  \end{align}
In this way, we find that  the types   of   asymptotic expansions  for   the  mCH equation (\ref{mcho})
 are closely related to the scope of two parameter $\xi=y/t$ and $c$.   So in our paper   we adopt double  coordinates $(\xi, c)$  to divide
the upper half plane $\{(\xi, c): \xi \in \mathbb{R}, c>1\}$  into  four  different  space-time
  regions (see Figure \ref{result1}),
 in which  we will  present   different leading
 order     asymptotic approximations  for the  mCH equation (\ref{mcho})
  with step-like initial value (\ref{initial}),
 see  Theorem \ref{last} in the section \ref{sec9}.
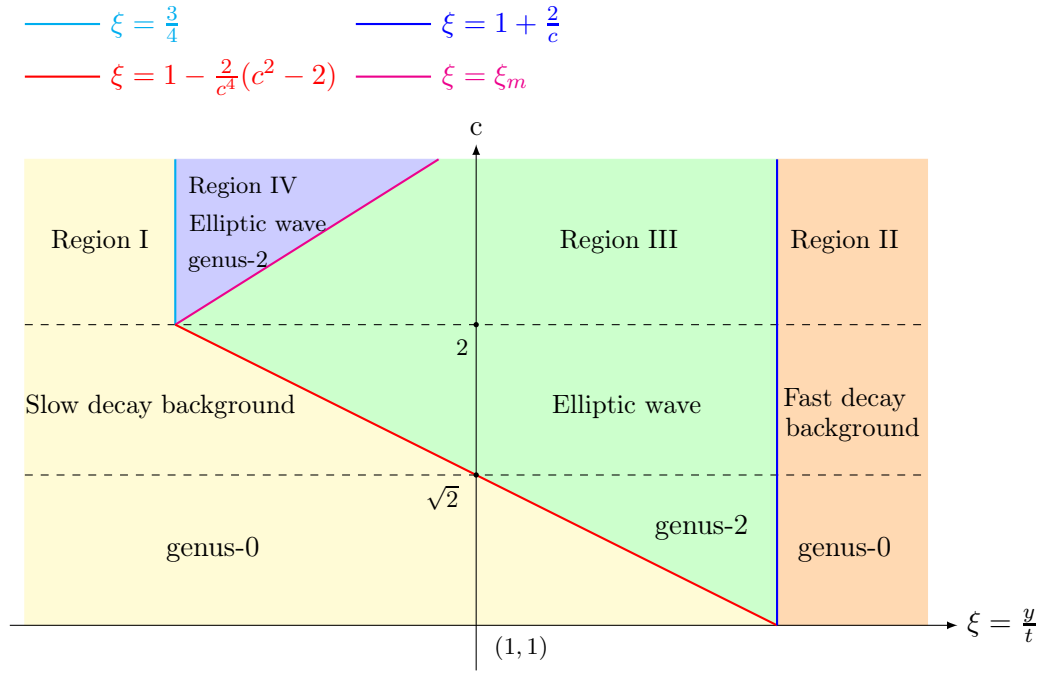
\begin{figure}
\begin{center}
\begin{tikzpicture}
\draw[yellow!20, fill=yellow!20](-4,4)--(-4,6.2)--(-6,6.2) --(-6,0)--(4,0)--(-4,4);
\draw[orange!30, fill=orange!30] (4,0)--(6,0)--(6,6.2)--(4,6.2)--(4,0);
\draw[green!20, fill=green!20] (4,0)--(4,6.2)--(-0.5,6.2)--(-4,4)--(0,2)--(4,0);
\draw[blue!20, fill=blue!20] (-4,4)--(-4,6.2)--(-0.5,6.2)--(-4,4);
\draw [ -latex ] (-6.2,0)--(6.4,0);
\draw [ dashed ] (-6,2)--(6,2);
\draw [ dashed ] (-6,4)--(6,4);
\draw [ -latex ](0,-0.6)--(0,6.4);
\draw[cyan,thick](-4,6.2)--(-4,4);
\draw[cyan,thick](-6,8)--(-5,8)node [right]{$\xi=\frac{3}{4}$};
\draw[red,thick](-6,7.3)--(-5,7.3)node [right]{$\xi=1-\frac{2}{c^4}(c^2-2)$};
\draw[magenta,thick](-4,4)--(-0.5,6.2);
\draw[magenta,thick](-1.6,7.3)--(-0.6,7.3)node [right]{$\xi=\xi_m$};
\draw[color=red,thick](-4,4)--(4,0);
\draw[blue,thick](4,0)--(4,6.2);
\draw[blue,thick](-1.6,8)--(-0.6,8)node [right]{$\xi=1+\frac{2}{c}$};
\node    at (  0.6,-0.3)  {\footnotesize $ (1,1)$};
\node    at (7,0)  {$\xi=\frac{y}{t}$};
\fill (0,2)  circle (1pt);
\fill (0,4)  circle (1pt);
 \node [right] at (-0.4, 3.7) {\footnotesize $2$};
\node    at (0,6.6 )  {c};
\node  [below]  at (-5,5.4) {\small Region  I };
\node  [below]  at (-4.2,3.2) {\small Slow decay background};
\node  [below]  at (-3.5,1.3) {genus-0};
\node  [below]  at (4.9,5.4) {\small Region II};
\node  [below]  at (4.9,3.3) {\small Fast decay };
 \node  [below]  at (5,2.9) {\small background  };
\node  [below]  at (4.9,1.3) {genus-0};
\node  [below]  at (-3.1,6.1) {\footnotesize Region IV};
\node  [below]  at (-2.9,5.6) {\footnotesize Elliptic wave};
\node  [below]  at (-3.3,5.1) {\footnotesize genus-2};
\node  [below]  at (1.9,5.4) {\small Region  III };
\node  [below]  at (2,3.2) {\small Elliptic wave  };
\node  [below]  at (3,1.6) {genus-2};
\node [left] at (-0.1,1.7)  {\footnotesize $\sqrt{2}$};
\end{tikzpicture}
\end{center}
\caption{\footnotesize  Asymptotic approximations    of the  mCH equation  in  different space-time-$(\xi,c)$   regions,
where  the  regions  I and II   corresponding to genus-0, they  are slow-decay and fast-decay background  regions, respectively;
  The  Regions III and   IV    corresponding to genus-2 region, they are the first-type  and second-type   elliptic wave regions.
  Here  $\xi_m$ is the critical condition that under the case of Region  III, the stationary phase point of $g$-function  merge $c$. }
\label{result1}
\end{figure}

Our paper is arranged as follows.  In Section \ref{sec2},
 we study   the eigenfunctions and the corresponding spectral functions associated with
 step-like initial value (\ref{initial}).   Further  we  analyze their   analyticity,  symmetries and asymptotic
to    construct   the   RH  problem  for  $M(z)$ of step-like initial value problem,  which will be used
 to analyze   long-time asymptotics  of the mCH equation in our paper.
 In Section \ref{sec3} and section \ref{sec+},    we  construct   the RH problem associated with the  Regions
    I   and  II,   further    transform   it  into  a model RH problem.
  In Sections  \ref{sec5} and \ref{sec8},  to analyze the RH problem  in the   regions III and IV,
  we introduce a $g$-function in  genus
   two Riemann surface and  transform  the original
 RH problem  to a hybrid RH problem $M^{(2)}(z)$, which is  further  decompose    into a $M^{mod}(z)$
 model problem and an inner local problems.
The  $M^{mod}(z)$  contributes to the leading term of the asymptotics and
 is given by Riemann theta functions attached to a hyperelliptic Riemann surface in  subsection \ref{secmod} and subsection \ref{secmod2} in different region.
 Finally, in Section \ref{sec9},   based on  a series of transformations  above,  a decomposition formula   for  $M(z)$
from which   we then obtain the   long-time   asymptotic behavior  for  the solutions of   the  Cauchy  problem of the mCH equation (\ref{mcho}) and  (\ref{initial})
 via a  reconstruction formula.   The main result is summarized  in the  Theorem \ref{last}.

\section {Direct scattering   and the   RH problem}\label{sec2}

\subsection{Spectral analysis on the Lax pair}

\quad The mCH equation (\ref{mcho})    admits the Lax pair \cite{Mon}
\begin{equation}
\Phi_x = X \Phi,\hspace{0.5cm}\Phi_t =T \Phi, \label{lax0}
\end{equation}
where
\begin{equation}
	X=\frac{1}{2}\left(\begin{array}{cc}
		-1 & z m \\
		-z m & 1
	\end{array}\right),\nonumber
\end{equation}
\begin{equation}
	T=\left(\begin{array}{cc}
		z^{-2}+\frac{u^2-u_x^2}{2} & -z^{-1} (u-u_x)-\frac{z}{2}(u^2-u_x^2)m \\[5pt]
		z^{-1} (u-u_x)+\frac{z}{2}(u^2-u_x^2)m & -z^{-2}-\frac{u^2-u_x^2}{2}
	\end{array}\right).\nonumber
\end{equation}
Since  the Lax pair (\ref{lax0}) admit spectral singularity at $z=\infty$ and $z=0$,
we  should  control  the  asymptotic behavior of   the  eigenfunction $\Phi$   as $z\to \infty$ and $z\to 0$    for any real constant $c$.

\noindent \textbf{Case I. $z=\infty$.}

We denote a matrix function relying on $c$
\begin{equation}
	D_c(z)=\frac{1}{2}\left(\begin{array}{cc}
		\phi_c(z)+\phi_c(z)^{-1} & \phi_c(z)^{-1}-\phi_c(z)\\
		\phi_c(z)^{-1}-\phi_c(z) & \phi_c(z)+\phi_c(z)^{-1}
	\end{array}\right),
\end{equation}
where $\phi_c(z)$ is a branch function given by
$$\phi_c(z)=\left( \frac{c+z}{c-z}\right) ^{1/4} \sim e^{\frac{3i\pi}{4}}+\mathcal{O}(z^{-1}),\hspace{0.5cm}z\to\infty. $$
 Obviously,
$$D_c(z+0i)=\left(\begin{array}{cc}
0 & -i\\
	-i & 0
\end{array}\right)D_c(z-0i), \ \ z\in [-c,c], $$
 and $\phi_c(z+0i)=i\phi_c(z-0i)$.  We  define   two gauge transformations
\begin{align}
	\Psi^\pm(z)=D_{c_\pm}(z)\Phi(z),
\end{align}
where  $c_+=c$, $c_-=1$, the  $\Psi^\pm(z)$  satisfy  the following  Lax pair
\begin{align}
	& \Psi_{x}^\pm  = -\frac{im\sqrt{z^2-c_\pm^2}}{2}\sigma_3\Psi^\pm+P^\pm\Psi^\pm,\label{lax0.1}\\
	& \Psi_{t}^\pm  =i\sqrt{z^2-c_\pm^2}\left( \dfrac{m(u^2-u_x^2)}{2}+\dfrac{1}{c_\pm z}\right) \sigma_3\Psi^\pm+L^\pm\Psi^\pm, \label{lax0.2}
\end{align}
where
\begin{align*}
	P^\pm=&i\dfrac{c_\pm m-1}{2\sqrt{z^2-c_\pm^2}}\left(\begin{array}{cc}
		c_\pm & z \\
		-z & -c_\pm
	\end{array}\right), \nonumber\\
	L^\pm=&i\left( \dfrac{c_\pm(u^2-u_x^2)(1-c_\pm m)}{2\sqrt{z^2-c_\pm^2}}-\dfrac{u-1/c_\pm}{\sqrt{z^2-c_\pm^2}}\right) \sigma_3+\dfrac{u_x}{c_\pm}\sigma_1\\
	&+i\left( \dfrac{z(u^2-u_x^2)(1-c_\pm m)}{2\sqrt{z^2-c_\pm^2}}-\dfrac{c_\pm u-1}{z\sqrt{z^2-c_\pm^2}}\right) \left(\begin{array}{cc}
		0 & 1 \\
		-1 & 0
	\end{array}\right).
\end{align*}
Further  we introduce a transformation
\begin{align}
	\mu^\pm(z)=\Psi^\pm(z)e^{itp_\pm(z)\sigma_3},\label{transmu}
\end{align}
where
\begin{align}
	tp_\pm(z)&=\dfrac{\sqrt{z^2-c_\pm^2}}{2}\left( \int_{+\infty}^{x}(m(s)-1/c_\pm)ds+\frac{x}{c_\pm}-\frac{2t}{c_\pm z^2}+\frac{t}{c_\pm^3} \right).
\end{align}
Then  $\mu^\pm(z)$ admit asymptotics
\begin{align}
	\mu^\pm (z) \sim I, \hspace{0.5cm} x \rightarrow \pm\infty,\nonumber
\end{align}
and satisfy  a new Lax pair
\begin{align}
	& \mu^{\pm }_x  = -i\partial_x(tp_\pm)[\sigma_3,\mu^\pm]+P^\pm\mu^\pm,\label{lax1.1}\\
	&\mu^{\pm }_t  =-i\partial_t(tp_\pm)[\sigma_3,\mu^\pm]+L^\pm\mu^\pm. \label{lax1.2}
\end{align}
 The above  Lax pair   can be written as a  full  derivative form,  which is  further
 integrated along $(\pm \infty, t ) \rightarrow (x,t)$ and  leads to   two  Volterra type integrals
\begin{equation}
	\mu^\pm(x,t;z)=I+\int^{x}_{\pm \infty}e^{\frac{i}{2}\sigma_3\sqrt{z^2-c_\pm^2}\int^{s}_{x}m(v)dv}\left[ P^\pm\mu^\pm (s,t;z) \right] ds\label{intmu}.
\end{equation}
Denote
$$\Sigma_-=[-1,1],\hspace{0.5cm}\Sigma_+=[-c,c], $$
then we can show  that 	$\mu^\pm(x,t;z)$ is analytical in $\mathbb{C}\setminus \Sigma_\pm$ respectively.

\begin{Proposition}\label{sym}
	The Jost functions $ \mu^\pm (z)$ admit two kinds of symmetries
\begin{equation}
\mu^\pm(z)=\sigma_1\overline{\mu^\pm(\bar{z})}\sigma_1=\sigma_2\overline{\mu^\pm(-z)}\sigma_2.\label{symPhi1}
\end{equation}
\end{Proposition}

Since   $D_{c_\pm}(z)^{-1}\Psi^\pm(z;x,t)$ are two fundamental matrix solutions of the  Lax  pair (\ref{lax0}), they satisfy  a linear  relation
\begin{equation}
	D_{c_+}(z)^{-1}\Psi^+(z;x,t)=D_{c_-}(z)^{-1}\Psi^-(z;x,t)S(z), \label{scattering}
\end{equation}
where $S(z)$ is  a  scattering matrix
\begin{equation}
		S(z) =\left(\begin{array}{cc}
			s_{11}(z) &s_{12}(z)   \\[4pt]
			 s_{21}(z) & s_{22}(z)
		\end{array}\right),\hspace{0.5cm}\det S(z) =1. \nonumber
\end{equation}
Combing the transformation (\ref{transmu})  with  the  equation (\ref{scattering}) gives
\begin{align}
	S(z)=e^{itp_-\sigma_3}(\mu^-)^{-1}D_{c_-}D_{c_+}^{-1}\mu^+e^{-itp_+\sigma_3}, \label{scattering23}
\end{align}
which is analytical  on $\mathbb{C}\setminus(\Sigma_+\cup\Sigma_-)$.
Defining  two reflection coefficients   by
\begin{equation}
	r_1(z)=\frac{s_{21}(z)}{s_{11}(z)},\hspace{0.5cm}r_2(z)=\frac{s_{12}(z)}{s_{22}(z)}.\label{symr}
\end{equation}
Let
\begin{align}
	\tilde{\mu}^\pm(z) =D_{c_\pm}^{-1}\mu^\pm (z),\label{transhmu}
\end{align}
then the  Volterra type integrals  (\ref{intmu}) are changed into
\begin{align}
	&\tilde{\mu}^\pm(z) =D_{c_\pm}^{-1} \nonumber\\
	&+\int^{x}_{\pm \infty}F (z) \left(X+ \frac{mi\sqrt{z^2-c_\pm^2}}{2}D_{c_\pm}^{-1}\sigma_3D_{c_\pm}\right) \tilde{\mu}^\pm e^{-\frac{i}{2}\sigma_3\sqrt{z^2-c_\pm^2}\int^{s}_{x}m(l)dl}ds,\label{inthmu}
\end{align}
where
\begin{align*}
	F(z) =D_{c_\pm}^{-1}e^{\frac{i}{2}\sqrt{z^2-c_\pm^2}\int^{s}_{x}m(l)dl\sigma_3}D_{c_\pm}
\end{align*}
is a analytical  function on $\mathbb{C}$.   Thus  we give the following proposition
\begin{Proposition} The Jost functions
	$\tilde{\mu}^\pm(z)$ and the  scattering matrix   $S(z)$  have   $-\frac{1}{4}$-weak singularity  at $z=\pm1 $ and $z=\pm c  $.
\end{Proposition}
\begin{corollary}
    We  have asymptotics
$$1-r_1(z)r_2(z)=\mathcal{O}(z\mp c)^{1/2}, \ \ z\to\pm c.$$
\end{corollary}
To construct the RH problem,   we need to consider the   jump  of the Jost functions
	$\tilde{\mu}^\pm(z)$ on the cut $\Sigma_\pm$.
\begin{Proposition}  The Jost functions $\mu^\pm(z)$  admit the jump relations

(i) \ For $z\in\Sigma_\pm$,
\begin{align}
	\mu^\pm(z+0i)=\sigma_1	\mu^\pm(z-0i)\sigma_1.
\end{align}

(ii)\ Especially   for $z\in[-1,1]$,
\begin{align}
	& \tilde{\mu}^\pm _{11}(z+0i)=i \tilde{\mu}^\pm_{12}(z-0i),\hspace{0.3cm} \tilde{\mu}^\pm_{21}(z+0i)=i \tilde{\mu}^\pm_{22}(z-0i),\nonumber\\
	&s_{11}(z+0i)=s_{22}(z-0i),\hspace{0.3cm}s_{12}(z+0i)=s_{21}(z-0i).\nonumber
\end{align}

(iii)\ For $z\in[-c,c]\setminus[-1,1]$,  $\tilde{\mu}^+(z)$ has same jump as above equation  while  $\tilde{\mu}^-(z)$ has no jump. And
\begin{align}
& s_{11}(z+0i)=is_{12}(z-0i),\hspace{0.3cm}s_{21}(z+0i)=is_{22}(z-0i).
\end{align}
\end{Proposition}

From (\ref{scattering23}),  we have
\begin{align}
	S(z)\sim e^{\frac{1}{2}Hz\sigma_3 },\ \ z\to\infty,
\end{align}
where $H$ is conserved quantity with
$$H=(1-1/c)x+(1/c^3-1)t+\int_{-\infty}^{x}(m-1)ds+\int_{x}^{+\infty}(m-1/c)ds.$$

The zeros of $s_{11}(z)$ and $s_{22}(z)$ on $\mathbb{C}$   are known to
occur and they correspond to spectral singularities.  They are excluded from our analysis in the this paper.
 Thus, $r_1(z)$ and $r_2(z)$ are analytic in $\mathbb{C}\setminus(\Sigma_-\cup\Sigma_+)$.

\noindent \textbf{Case II: $z=0$}.

We rewrite  the Lax pair (\ref{lax0.1})-(\ref{lax0.2}) in the form
\begin{align}
	& \Psi^{\pm}_x = -\frac{i\sqrt{z^2-c_\pm^2}}{2c_\pm}\sigma_3\Psi^\pm+P_0^\pm\Psi^\pm,\\
	& \Psi^{\pm}_t  =i\sqrt{z^2-c_\pm^2}\left( \dfrac{1}{2c_\pm ^3}+\dfrac{1}{c_\pm z^2}\right) \sigma_3\Psi^\pm+L_0^\pm\Psi^\pm,
\end{align}
where
\begin{align*}
	P_0^\pm=&iz\dfrac{c_\pm m-1}{2c_\pm\sqrt{z^2-c_\pm^2}}\left(\begin{array}{cc}
		z & c_\pm \\
		-c_\pm & -z
	\end{array}\right), \nonumber\\
	L_0^\pm=&L_\pm+i\sqrt{z^2-c_\pm^2}\left( \dfrac{m(u^2-u_x^2)}{2}+\dfrac{1}{c_\pm z}-\dfrac{1}{2c_\pm ^3}+\dfrac{1}{c_\pm z^2}\right) \sigma_3.
\end{align*}
Making transformation
\begin{align}
\mu^\pm_0(z)=\Psi_\pm e^{iq^\pm\sigma_3},  \hspace{0.3cm}\ \ 	q^\pm=\frac{i\sqrt{z^2-c_\pm^2}}{2c_\pm}\left[x-\left(\dfrac{1}{c_\pm ^2}+\dfrac{2}{ z^2} \right)t  \right],
\end{align}
then  $\mu^\pm_0(z)$ admit a new Lax pair
\begin{align}
	& \mu^{\pm}_{0,x}(z)  = -\frac{i\sqrt{z^2-c_\pm^2}}{2c_\pm}[\sigma_3,\mu^\pm_0(z)]+P_0^\pm\mu^\pm_0(z),\\
	& \mu^{\pm}_{0,t}(z) =i\sqrt{z^2-c_\pm^2}\left( \dfrac{1}{2c_\pm ^3}+\dfrac{1}{c_\pm z^2}\right) [\sigma_3,\mu^\pm_0(z)]+L_0^\pm\mu^\pm_0(z).
\end{align}
It also can be written in to  two  Volterra type integrals
\begin{equation}
	\mu^\pm_0(z)=I+\int^{x}_{\pm \infty}e^{\frac{i}{2c_\pm}\sigma_3\sqrt{z^2-c_\pm^2}(s-x)}\left[ P_0^\pm\mu^\pm_0(s;z) \right] ds.
\end{equation}
 To reconstruct the potential $u(x,t)$,  we take $z=0$,  then
\begin{align}
P^0_\pm(0)=0, \ \  \ \ 	\mu^\pm_0(0)=I.\label{asymu0}
\end{align}
Expanding $\mu^\pm_0(z)$ at $z=0$ gives
 \begin{align}
 \mu^\pm_0(z)=I+\frac{z}{2}\left(\begin{array}{cc}
 		0 & -\int_{\pm\infty}^x(m-\frac{1}{c_\pm})e^{s-x}ds \\
 		\int_{\pm\infty}^x(m-\frac{1}{c_\pm})e^{x-s}ds & 0
 	\end{array}\right)+\mathcal{O}(z^2). \nonumber
 \end{align}
Because $\Psi^\pm(z)$ admit same Lax pair (\ref{lax0.1})-(\ref{lax0.2}), there exist two matrix function $C_\pm(z)$
independent of $x$ and $t$
\begin{align}
	\mu^\pm_0(z)e^{-iq^\pm\sigma_3}C_\pm(z)=\mu^\pm(z)e^{-itp^\pm\sigma_3}.
\end{align}
Since
$$q_\pm-tp_\pm=-\dfrac{1}{2}\sqrt{z^2-c_\pm^2}\int_{\pm\infty}^x(m-1/c_\pm)ds,$$
taking  the limits $x\to\pm\infty$, we obtain $C_\pm(z)\equiv I$. Then
\begin{align}
	\mu^\pm(0+0i)=\exp \left\lbrace i(q_\pm(0+0i)-tp_\pm (0+0i))\sigma_3\right\rbrace,
\end{align}
which combines  with $D_{c_\pm}(0+0i)=-i\sigma_1$ implies that
\begin{align}
	\tilde{\mu}^\pm(0+0i)=\left(\begin{array}{cc}
		0 & ie^{i(q_\pm(0+0i)-tp_\pm (0+0i))} \\
		ie^{-i(q_\pm(0+0i)-tp_\pm (0+0i))} & 0
	\end{array}\right).\label{tmu0}
\end{align}
Consequently,
\begin{align}
	S(0+0i)=\exp \left\lbrace i(q_-(0+0i)-q_+ (0+0i))\sigma_3\right\rbrace.\label{s0}
\end{align}

\subsection{Setting up a   RH problem  with step-like initial data  }

\quad
Define  a   sectionally analytical  matrix
\begin{equation}
	M(z)\triangleq M(z;x,t)=\left\{ \begin{array}{ll}
		\left( \frac{ \tilde{\mu}^-_1  (z) } {s_{11}(z)}e^{-it(p_+-p_-)},  \tilde{\mu}^+_2(z)\right),   &\text{as } z\in \mathbb{C}^+,\\[12pt]
		\left( \tilde{\mu}^+_1(z),\frac{ \tilde{\mu}^-_2(z)}{s_{22}(z)}e^{it(p_+-p_-)}\right)  , &\text{as }z\in \mathbb{C}^-,\\
	\end{array}\right.
\end{equation}
where $\tilde{\mu}^\pm_1 (z)$ and  $\tilde{\mu}^\pm_2 (z)$ denote  the first and second column of $\tilde{\mu}^\pm(z)$, respectively. Then
$M(z)$ solves the following RH problem.
\begin{RHP}\label{RHP1}
	Find a matrix-valued function $M(z)=M(z;x,t)$ which satisfies

$\blacktriangleright$ Analyticity: $M(z)$ is analytical  in $\mathbb{C}\setminus \mathbb{R}$;

$\blacktriangleright$ Symmetry: $M(z)=\sigma_2M(-z)\sigma_2$=$\sigma_1\overline{M(\bar{z})}\sigma_1$;

$\blacktriangleright$ Jump condition: $M$ has continuous boundary values $M_\pm(z)$ on $\mathbb{R}$ and
\begin{equation}
	M_+(z)=M_-(z)\tilde{V}(z),\hspace{0.5cm}z \in \mathbb{R},
\end{equation}
where
\begin{equation}
	\tilde{V}(z)=\left\{ \begin{array}{ll}
		\left(\begin{array}{cc}
			1-r_1r_2 & -r_2e^{-2itp_-}\\
			r_1e^{2itp_-} & 1
		\end{array}\right),   &\text{as } z\in\mathbb{R}\setminus[-c,c],\\[12pt]
		\left(\begin{array}{cc}
			0 & -r_2(z-0i)e^{-2itp_-}\\
			r_1(z+0i)e^{2itp_-} & 1
		\end{array}\right),   &\text{as } z\in[-c,c]\setminus[-1,1] ,\\[12pt]
		\left(\begin{array}{cc}
			0 & i\\
			i & 0
		\end{array}\right),   &\text{as } z\in[-1,1],
	\end{array}\right. \nonumber
\end{equation}

$\blacktriangleright$ Asymptotic behaviors
\begin{align}
	&M(z) = I+\mathcal{O}(z^{-1}),\hspace{0.5cm}z \rightarrow \infty;
\end{align}

$\blacktriangleright$ Singularity: $M(z)$ has singularity at $z=\pm1$ with
\begin{align}
	&M(z)\sim \left(\mathcal{O}(z\mp1)^{1/4},\mathcal{O}(z\mp1)^{-1/4} \right) ,\ z\to \pm 1\ \text{ in }\ \mathbb{C}^+,\\
	&M(z)\sim \left(\mathcal{O}(z\mp1)^{-1/4},\mathcal{O}(z\mp1)^{1/4} \right) ,\ z\to \pm 1 \ \text{ in }\ \mathbb{C}^-.
\end{align}
\end{RHP}

The solution of  mCH equation (\ref{mcho}) is difficult to reconstruct from the above RHP1,   since  $p_-(x,t,z)$ is still unknown.
 Boutet de Monvel and Shepelsky  proposed  an  idea  to  change  the spatial variable $x$ to variable $y$  \cite{CH,CH1}.
Following this idea, we introduce a  new   scale
\begin{equation}
	y(x,t)=x-\int_{x}^{-\infty} \left(m(s)-1\right) ds.
\end{equation}
The price to pay for this is that the solution of the initial problem can be given only implicitly
or perimetrically.  It will be given in terms of functions in the new scale, whereas the original scale will also be given in terms of functions in the new scale.
By the definition of the new scale $y(x, t)$, we define
\begin{equation}
	N(z)=N(z;y,t)\triangleq M(z;x(y,t),t).
\end{equation}
For convenience we denote $\xi=y/t$ with $$p_-=\frac{\sqrt{z^2-1}}{2}\left(\xi-1-2z^{-2} \right), $$ then  we can get the RH problem for the new variable $(y,t)$.
\begin{RHP}\label{RHP2}
	 Find a matrix-valued function $	N(z)\triangleq N(z;y,t)$ which satisfies

$\blacktriangleright$ Analyticity: $N(z)$ is meromorphic in $\mathbb{C}\setminus \mathbb{R}$;

$\blacktriangleright$ Symmetry: $N(z)=\sigma_2N(-z)\sigma_2$=$\sigma_1\overline{N(\bar{z})}\sigma_1$;

$\blacktriangleright$ Jump condition: $N$ has continuous boundary values $N_\pm(z)$ on $\mathbb{R}$ and
\begin{equation}
	N_+(z)=N_-(z)\tilde{V}(z),\hspace{0.5cm}z \in \mathbb{R},
\end{equation}
where
\begin{equation}
	\tilde{V}(z)=\left\{ \begin{array}{ll}
		\left(\begin{array}{cc}
			1-r_1r_2 & -r_2e^{-2itp_-}\\
			r_1e^{2itp_-} & 1
		\end{array}\right),   &\text{as } z\in\mathbb{R}\setminus[-c,c],\\[12pt]
		\left(\begin{array}{cc}
			0 & -r_2(z-0i)e^{-2itp_-}\\
			r_1(z+0i)e^{2itp_-} & 1
		\end{array}\right),   &\text{as } z\in[-c,c]\setminus[-1,1],\\[12pt]
		\left(\begin{array}{cc}
			0 & i\\
			i & 0
		\end{array}\right),   &\text{as } z\in[-1,1].
	\end{array}\right.\nonumber
\end{equation}

$\blacktriangleright$ Asymptotic behaviors: $N(z) = I+\mathcal{O}(z^{-1}),\hspace{0.5cm}z \rightarrow \infty;$

$\blacktriangleright$ Singularity: $N(z)$ has singularity at $z=\pm1$ with
\begin{align}
	&N(z)\sim \left(\mathcal{O}(z\mp1)^{1/4},\mathcal{O}(z\mp1)^{-1/4} \right) ,\ z\to \pm1\text{ in }\mathbb{C}^+,\\
	&N(z)\sim \left(\mathcal{O}(z\mp1)^{-1/4},\mathcal{O}(z\mp1)^{1/4} \right) ,\ z\to \pm1\text{ in }\mathbb{C}^-.
\end{align}
\end{RHP}
From the asymptotic behavior of the functions $\tilde{\mu}^\pm(z)$ and (\ref{s0}), we have
\begin{align}
	N(z)=&	\left(\begin{array}{cc}
		0 & if_1 \\
		if_1^{-1} & 0
	\end{array}\right)+iz	\left(\begin{array}{cc}
f_2 & 0\\
0  & f_3
\end{array}\right)+\mathcal{O}(z^2), \ \ z\to 0\in\mathbb{C}^+,\label{N0}
\end{align}
where
\begin{align}
	f_1&=\exp\left\lbrace -\frac{1}{2}\int_{-\infty}^x(m-1)ds\right\rbrace,\nonumber\\
	f_2&=\frac{e^{\frac{i}{2}\int_{-\infty}^x(m-1)ds}}{2c}\left(\int_{+\infty}^x(cm-1)e^{x-s}ds +e^{\frac{i}{2}\int_{+\infty}^x(cm-1)ds}\right),\nonumber\\
	f_3&=\frac{e^{-\frac{i}{2}\int_{-\infty}^x(m-1)ds}}{2}\left(- \int_{-\infty}^x(m-1)e^{x-s}ds +e^{-\frac{i}{2}\int_{-\infty}^x(m-1)ds}\right).\nonumber
\end{align}
Thus,
\begin{equation}
	x(y,t)=y-2\ln\left( -iM_{21}(0+0i)\right).\label{recons x}
\end{equation}
Combining with (\ref{mcho}), we arrive at following reconstruction formula
\begin{align}
&u(x,t)= f_1(x,t)f_2(x,t)+f_1(x,t)^{-1}f_3(x,t),\\
&\partial_xu(x,t)= -f_1(x,t)f_2(x,t)+f_1(x,t)^{-1}f_3(x,t),
\end{align}
 namely,
\begin{align}
	&u(x,t)= -N_{12}(0+0i)\lim_{z\to 0 }\frac{N_{11}(z)-N_{11}(0)}{z}\nonumber\\
	&\qquad\qquad -N_{12}(0+0i)^{-1}\lim_{z\to 0}\frac{N_{22}(z)-N_{22}(0)}{z}, \ \ z\in\mathbb{C}^+, \label{recons u}\\
& \partial_xu(x,t)= N_{12}(0+0i)\lim_{z\to 0 }\frac{N_{11}(z)-N_{11}(0)}{z}\nonumber\\
	&\qquad\qquad -N_{12}(0+0i)^{-1}\lim_{z\to 0 }\frac{N_{22}(z)-N_{22}(0)}{z}, \ \ z\in\mathbb{C}^+.
\end{align}
Moreover, the jump matrix $\tilde{V}(z)$  admits the  following  decomposition \\
On the interval  $ \mathbb{R}\setminus[-c,c]$,
\begin{align}
	\tilde{V}(z)=&\left(\begin{array}{cc}
	1 & -r_2e^{-2itp_-}\\
	0 & 1
\end{array}\right)\left(\begin{array}{cc}
1 & 0\\
r_1e^{2itp_-} & 1
\end{array}\right)\nonumber\\
=&\left(\begin{array}{cc}
	1 &0\\
	\frac{r_1e^{2itp_-}}{1-r_1r_2} & 1
\end{array}\right)(1-r_1r_2)^{\sigma_3}\left(\begin{array}{cc}
1 &\frac{ -r_2e^{-2itp_-}}{1-r_1r_2}\\
0 & 1
\end{array}\right).
\end{align}
On the interval   $ [-c,c]\setminus[-1,1]$,
\begin{align}
	\tilde{V}(z)
	=&\left(\begin{array}{cc}
		1 & -r_2(z-0i)e^{-2itp_-}\\
		0 & 1
	\end{array}\right)\left(\begin{array}{cc}
		1 & 0\\
		r_1(z+0i)e^{2itp_-} & 1
	\end{array}\right)\nonumber\\
	=&\left(\begin{array}{cc}
		1 &0\\
		\frac{r_1(z-0i)e^{2itp_-}}{1-r_1(z-0i)r_2(z-0i)} & 1
	\end{array}\right)\left(\begin{array}{cc}
	0 & \frac{-r_2(z-0i)}{e^{2itp_-}}\\
	\frac{e^{2itp_-}}{r_2(z-0i)} & 0
\end{array}\right)\left(\begin{array}{cc}
		1 &\frac{ -r_2(z+0i)e^{-2itp_-}}{1-r_1(z+0i)r_2(z+0i)}\\
		0 & 1
	\end{array}\right).\nonumber
\end{align}
The long-time asymptotic  of RHP\ref{RHP2}  is affected by the growth or decay of the exponential function $e^{\pm2itp_-}$   with
$$p_-=\frac{\sqrt{z^2-1}}{2}\left(\xi-1-2z^{-2} \right),\ \text{d}p_-=\frac{1}{2\sqrt{z^2-1}z^3}\left[ \left(  \xi-1)z^4+2z^2-4 \right) \right]. $$
So   to analyze the long-time asymptotics, we need control the real part of $ 2itp_-$. But  we find that in some cases,  this exponential function is not applicable. Its adaptation to problems with nonzero background has required the development of the g-function mechanism \cite{gfunction}. This mechanism is relevant when some entries of
the jump matrix grow exponentially or oscillate as $t\to\infty$. The general idea consists in replacing the original phase function in the jump matrix. This new function $g(\xi , z)$ need to be analytic on $\mathbb{C}$ except cut, and satisfies that,  after appropriately choosing triangular factorizations of the jump matrices and associated
deformations of the original RH problem, the jumps containing  exponentially growing entries, become   constant matrices (independent
of $z$, but dependent on $\xi$) of special structure whereas the other jumps decay exponentially to the identity matrix. And it must have same asymptotic properties as $z\to\infty,\ 0\in\mathbb{C}^+$ as $p_-$:
\begin{align}
	&p_-=\frac{1-\xi}{2}z+\mathcal{O}(z^{-1}),\ \frac{\text{d}}{\text{d}z}p_-=\frac{1-\xi	}{2}+\mathcal{O}(z^{-2}),\ z\to\infty;\\
	&p_-=\frac{i}{z^2}-\frac{i\xi}{2}+\mathcal{O}(z),\ \frac{\text{d}}{\text{d}z}p_-=\frac{-2i}{z^3}+\mathcal{O}(z),\ z\to0\in\mathbb{C}^+.
\end{align}
The structure of the limiting RH problem
is such that the problem can be solved explicitly in terms of Riemann theta functions
and Abel integrals on Riemann surfaces associated with the limiting RH problem. For
different ranges of the parameter $\xi= y/t$ , different Riemann surfaces
 may appear \cite{Biondini1,Biondini2,Monvel2011,MonvelCMP1,MonvelCMP2,Buck2007}. To find the g-function, we first consider the signature table
and stationary phase points of $p_-$ in Figure \ref{figp-}.

\begin{itemize}
\item[ (a):]  \  For the case  $\xi<\frac{3}{4}$,  there is  no  stationary point on $\mathbb{R}$;

\item[   (b):]\  For the case  $3/4<\xi<1$    there are four stationary points on $\mathbb{R}$
\begin{align}
	\pm\lambda_1=\pm \left(\dfrac{1-\sqrt{4\xi-3}}{1-\xi}\right)  ^{1/2},\  \pm\lambda_2=\pm \left(\dfrac{1+\sqrt{4\xi-3}}{1-\xi}\right)   ^{1/2}.
\end{align}
\item[   (c):]\  For the case  $1<\xi<3$,there are two stationary points on $\mathbb{R}$
\begin{align}
	\pm\lambda_1=\pm \left(\dfrac{ \sqrt{4\xi-3}-1}{\xi-1}\right)  ^{1/2}.
\end{align}
\item[  (d):]\  For the  case $3<\xi$, there are two  stationary points on $\mathbb{R}$
\begin{align}
	\pm\lambda_1=\pm \sqrt{  2} (\xi-1)^{-1/2}.
\end{align}

\end{itemize}

\begin{figure}[h]
	\centering
	\subfigure[]{
		\begin{tikzpicture}[node distance=2cm]
			\draw[lime!50, fill=lime!30] (-2.2,0)--(-2.2,1.5)--(2.2,1.5)--(2.2,0)--(0,0)arc (-90:270:0.4 and 0.7)--(-2.2,0);
			\draw[->](1,0)--(2.2,0)node[right]{ Re$z$};
			\draw[lime!50, fill=lime!30] (0,0)arc (90:450:0.4 and 0.7);
			\draw [](-0,-0.71) ellipse (0.4 and 0.7);
			\draw [](0,0.71) ellipse (0.4 and 0.7);		
			\coordinate (I) at (0,0);
			\fill (I) circle (0pt) node[below] {$0$};
			\coordinate (b) at (1,0);
			\fill (b) circle (1pt) node[below] {$1$};
			\coordinate (ba) at (-1,0);
			\fill (ba) circle (1pt) node[below] {$-1$};
			\coordinate (c) at (1.5,0);
			\fill (c) circle (1pt) node[below] {$c$};
			\coordinate (ca) at (-1.5,0);
			\fill (ca) circle (1pt) node[below] {$-c$};
			\draw[dashed](-1,0)--(1,0);
			\draw[](-2.2,0)--(-1,0);
			\draw[dashed](0,-1.5)--(0,1.5)node[above]{ Im$z$};
		\end{tikzpicture}
	}
	\subfigure[]{
		\begin{tikzpicture}[node distance=2cm]
			\draw[lime!50, fill=lime!30] (-1.4,0)--(-2.7,0)arc (-180:0:2.7 and 1.2)--(2.7,0)--(2,0)arc (0:-180:1 and 0.45)--(0,0)--(0,0)arc (0:-180:1 and 0.45)--(-2,0);
			\draw[lime!50, fill=lime!30] (-2.7,0)--(-2.7,0)arc (180:0:2.7 and 1.2)--(2.7,0)--(3.2,0)--(3.2,1.5)--(-3.2,1.5)--(-3.2,0)--(-2.7,0);
			\draw [](0,0) ellipse (2.7 and 1.2);
			\draw[dashed](0,-1.5)--(0,1.5)node[above]{ Im$z$};
			\draw[lime!50, fill=lime!30] (0,0)--(0,0)arc (0:180:1 and 0.45)--(0,0);
			\draw [](-1.01,0) ellipse (1 and 0.45);
			\draw[lime!50, fill=lime!30] (0,0)--(0,0)arc (180:0:1 and 0.45)--(0,0);
			\draw [](1.01,0) ellipse (1 and 0.45);		
			\coordinate (I) at (0,0);
			\fill (I) circle (0pt) node[below] {$0$};
			\coordinate (a) at (2,0);
			\fill (a) circle (1pt) node[below] {\footnotesize$\lambda_1$};
			\coordinate (aa) at (-2,0);
			\fill (aa) circle (1pt) node[below] {\footnotesize$-\lambda_1$};
			\coordinate (xa) at (2.7,0);
			\fill (xa) circle (1pt) node[below] {\footnotesize$\lambda_2$};
			\coordinate (xaa) at (-2.7,0);
			\fill (xaa) circle (1pt) node[below] {\footnotesize$-\lambda_2$};
			\coordinate (b) at (0.8,0);
			\fill (b) circle (1pt) node[below] {\footnotesize$1$};
			\coordinate (ba) at (-0.8,0);
			\fill (ba) circle (1pt) node[below] {\footnotesize$-1$};
			\coordinate (c) at (1.35,0);
			\fill (c) circle (1pt) node[below] {\footnotesize$\sqrt{2}$};
			\coordinate (ca) at (-1.35,0);
			\fill (ca) circle (1pt) node[below] {\footnotesize$-\sqrt{2}$};
			\draw[dashed](-2.3,0)--(2.3,0);
			\draw[](-3.3,0)--(-0.8,0);			
			\draw[->](0.8,0)--(3.3,0)node[right]{ Re$z$};
		\end{tikzpicture}
	}
	\subfigure[]{
		\begin{tikzpicture}[node distance=2cm]
			\draw[lime!50, fill=lime!30] (-2.6,0)--(-2.6,-1.5)--(2.6,-1.5)--(2.6,0)--(0.8,0)--(1.6,0)arc (0:-180:0.8 and 0.5)--(0,0)--(0,0)arc (0:-180:0.8 and 0.5)--(-2.6,0);
			\draw[dashed](0,-1.5)--(0,1.5)node[above]{ Im$z$};
			\draw[lime!50, fill=lime!30] (0,0)--(0,0)arc (0:180:0.8 and 0.5)--(0,0);
			\draw [](-0.81,0) ellipse (0.8 and 0.5);
			\draw[lime!50, fill=lime!30] (0,0)--(0,0)arc (180:0:0.8 and 0.5)--(0,0);
			\draw [](0.81,0) ellipse (0.8 and 0.5);		
			\coordinate (I) at (0,0);
			\fill (I) circle (0pt) node[below] {$0$};
			\coordinate (a) at (1.6,0);
			\fill (a) circle (1pt) node[below] {\footnotesize$\lambda_1$};
			\coordinate (aa) at (-1.6,0);
			\fill (aa) circle (1pt) node[below] {\footnotesize$-\lambda_1$};
			\coordinate (b) at (1,0);
			\fill (b) circle (1pt) node[below] {\footnotesize$1$};
			\coordinate (ba) at (-1,0);
			\fill (ba) circle (1pt) node[below] {\footnotesize$-1$};
			\coordinate (c) at (2.3,0);
			\fill (c) circle (1pt) node[below] {\footnotesize$\sqrt{2}$};
			\coordinate (ca) at (-2.3,0);
			\fill (ca) circle (1pt) node[below] {\footnotesize$-\sqrt{2}$};
			\draw[dashed](-2.3,0)--(2.3,0);
			\draw[](-2.6,0)--(-1,0);
				\draw[->](1,0)--(2.6,0)node[right]{ Re$z$};
		\end{tikzpicture}
	}
	\subfigure[]{
		\begin{tikzpicture}[node distance=2cm]
			\draw[lime!50, fill=lime!30] (-3,0)--(-3,-1.5)--(3,-1.5)--(3,0)--(0.8,0)--(1.6,0)arc (0:-180:0.8 and 0.5)--(0,0)--(0,0)arc (0:-180:0.8 and 0.5)--(-3,0);
			\draw[->](2.1,0)--(3,0)node[right]{ Re$z$};
			\draw[dashed](0,-1.5)--(0,1.5)node[above]{ Im$z$};
			\draw[lime!50, fill=lime!30] (0,0)--(0,0)arc (0:180:0.8 and 0.5)--(0,0);
			\draw [](-0.81,0) ellipse (0.8 and 0.5);
			\draw[lime!50, fill=lime!30] (0,0)--(0,0)arc (180:0:0.8 and 0.5)--(0,0);
			\draw [](0.81,0) ellipse (0.8 and 0.5);		
			\coordinate (I) at (0,0);
			\fill (I) circle (0pt) node[below] {$0$};
			\coordinate (a) at (1.6,0);
			\fill (a) circle (1pt) node[below] {\footnotesize$\lambda_1$};
			\coordinate (aa) at (-1.6,0);
			\fill (aa) circle (1pt) node[below] {\footnotesize$-\lambda_1$};
			\coordinate (b) at (2.1,0);
			\fill (b) circle (1pt) node[below] {\footnotesize$1$};
			\coordinate (ba) at (-2.1,0);
			\fill (ba) circle (1pt) node[below] {\footnotesize$-1$};
			\coordinate (c) at (2.6,0);
			\fill (c) circle (1pt) node[below] {\footnotesize$c$};
			\coordinate (ca) at (-2.6,0);
			\fill (ca) circle (1pt) node[below] {\footnotesize$-c$};
			\draw[dashed](-2.6,0)--(2.6,0);
			\draw[](-3,0)--(-2.1,0);
		\end{tikzpicture}
	}
	\caption{\footnotesize  In the white region, Im$p_-<0$, while in another region, Im$p_-<0$.  }
	\label{figp-}
\end{figure}
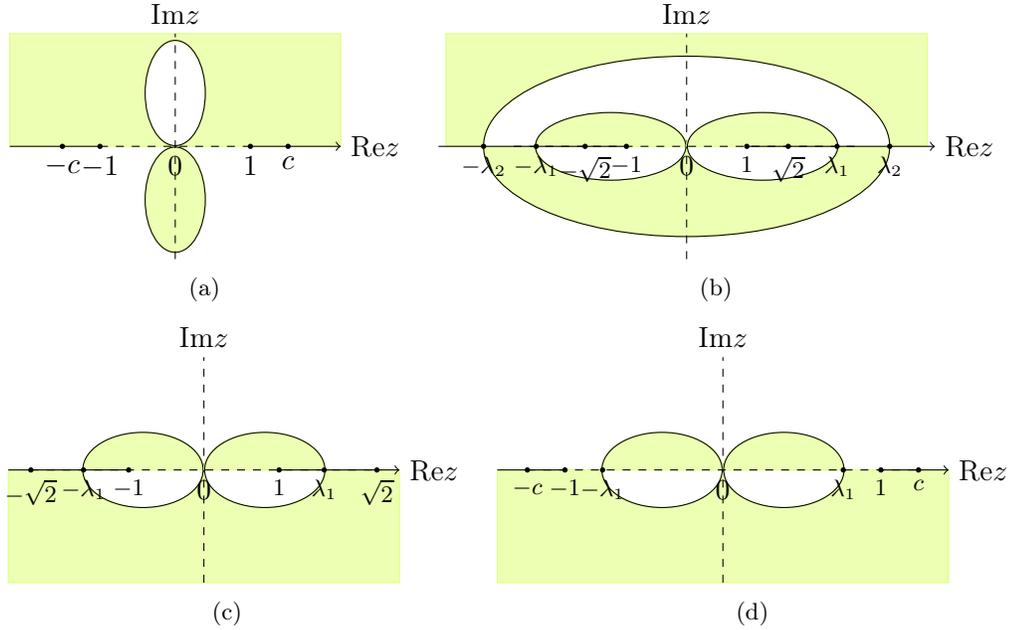
These figures mean that in the case of subfigure (a), we can deal with the jump reserve $p_-$. But in the other region of $\xi$, the applicability of $p_-$ need to discussed. In the following section, we will study different g-functions respectively.

\section{ Slow-decay   background  region    }\label{sec3}
\quad In this section, we will discuss the applicability of $p_-$ in different region of $\xi$.  The Region  I contains the following  three different cases:
$$(i)  \ \xi<\frac{3}{4}; \  \ \ (ii)\  1< c\leq \lambda_1, \  \frac{3}{4}<\xi<1; \ \ (iii)\ 1< c\leq \lambda_1, \  1<\xi,$$
 where $\lambda_1=\sqrt{\frac{1}{1-\xi}-\sqrt{(\frac{1}{1-\xi}-4)\frac{1}{1-\xi}}}$.
As shown in Figure \ref{figp-},
in the  case (i),  the jump on $\mathbb{R}\setminus[-1,1]$ is easy to deal with.
Specially in the case (ii)  and case (iii),  we have that $c<\lambda_1$.

\textbf{(i) } \  The case $\xi<\frac{3}{4}$.\\
 There has no stationary point, we define
\begin{align*}
&\Sigma^{\pm}(\xi)=e^{\pm i \psi  }\mathbb{R}^+\cup e^{i ( \pi  \mp\psi)}\mathbb{R}^+,  \ \ \overline{\Omega} = \mathbb{C}\setminus\left( \Omega^+\cup\Omega^-\right),\\
	&\Omega^\pm (\xi)=\left\lbrace z;z=e^{ \pm \phi i}l,\ l\in\mathbb{R},\ 0<\phi<\psi \right\rbrace \cup\left\lbrace z;z=e^{\pm\phi i}l,\ l\in\mathbb{R},\ \pi\psi<\phi<\pi \right\rbrace ,
\end{align*}
where   $\psi$ is a small enough positive angle such that it is non-intersect with the curve  Im$p_-(z)=0$.
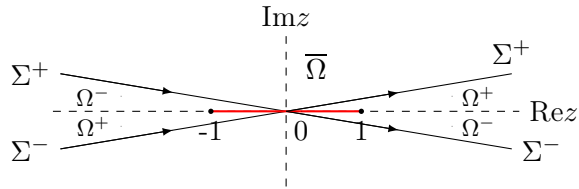
\begin{figure}[h]
	\centering
	\begin{tikzpicture}[node distance=2cm]
		\draw(0,0)--(3,0.5)node[above]{$\Sigma^+$};
		\draw(0,0)--(-3,0.5)node[left]{$\Sigma^+$};
		\draw(0,0)--(-3,-0.5)node[left]{$\Sigma^-$};
		\draw(0,0)--(3,-0.5)node[right]{$\Sigma^-$};
		\draw[dashed](-3.1,0)--(3.1,0)node[right]{ Re$z$};
		\draw[dashed](0,-1)--(0,1)node[above]{ Im$z$};
		\draw[red,thick](-1,0)--(1,0);
		\draw[-latex](-3,-0.5)--(-1.5,-0.25);
		\draw[-latex](-3,0.5)--(-1.5,0.25);
		\draw[-latex](0,0)--(1.5,0.25);
		\draw[-latex](0,0)--(1.5,-0.25);
		\coordinate (a) at (1,0);
		\fill (a) circle (1pt) node[below] {1};
		\coordinate (aa) at (-1,0);
		\fill (aa) circle (1pt) node[below] {-1};
		\coordinate (C) at (-0.2,2.2);
		\coordinate (D) at (2.2,0.2);
		\fill (D) circle (0pt) node[right] {\footnotesize $\Omega^+$};
		\coordinate (J) at (-2.2,-0.2);
		\fill (J) circle (0pt) node[left] {\footnotesize $\Omega^+$};
		\coordinate (k) at (-2.2,0.2);
		\fill (k) circle (0pt) node[left] {\footnotesize $\Omega^-$};
		\coordinate (k) at (2.2,-0.2);
		\fill (k) circle (0pt) node[right] {\footnotesize $\Omega^-$};
		\coordinate (I) at (0.2,0);
		\fill (I) circle (0pt) node[below] {$0$};
		\node at (0.4,0.6){$\overline{\Omega} $};
	\end{tikzpicture}
	\caption{\footnotesize  Figure of $\Sigma^{(1)}$ in the case of $\xi<\frac{3}{4}$. }
	\label{figR2}
\end{figure}

\textbf{(ii) }\ The case $ 1< c\leq \lambda_1$,  $\frac{3}{4}<\xi<1$.\\
Let curve $\Sigma_{j}^\pm, j=1,2 $ as Figure \ref{I2} shown.  Near the point $(0,0)$ $\Sigma_{1},\Sigma_{2}$ are the same as the case $\xi<\frac{3}{4}$.

\begin{figure}[h]
	\centering
	\begin{tikzpicture}[node distance=2cm]
		\draw[dashed](0,-1.5)--(0,1.5)node[above]{ Im$z$};
		\draw[dashed](-2,0)--(2,0);
		\draw(-1,0)--(-6,0);
		\draw(1,0)--(6,0)node[right]{ Re$z$};
          \draw[red,thick](-2,0)--(-1,0);
             \draw[red,thick](1,0)--(2,0);
		\coordinate (I) at (0,0);
		\fill (I) circle (0pt) node[below] {\footnotesize$0$};
		\coordinate (a) at (3,0);
		\fill (a) circle (1pt) node[below] {\footnotesize$\lambda_1$};
		\coordinate (aa) at (-3,0);
		\fill (aa) circle (1pt) node[below] {\footnotesize$-\lambda_1$};
		\coordinate (c) at (1,0);
		\fill (c) circle (1pt) node[below] {\footnotesize$1$};
		\coordinate (ca) at (-1,0);
		\fill (ca) circle (1pt) node[below] {\footnotesize$-1$};
		\fill (2,0) circle (1pt) node[below] {\footnotesize$c$};
		\fill (-2,0) circle (1pt) node[below] {\footnotesize$-c$};
		\coordinate (y) at (4.5,0);
		\fill (y) circle (1pt) node[below] {\footnotesize$\lambda_2$};
		\coordinate (cy) at (-4.5,0);
		\fill (cy) circle (1pt) node[below] {\footnotesize$-\lambda_2$};
		\draw(4.5,0)--(5.7,1.4)node[above]{\footnotesize$\Sigma_1^+$};
		\draw(-4.5,0)--(-5.7,1.4)node[left]{\footnotesize$\Sigma_1^+$};
		\draw(-4.5,0)--(-5.7,-1.4)node[left]{\footnotesize$\Sigma_1^-$};
		\draw(4.5,0)--(5.7,-1.4)node[right]{\footnotesize$\Sigma_1^-$};
		\draw[-latex](-5.7,-1.4)--(-5.1,-0.7);
		\draw[-latex](-5.7,1.4)--(-5.1,0.7);
		\draw[-latex](4.5,0)--(5.1,0.7);
		\draw[-latex](4.5,0)--(5.1,-0.7);
		\draw(-3,0)--(-1.5,0.6)node[above]{\footnotesize$\Sigma_1^+$};
		\draw(0,0)--(-1.5,0.6);
		\draw(3,0)--(1.5,0.6)node[above]{\footnotesize$\Sigma_1^+$};
		\draw(0,0)--(1.5,0.6);
		\draw(-3,0)--(-1.5,-0.6)node[below]{\footnotesize$\Sigma_1^-$};
		\draw(0,0)--(-1.5,-0.6);
		\draw(3,0)--(1.5,-0.6)node[below]{\footnotesize$\Sigma_1^-$};
		\draw(0,0)--(1.5,-0.6);
		\draw[-latex](-1.5,0.6)--(-0.75,0.3);
		\draw[-latex](-3,0)--(-2.25,0.3);
		\draw[-latex](-1.5,-0.6)--(-0.75,-0.3);
		\draw[-latex](-3,0)--(-2.25,-0.3);
		\draw[-latex](0,0)--(0.75,0.3);
		\draw[-latex](1.5,0.6)--(2.25,0.3);
		\draw[-latex](0,0)--(0.75,-0.3);
		\draw[-latex](1.5,-0.6)--(2.25,-0.3);
		\draw(-4.5,0)--(-3.75,0.6)node[above]{\footnotesize$\Sigma_2^+$};
		\draw(-4.5,0)--(-3.75,-0.6)node[below]{\footnotesize$\Sigma_2^-$};
		\draw(-3.75,0.6)--(-3,0);
		\draw(-3.75,-0.6)--(-3,0);
		\draw(4.5,0)--(3.75,0.6)node[above]{\footnotesize$\Sigma_2^+$};
		\draw(4.5,0)--(3.75,-0.6)node[below]{\footnotesize$\Sigma_2^-$};
		\draw(3.75,0.6)--(3,0);
		\draw(3.75,-0.6)--(3,0);
		\draw[-latex](-4.5,0)--(-4.125,0.3);
		\draw[-latex](-4.5,0)--(-4.125,-0.3);
		\draw[-latex](-3.75,0.6)--(-3.375,0.3);
		\draw[-latex](-3.75,-0.6)--(-3.375,-0.3);
		\draw[-latex](3.75,0.6)--(4.125,0.3);
		\draw[-latex](3.75,-0.6)--(4.125,-0.3);
		\draw[-latex](3,0)--(3.375,0.3);
		\draw[-latex](3,0)--(3.375,-0.3);
		\coordinate (r) at (-1.5,0.02);
		\fill (r) circle (0pt) node[below] {\footnotesize$\Omega_1^-$};
		\coordinate (r1) at (1.5,0.02);
		\fill (r1) circle (0pt) node[below] {\footnotesize$\Omega_1^-$};
		\coordinate (hh) at (-1.5,-0.02);
		\fill (hh) circle (0pt) node[above] {\footnotesize$\Omega_1^+$};
		\coordinate (hr) at (1.5,-0.02);
		\fill (hr) circle (0pt) node[above] {\footnotesize$\Omega_1^+$};
		\coordinate (r) at (-3.75,0.02);
		\fill (r) circle (0pt) node[below] {\footnotesize$\Omega_2^-$};
		\fill (3.75,0.02) circle (0pt) node[below] {\footnotesize$\Omega_2^-$};
		\coordinate (r) at (3.75,-0.02);
		\fill (r) circle (0pt) node[above] {\footnotesize$\Omega_2^+$};
		\fill (-3.75,0.02) circle (0pt) node[above] {\footnotesize$\Omega_2^+$};
		\coordinate (C) at (-0.2,2.2);
		\coordinate (D) at (5.2,0.6);
		\fill (D) circle (0pt) node[right] {\footnotesize $\Omega_1^+$};
		\coordinate (J) at (-5.2,-0.6);
		\fill (J) circle (0pt) node[left] {\footnotesize $\Omega_1^-$};
		\coordinate (k) at (-5.2,0.6);
		\fill (k) circle (0pt) node[left] {\footnotesize $\Omega_1^+$};
		\coordinate (k) at (5.2,-0.6);
		\fill (k) circle (0pt) node[right] {\footnotesize $\Omega_1^-$};
		\coordinate (v) at (0.2,1.2);
		\fill (v) circle (0pt) node[right] {\footnotesize $\overline{\Omega} $};	
	\end{tikzpicture}
	\caption{\footnotesize  Figure of $\Sigma^{(1)}$ in the case of  $ 1< c\leq \lambda_1$,  $\frac{3}{4}<\xi<1$. }
	\label{I2}
\end{figure}
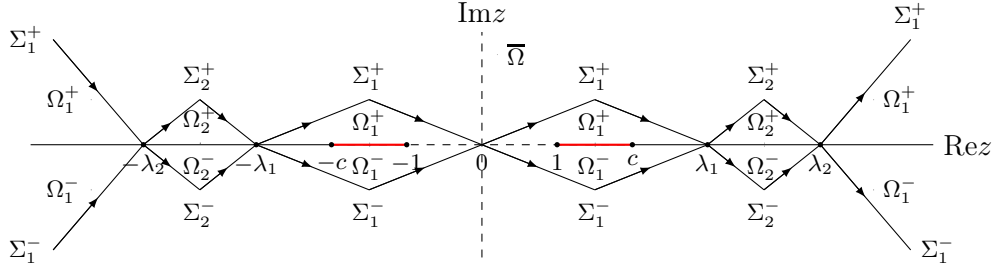

\textbf{(iii) }\ The case $ 1< c\leq \lambda_1, \  1<\xi$.\\
Let the  curve $\Sigma_{j}^\pm, j=1,2 $ as Figure \ref{I3} showing. The only difference from the case \textbf{(ii)} is there is only two phase point $\pm \lambda_1$.
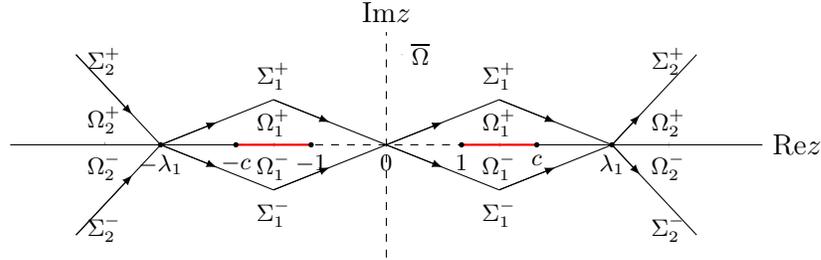
\begin{figure}[h]
	\centering
	\begin{tikzpicture}[node distance=2cm]
		\draw[dashed](0,-1.5)--(0,1.5)node[above]{ Im$z$};
		\draw[dashed](-2,0)--(2,0);
		\draw(-1,0)--(-5,0);
		\draw(1,0)--(5,0)node[right]{ Re$z$};
          \draw[red,thick](-2,0)--(-1,0);
             \draw[red,thick](1,0)--(2,0);
		\coordinate (I) at (0,0);
		\fill (I) circle (0pt) node[below] {\footnotesize$0$};
		\coordinate (a) at (3,0);
		\fill (a) circle (1pt) node[below] {\footnotesize$\lambda_1$};
		\coordinate (aa) at (-3,0);
		\fill (aa) circle (1pt) node[below] {\footnotesize$-\lambda_1$};
		\coordinate (c) at (1,0);
		\fill (c) circle (1pt) node[below] {\footnotesize$1$};
		\coordinate (ca) at (-1,0);
		\fill (ca) circle (1pt) node[below] {\footnotesize$-1$};
		\fill (2,0) circle (1pt) node[below] {\footnotesize$c$};
		\fill (-2,0) circle (1pt) node[below] {\footnotesize$-c$};

		\draw(-3,0)--(-1.5,0.6)node[above]{\footnotesize$\Sigma_1^+$};
		\draw(0,0)--(-1.5,0.6);
		\draw(3,0)--(1.5,0.6)node[above]{\footnotesize$\Sigma_1^+$};
		\draw(0,0)--(1.5,0.6);
		\draw(-3,0)--(-1.5,-0.6)node[below]{\footnotesize$\Sigma_1^-$};
		\draw(0,0)--(-1.5,-0.6);
		\draw(3,0)--(1.5,-0.6)node[below]{\footnotesize$\Sigma_1^-$};
		\draw(0,0)--(1.5,-0.6);
		\draw[-latex](-1.5,0.6)--(-0.75,0.3);
		\draw[-latex](-3,0)--(-2.25,0.3);
		\draw[-latex](-1.5,-0.6)--(-0.75,-0.3);
		\draw[-latex](-3,0)--(-2.25,-0.3);
		\draw[-latex](0,0)--(0.75,0.3);
		\draw[-latex](1.5,0.6)--(2.25,0.3);
		\draw[-latex](0,0)--(0.75,-0.3);
		\draw[-latex](1.5,-0.6)--(2.25,-0.3);
		\draw(-4.125,1.2)--(-3,0);
		\draw(-4.125,-1.2)--(-3,0);
		\draw(4.125,1.2)--(3.75,0.8)node[above]{\footnotesize$\Sigma_2^+$};
		\draw(4.125,-1.2)--(3.75,-0.8)node[below]{\footnotesize$\Sigma_2^-$};
		\draw(-4.125,1.2)--(-3.75,0.8)node[above]{\footnotesize$\Sigma_2^+$};
		\draw(-4.125,-1.2)--(-3.75,-0.8)node[below]{\footnotesize$\Sigma_2^-$};
		\draw(3.75,0.8)--(3,0);
		\draw(3.75,-0.8)--(3,0);
		\draw[-latex](-3.75,0.8)--(-3.375,0.4);
		\draw[-latex](-3.75,-0.8)--(-3.375,-0.4);
		\draw[-latex](3,0)--(3.375,0.4);
		\draw[-latex](3,0)--(3.375,-0.4);
		\coordinate (r) at (-1.5,0.02);
		\fill (r) circle (0pt) node[below] {\footnotesize$\Omega_1^-$};
		\coordinate (r1) at (1.5,0.02);
		\fill (r1) circle (0pt) node[below] {\footnotesize$\Omega_1^-$};
		\coordinate (hh) at (-1.5,-0.02);
		\fill (hh) circle (0pt) node[above] {\footnotesize$\Omega_1^+$};
		\coordinate (hr) at (1.5,-0.02);
		\fill (hr) circle (0pt) node[above] {\footnotesize$\Omega_1^+$};
		\coordinate (r) at (-3.75,0.02);
		\fill (r) circle (0pt) node[below] {\footnotesize$\Omega_2^-$};
		\fill (3.75,0.02) circle (0pt) node[below] {\footnotesize$\Omega_2^-$};
		\coordinate (r) at (3.75,-0.02);
		\fill (r) circle (0pt) node[above] {\footnotesize$\Omega_2^+$};
		\fill (-3.75,0.02) circle (0pt) node[above] {\footnotesize$\Omega_2^+$};
		\coordinate (C) at (-0.2,2.2);
		\coordinate (v) at (0.2,1.2);
		\fill (v) circle (0pt) node[right] {\footnotesize $\overline{\Omega} $};	
	\end{tikzpicture}
	\caption{\footnotesize  Figure of $\Sigma^{(1)}$ in the case of  $ 1< c\leq \lambda_1, \  1<\xi$. }
	\label{I3}
\end{figure}

 Moreover, in this region depending on $\xi$, $c$,  we introduce
a piecewise matrix interpolation function
\begin{equation}
	G(z;\xi,c)=\left\{ \begin{array}{ll}
	\left(\begin{array}{cc}
		1 & 0\\
		-r_1e^{2itp_-} & 1
	\end{array}\right),   &\text{as } z\in\Omega_1^+;\\[12pt]
	\left(\begin{array}{cc}
		1 & -r_2e^{-2itp_-}\\
		0 & 1
	\end{array}\right),   &\text{as } z\in\Omega_1^-;\\
	\left(\begin{array}{cc}
		1 &\frac{ r_2e^{-2itp_-}}{1-r_1r_2}\\
		0 & 1
	\end{array}\right) &\text{as } z\in\Omega_2^+;\\
	\left(\begin{array}{cc}
		1 &0\\
		\frac{r_1e^{2itp_-}}{1-r_1r_2} & 1
	\end{array}\right), &\text{as } z\in\Omega_2^-;\\
	I &\text{as } 	z \text{ in elsewhere},
\end{array}\right. \label{funcG-}
\end{equation}
To deal with the jump on $\mathbb{R}$, we denote a interval
\begin{align}
	I(\xi)=\left\{ \begin{array}{ll}
		\emptyset,   &\text{as } \xi<\frac{3}{4};\\[12pt]
		[-\lambda_2,-\lambda_1]\cup[\lambda_1,\lambda_2],   &\text{as } \frac{3}{4}<\xi<1;\\[12pt]
		(-\infty,-\lambda_1]\cup[\lambda_1,+\infty) &\text{as } \xi>1;
	\end{array}\right.
\end{align}
and   introduce  an  auxiliary function $\delta(z)=\delta(z;\xi,c)$, which satisfies  the following
scalar RH problem.

	(a) $\delta(z)$ is analytic on $\mathbb{C}\setminus I(\xi)$;

    (b)  $\delta_-(z)=\delta_+(z)(1-r_1r_2),\ \ z\in I(\xi), \ \ \delta_-(z)=\delta_+(z),\ \ z\in \mathbb{R} \setminus I(\xi);$

	(c) $\delta(z)\to 1$, \ as $z\to\infty\in\mathbb{C}\setminus I(\xi)$.\\
The solution of the RH problem can given by
\begin{align}
	\log \delta(z)=-\frac{1}{2\pi i}\int_{I(\xi)}\dfrac{\log(1-r_1(s)r_2(s))}{s-z}ds.
\end{align}
which has the following properties
	\begin{align}
		\delta(z)=&\exp\left\lbrace I_{\delta}^1\right\rbrace\cdot \left( 1+zI_{\delta}^2\right)
		+\mathcal{O}(z^2), \ z\to0\in\mathbb{C}^+, \nonumber
	\end{align}
where
	\begin{align}
		I_{\delta}^1=&-\frac{c}{2\pi }\int_{I(\xi)}\dfrac{\log(1-r_1(s)r_2(s))}{s}ds,\label{Ide0}\\
		I_{\delta}^2=&-\frac{c}{2\pi }\int_{I(\xi)}\dfrac{\log(1-r_1(s)r_2(s))}{s^2}ds.\label{Ide1}
	\end{align}
 Define a  new  matrix-valued   function
\begin{equation}
	M^{(1)}(z;\xi,c)\triangleq  N(z)G(z;\xi,c)\delta^{\sigma_3},
\end{equation}
which  then satisfies the following RH problem.

\begin{RHP}\label{RHP-1}
	Find a matrix-valued function  $M^{(1)}(z)= M^{(1)}(z;\xi,c )$ which satisfies
	
	$\blacktriangleright$ Analyticity: $M^{(1)}(z)$ is meromorphic in $\mathbb{C}\setminus  \Sigma^{(1)}$, where
	\begin{equation}
		\Sigma^{(1)}(\xi)= \left(\cup_{j=1}^4\Sigma_j(\xi)\right)\cup \left[-1,1\right] ;
	\end{equation}

	$\blacktriangleright$ Symmetry: $M^{(1)}(z)=\sigma_2M^{(1)}(-z)\sigma_2$=$\sigma_1\overline{M^{(1)}(\bar{z})}\sigma_1$;
	
	$\blacktriangleright$ Jump condition: $M^{(1)}$ has continuous boundary values $M^{(1)}_\pm(z)$ on  $ \Sigma^{(1)}(\xi)$ and
	\begin{equation}
		M^{(1)}_+(z)=M^{(1)}_-(z)V^{(1)}(z),\hspace{0.5cm}z \in  \Sigma^{(1)},
	\end{equation}
	where
	\begin{equation}
		V^{(1)}(z) =\left\{ \begin{array}{ll}
				\left(\begin{array}{cc}
					1 & 0\\
					r_1\delta^2e^{2itp_-} & 1
				\end{array}\right),   &\text{as } z\in\Sigma_1^+;\\[12pt]
				\left(\begin{array}{cc}
					1 & -r_2\delta^{-2}e^{-2itp_-}\\
					0 & 1
				\end{array}\right),   &\text{as } z\in\Sigma_1^-;\\
				\left(\begin{array}{cc}
					1 &-\frac{ r_2\delta^{-2}e^{-2itp_-}}{1-r_1r_2}\\
					0 & 1
				\end{array}\right) &\text{as } z\in\Sigma_2^+;\\
				\left(\begin{array}{cc}
					1 &0\\
					\frac{r_1\delta^2e^{2itp_-}}{1-r_1r_2} & 1
				\end{array}\right), &\text{as } z\in\Sigma_2^-;	\\
			\left(\begin{array}{cc}
				0 & i\\
				i & 0
			\end{array}\right), z\in[-1,1]
			\end{array}\right..
	\end{equation}
	
	$\blacktriangleright$ Asymptotic behaviors: $M^{(1)}(z) = I+\mathcal{O}(z^{-1}),\hspace{0.5cm}z \rightarrow \infty;$

	$\blacktriangleright$ Singularity: $M^{(1)}(z)$ has singularity at $z=\pm1$ with
	\begin{align}
		&M^{(1)}(z)\sim \mathcal{O}(z\mp1)^{-1/4} ,\ z\to \pm1\text{ in }\mathbb{C}\setminus \Sigma^{(1)}.
	\end{align}
\end{RHP}
We construct the solution $M^{(1)}(z)$ as follow
\begin{equation}
	M^{(1)}(z) =\left\{\begin{array}{ll}
		E(z;\xi,c)M^{mod1}(z;\xi,c), & z\notin U_{\pm \lambda_1}\cup U_{\pm \lambda_2},\\[4pt]
		E(z;\xi,c)M^{1,\pm}(z;\xi,c),  &z\in U_{\pm \lambda_1},\\[4pt]
		E(z;\xi,c)M^{2,\pm}(z;\xi,c),  &z\in U_{\pm \lambda_2},
	\end{array}\right. \label{transm0}
\end{equation}
where  we denote $ U_{\pm \lambda_j}$ as the   neighborhood of $\pm \lambda_j$
\begin{equation}
	U_{\pm \lambda_j}= \left\lbrace z:|z\pm \lambda_j|\leq \varrho \right\rbrace .
\end{equation}
Here, $\varrho$ is a small positive constant such that $\varrho<\min\left\lbrace \frac{z_j-1}{3}, \frac{z_j-c}{3},\epsilon \right\rbrace $.
In the case (i),  the   jump matrix exponentially  decays  to the identity matrix $I$ as $t\to\infty$ on $\Sigma_1$ and $\Sigma_2$ because of the absence of phase point, namely, $U(\pm \lambda_j)=\emptyset$.  Although $z=0$ is a pole of $p_-$,  the exponential function   decay to $0$ at a speed of $e^{-C(\xi,c)t/|z|^2}$. So this case is trivial. Then the   existence and uniqueness of $E(z;\xi,c)$ can be shown  by  a  small-norm RH problem \cite{RN9,RN10} with
\begin{equation}
	E(z)=I+\mathcal{O}(e^{-C(\xi,c)t}) \label{asyEr-1}
\end{equation}
with  $C(\xi,c)$ being  a positive constant rely on $\xi$ and $c$. However, in the case(ii) and (iii), the phase points have contribution on $t\to\infty$.  The jump matrix exponentially  decaying to the identity matrix $I$ as $t\to\infty$ away from phase points.
\subsection{A model RH problem on  cuts }
\quad The jump matrix exponentially  decays to the identity matrix $I$ as $t\to\infty$ on $\Sigma_1^\pm$, which finally leads to the model RH problem:
\begin{RHP}\label{model1}
	Find a matrix-valued function  $  M^{mod1}(z )$ which satisfies 	
	
	$\blacktriangleright$ Analyticity: $  M^{mod1}(z )$ is holomorphic in  $\mathbb{C}\setminus[-1,1]$;
	
	$\blacktriangleright$ Jump condition: $M^{mod1}$ has continuous boundary values $M^{mod1}_\pm(z)$

\quad on $[-1,1]$ and
	\begin{equation}
		M^{mod1}_+(z)=M^{mod1}_-(z)V^{mod1}(z),\hspace{0.5cm}z \in[-1,1],
	\end{equation}
	where
	\begin{equation}
		V^{mod1}(z)=\left(\begin{array}{cc}
			0 & i\\
			i & 0
		\end{array}\right), z\in[-1,1];
	\end{equation}

	$\blacktriangleright$ Asymptotic behaviors: $M^{mod1}(z) = I+\mathcal{O}(z^{-1}),\hspace{0.5cm}z \rightarrow \infty;$
 
	$\blacktriangleright$ Singularity: $M^{mod1}(z)$ has singularity at $z=\pm 1$ with:
	\begin{align}
		&M^{mod1}(z)\sim \mathcal{O}(z\mp c)^{-1/4} ,\ z\to \pm 1\text{ in }\mathbb{C}\setminus\mathbb{R}.
	\end{align}
\end{RHP}
The solution of  this  model RH problem  can be given  by
\begin{equation}
	M^{mod1}(z)=\frac{1}{\sqrt{2}}\left(\begin{array}{cc}
		\phi_1(z)+\phi_1(z)^{-1} & \phi_1(z)-\phi_1(z)^{-1}\\
		\phi_1(z)-\phi_1(z)^{-1} & \phi_1(z)+\phi_1(z)^{-1}
	\end{array}\right),
\end{equation}
where $\phi_1(z)=\left( \frac{1+z}{1-z}\right) ^{1/4}$.
As $z\to 0\in\mathbb{C}^+$,
\begin{equation}
	M^{mod1}(z)=\sqrt{2}\left(\begin{array}{cc}
		0 & i\\
		i & 0
	\end{array}\right)+\frac{zi}{\sqrt{2}}I+\mathcal{O}(z^2)
\end{equation}

\subsection{Localized RH problem near phase points}
\quad
As $t\to +\infty$, we consider to reduce the  RHP \ref{RHP-1} to a model RH problem  whose solution can be given explicitly in terms of parabolic cylinder
functions on every contour $\Sigma^{(0)}_{j,\pm}=\Sigma^{(1)}\cap U(\pm \lambda_j)$ respectively. And we only give the details of $\Sigma^{(0)}_{1,+}$, the model of other  critical point can be  constructed similar. We denote $\hat{\Sigma}^{(0)}_{1,+}$ as the contour $\{z=\lambda_1+le^{\pm\varphi i},\ l\in\mathbb{R}\}$ oriented from $\Sigma^{(0)}_{1,+}$, and $\hat{\Sigma}_{j,\pm}$ is the extension of $\Sigma_{j,\pm}$  respectively. And for $z$ near $\lambda_1$, note that $p_-''(\lambda_1)>0$, so we rewrite phase function $p_-$ as
\begin{align}
	p_-(z)=p_-(\lambda_1)+(z-\lambda_1)^2\frac{p_-''(\lambda_1)}{2}+\mathcal{O}((z-\lambda_1)^3).
\end{align}

\begin{figure}
	\centering
		\begin{tikzpicture}[node distance=2cm]
			\draw[->](0,0)--(2.5,1.4)node[right]{$\hat{\Sigma}_{3}$};
			\draw(0,0)--(-2.5,1.4)node[left]{$\hat{\Sigma}_{1}$};
			\draw(0,0)--(-2.5,-1.4)node[left]{$\hat{\Sigma}_{2}$};
			\draw[->](0,0)--(2.5,-1.4)node[right]{$\hat{\Sigma}_{4}$};
			\draw[dashed](-3.8,0)--(3.8,0)node[right]{\scriptsize Re$z$};
			\draw[->](-2.5,-1.4)--(-1.25,-0.7);
			\draw[->](-2.5,1.4)--(-1.25,0.7);
			\coordinate (A) at (-1.2,0.5);
			\coordinate (B) at (-1.2,-0.5);
			\coordinate (G) at (1.4,0.5);
			\coordinate (H) at (1.4,-0.5);
			\coordinate (I) at (0,0);
		\end{tikzpicture}
		\caption{\footnotesize The jump  contour $\hat{\Sigma}^{(0)}_+$   near the point $z_2$.}
	\label{figS00}
\end{figure}
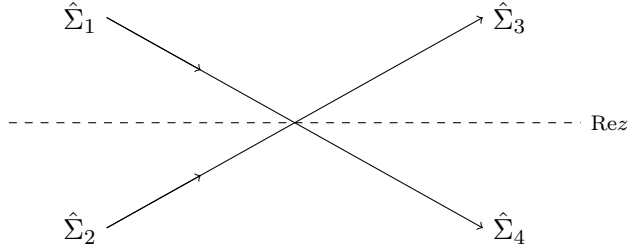
Consider following local RH problem
\begin{RHP}\label{RHPlo0}
	Find a matrix-valued function  $ M^{1,+}(z)$ with following properties
	
	$\blacktriangleright$ Analyticity: $M^{1,+}(z)$ is analytical  in $\mathbb{C}\setminus \hat{\Sigma}^{(0)}_+ $;

	$\blacktriangleright$ Jump condition: $M^{1,+}(z)$ has continuous boundary values $M^{1,+}_\pm$ on $\hat{\Sigma}^{(0)}_+$ and
	\begin{equation}
		M^{1,+}_+(z)=M^{1,+}_-(z)V^{1,+}(z),\hspace{0.2cm}z \in \hat{\Sigma}^{(0)}_+,
	\end{equation}
	where jump matrix $V^{1,+}(z)$ is given by (see Figure \ref{figS00})
	\begin{align}
		V^{1,+}(z)=\left\{ \begin{array}{ll}
			\left(\begin{array}{cc}
				1 & 0\\
				r_1\delta^2e^{2itp_-} & 1
			\end{array}\right),   &\text{as } z\in\hat{\Sigma}_1;\\[12pt]
			\left(\begin{array}{cc}
				1 & -r_2\delta^{-2}e^{-2itp_-}\\
				0 & 1
			\end{array}\right),   &\text{as } z\in\hat{\Sigma}_2;\\
			\left(\begin{array}{cc}
				1 &-\frac{ r_2\delta^{-2}e^{-2itp_-}}{1-r_1r_2}\\
				0 & 1
			\end{array}\right) &\text{as } z\in\hat{\Sigma}_3;\\
			\left(\begin{array}{cc}
				1 &0\\
				\frac{r_1\delta^2e^{2itp_-}}{1-r_1r_2} & 1
			\end{array}\right), &\text{as } z\in\hat{\Sigma}_4;	
		\end{array}\right..
	\end{align}
	
	$\blacktriangleright$ Asymptotic behaviors: $M^{1,+}(z) =  I+\mathcal{O}(z^{-1}),\hspace{0.5cm}z \rightarrow \infty.$
	 
\end{RHP}	

RHP \ref{RHPlo0} does not possess the symmetry condition, because it is a local model and will only be used for bounded values of $z$.
In order to motivate the model, let $\zeta = \zeta(z)$ denote the rescaled local variable
\begin{align}
	\zeta(z)=2t^{1/2}\sqrt{p_-''(\lambda_1)}(z-\lambda_1).
\end{align}
This map is a conformal bijection maps $U( \lambda_1)$ to an expanding neighborhood of $\zeta= 0$. We choose the branch which maps the upper half plane to the lower half plane. Moreover, we denote:
\begin{align}
	r_{\lambda_1}^\pm=r_1(\pm \lambda_1)\delta^2_{\lambda_1}(\pm z_2)e^{2itp_-(\pm z_2)}(4tp_-''(\lambda_1))^{i\nu(\pm\lambda_1)}.
\end{align}
where
 $$\nu(\lambda_1)=\frac{1}{2\pi}\log(1-r_1r_2(\lambda_1)).$$

Through this change of variable,  the jump $V^{1,+}(z)$ approximates to  the jump of a parabolic cylinder model problem as follow:
\begin{RHP}\label{RHPpc0}
	Find a matrix-valued function  $ M^{pc}(\zeta;\xi)$ with following properties:
	
	$\blacktriangleright$ Analyticity: $M^{pc}(\zeta;\xi)$ is analytical  in $\mathbb{C}\setminus \Sigma^{pc} $ with $\Sigma^{pc}=\left\lbrace\mathbb{R}e^{\varphi i} \right\rbrace \cup \left\lbrace\mathbb{R}e^{(\pi-\varphi) i} \right\rbrace$ shown in Figure \ref{sigpc0};

	$\blacktriangleright$ Jump condition: $M^{pc}$ has continuous boundary values $M^{pc}_\pm$ on $\Sigma^{pc}$ and
	\begin{equation}
		M^{pc}_+(\zeta;\xi)=M^{pc}_-(\zeta;\xi)V^{pc}(\zeta),\hspace{0.5cm}\zeta \in \Sigma^{\zeta},
	\end{equation}
	where
	\begin{align}
		V^{pc}(\zeta;\xi)=\left\{\begin{array}{ll}
			\left(\begin{array}{cc}
				1 & 0\\
				r_{\lambda_1}^+\zeta^{2i\nu(\lambda_1)}e^{\frac{i}{2}\zeta^2} & 1
			\end{array}\right),  & \zeta\in\mathbb{R}^+e^{\varphi i},\\[10pt]
			\left(\begin{array}{cc}
				1& -\bar{r}_{\lambda_1}\zeta^{-2i\nu(\lambda_1)}e^{-\frac{i}{2}\zeta^2}\\
				0&1
			\end{array}\right),   & \zeta\in \mathbb{R}^+e^{-\varphi i},\\[10pt]
			\left(\begin{array}{cc}
				1 & 0\\
				\frac{r_{\lambda_1}^+}{1-|r_{\lambda_1}^+|^2}\zeta^{2i\nu(\lambda_1)}e^{\frac{i}{2}\zeta^2} & 1
			\end{array}\right),   & \zeta\in \mathbb{R}^+e^{(-\pi+\varphi) i},\\[10pt]
			\left(\begin{array}{cc}
				1 & -\frac{\bar{r}_{\lambda_1}}{1-|r_{\lambda_1}^+|^2}\zeta^{-2i\nu(\lambda_1)}e^{-\frac{i}{2}\zeta^2}\\
				0 & 1
			\end{array}\right),   & \zeta\in \mathbb{R}^+e^{(\pi-\varphi) i}.
		\end{array}\right.
	\end{align}

	$\blacktriangleright$ Asymptotic behaviors: $M^{pc}(\zeta;\xi) =  I+M^{pc}_1\zeta^{-1}+\mathcal{O}(\zeta^{-2}),\hspace{0.5cm}\zeta \rightarrow \infty.$

\end{RHP}	
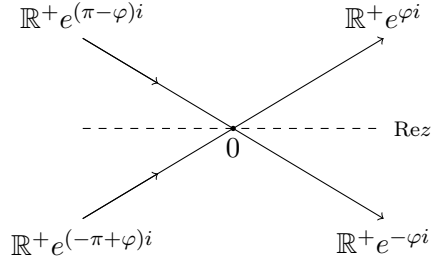
\begin{figure}
	\centering
	\begin{tikzpicture}[node distance=2cm]
		\draw[->](0,0)--(2,1.2)node[above]{$\mathbb{R}^+e^{\varphi i}$};
		\draw(0,0)--(-2,1.2)node[above]{$\mathbb{R}^+e^{(\pi-\varphi)i}$};
		\draw(0,0)--(-2,-1.2)node[below]{$\mathbb{R}^+e^{(-\pi+\varphi)i}$};
		\draw[->](0,0)--(2,-1.2)node[below]{$\mathbb{R}^+e^{-\varphi i}$};
		\draw[dashed](-2,0)--(2,0)node[right]{\scriptsize Re$z$};
		\draw[->](-2,-1.2)--(-1,-0.6);
		\draw[->](-2,1.2)--(-1,0.6);
		\coordinate (A) at (-1.2,0.5);
		\coordinate (B) at (-1.2,-0.5);
		\coordinate (G) at (1.4,0.5);
		\coordinate (H) at (1.4,-0.5);
		\coordinate (I) at (0,0);
		\fill (I) circle (1pt) node[below] {$0$};
	\end{tikzpicture}
	\caption{\footnotesize The jump  contour $\Sigma^{pc}$ for $ M^{pc}(\zeta;\xi)$.}
	\label{sigpc0}
\end{figure}
Then the solution of the RHP \ref{RHPlo0} can be given by  ( for example,  see \cite{HG2009} Theorem A.1-A.6)
\begin{align}
	M^{j,\pm}(z)=I+\frac{t^{-1/2}}{z\mp\lambda_j}\frac{i}{2\sqrt{p_-''(\pm\lambda_j)}} \left(\begin{array}{cc}
		0 & [M^{j,pc}_1]_{12}\\
		-[M^{j,pc}_1]_{21} & 0
	\end{array}\right)+\mathcal{O}(t^{-1}).\label{asyMlo0}
\end{align}
While the  RHP \ref{RHPpc0} has an explicit solution, which is expressed in terms of solutions of the parabolic cylinder equation
$$ \frac{\partial^2 D_a(z)}{\partial z^2}+\left(\frac{1}{2}-\frac{z^2}{2}+a \right) D_a(z)=0.$$
A derivation of this result is given in \cite{RN6}. Substitute above consequence into (\ref{asyMlo0}) and obtain
\begin{align}
	M^{j,\pm}(z)=I+\frac{t^{-1/2}}{z\mp\lambda_j}\frac{i}{2\sqrt{p_-''(\pm\lambda_j)}} \left(\begin{array}{cc}
		0 & \tilde{\beta}^{j,\pm}_{12}\\
		\tilde{\beta}^{j,\pm}_{21} & 0
	\end{array}\right)+\mathcal{O}(t^{-1}),\label{asyMpc}
\end{align}
where
\begin{align}
	&\tilde{\beta}^{j,\pm}_{12}=\frac{\sqrt{2\pi}e^{-\frac{1}{2}\pi\nu(\pm\lambda_j)}e^{\frac{\pi i}{4} i}}{r_{\pm\lambda_j}^\pm\Gamma(-i\nu(\pm\lambda_j))},\hspace{0.5cm}\tilde{\beta}^{j,\pm}_{21}\tilde{\beta}^{j,\pm}_{12}=-\nu(\pm\lambda_j), \ \ j=1,2.
\end{align}
We finally  obtain
\begin{Proposition}\label{asympc0}
	As $t\to+\infty$,
	\begin{align}
		M^{j,\pm}(z)=I+t^{-1/2} \frac{A_{j,\pm}(\xi)}{z\mp z_2} +\mathcal{O}(t^{-1}),\ j=1,2,
	\end{align}
	where
	\begin{align}
		A_{j,\pm}(\xi)=\frac{i}{2\sqrt{p_-''(\pm \lambda_j)}} \left(\begin{array}{cc}
			0 & \tilde{\beta}^{j,\pm}_{12}\\
			\tilde{\beta}^{j,\pm}_{21} & 0
		\end{array}\right).\label{pcA0}
	\end{align}
\end{Proposition}

\subsection{The small norm RH problem  for error function }\label{erroranalysis}
\quad In this subsection,  we consider the error matrix-function $E(z;\xi,c)$ in this region.

\begin{RHP}\label{RHPE0}
	Find a matrix-valued function $E(z;\xi,c)$  with following properties:
	
	$\blacktriangleright$ Analyticity: $E(z;\xi,c)$ is analytical  in $\mathbb{C}\setminus  \Sigma^{ E } $, where
	$$\Sigma^{ E }= \partial U_\xi \cup
	\left[ \Sigma^{(1)}\setminus \left( U_\xi \cup  \left[-1,1\right]\right) \right], \ U_\xi =(\cup_{j=1,2} U_{\pm \lambda_j}) \cup   U_{\pm 1};$$

	$\blacktriangleright$ Asymptotic behaviors:
	\begin{align}
		&E(z;\xi,c) \sim I+\mathcal{O}(z^{-1}),\hspace{0.5cm}|z| \rightarrow \infty;
	\end{align}

	$\blacktriangleright$ Jump condition: $E(z;\xi,c)$ has continuous boundary values $E_\pm(z;\xi,c)$ on $\Sigma^{ E }$ satisfying
	$$E_+(z;\xi,c)=E_-(z;\xi,c)V^{E}(z),$$
	where the jump matrix $V^{E}(z)$ is given by
	\begin{equation}
		V^{ E }(z)=\left\{\begin{array}{llll}
			M^{mod1}(z_-)V^{(1)}(z)M^{mod1}(z_+)^{-1}, & z\in \Sigma^{ E }\setminus \partial U_\xi,\\[4pt]
			M^{j,\pm}(z)M^{mod1}(z)^{-1},  & z\in \partial U_{\xi},
		\end{array}\right. \label{deVE0}
	\end{equation}
	
\end{RHP}
We will show  that for large times, the error function  $E(z;\xi,c)$  solves following small norm  RH problem.

Out of $U_\xi $,  the jump $V^{ E }$ has the following estimates
\begin{equation}
	\parallel V^{ E }(z)-I\parallel_p\lesssim
	\exp\left\{-tK_p\right\},   z\in \Sigma^{ E }\setminus U_\xi , \ p\in[1,\infty].\label{VE-I2}
\end{equation}

For $z\in \partial U_\xi $,  $M^{mod1}(z)$ is bounded,  so   by using  (\ref{pcA0}),  we find that
\begin{align}
	| V^{ E }(z)-I|&=   \mathcal{O}(t^{-1/2}).\label{VE0}
\end{align}
Therefore,    the   existence and uniqueness  of  the RHP  \ref{RHPE0} can  shown  by using  a  small-norm RH problem \cite{RN9,RN10}.  Moreover, according to Beal-Coifman theory,
the solution of  the RHP \ref{RHPE0}  can be given by
\begin{equation}
	E(z;\xi,c)=I+\frac{1}{2\pi i}\int_{\Sigma^{E}}\dfrac{\left( I+\varpi(s)\right) (V^{E}(s)-I)}{s-z}ds,\label{Ez0}
\end{equation}
where the $\varpi\in L^\infty(\Sigma^{E})$ is the unique solution of following equation
\begin{equation}
	(1-C_E)\varpi=C_E\left(I \right),
\end{equation}
and  $C_E$ is  a integral operator: $L^\infty(\Sigma^{E})\to L^2(\Sigma^{E})$ defined by
\begin{equation}
	C_E(f)(z)=C_-\left( f(V^{E}(z) -I)\right),
\end{equation}
where the $C_-$ is the usual Cauchy projection operator on $\Sigma^{E}$
\begin{equation}
	C_-(f)(s)=\lim_{z\to \Sigma^{E}_-}\frac{1}{2\pi i}\int_{\Sigma^{E}}\dfrac{f(s)}{s-z}ds.
\end{equation}
By  (\ref{VE0}),   we have
\begin{equation}
	\parallel C_E\parallel\leq\parallel C_-\parallel \parallel V^{E}(z)-I\parallel_2 \lesssim \mathcal{O}(t^{-1/2}),
\end{equation}
which implies that  $1-C_E$ is invertible  for   sufficiently large $t$.    So  $\varpi$  exists and is unique.
Besides,
\begin{equation}
	\parallel \varpi\parallel_{L^\infty(\Sigma^{E})}\lesssim\dfrac{\parallel C_E\parallel}{1-\parallel C_E\parallel}\lesssim t^{-1/2}.
\end{equation}
In order to reconstruct the solution $u(y,t)$ of (\ref{mch}), we need the asymptotic behavior of $E(z;\xi,c)$ as $z\to 0\in \mathbb{C}^+$ and the long time asymptotic behavior of $E(0)$. Note that when we estimate its  asymptotic behavior, from (\ref{Ez0}) and (\ref{VE-I}) we only need to consider the calculation on $\partial U_\xi$ because it  approach zero exponentially on other boundary.
\begin{Proposition}\label{AsyE0I}
	As $z\to 0\in \mathbb{C}^+$, we have
	\begin{align}
		E(z;\xi,c)=E(0)+E_1z+\mathcal{O}(z^2),
	\end{align}
	where
	\begin{align}
		E(0)=I+\frac{1}{2\pi i}\int_{\Sigma^{E}}\dfrac{\left( I+\varpi(s)\right) (V^{E}-I)}{s}ds,
	\end{align}
	with long time asymptotic behavior
	\begin{equation}
		E(0)=I+t^{-1/2}H^{(0)}+\mathcal{O}(t^{-1}),
	\end{equation}
	where
	\begin{align}
		H^{(0)}=	 &\sum_{ p=\pm \lambda_j,j=1,2 }\frac{M^{mod1}(p)A_{j,\pm}(\xi)M^{mod1}(p)^{-1}}{p}.
	\end{align}
	Here  $A_{j,\pm}(\xi)$ is given by (\ref{pcA0}).
	And
	\begin{equation}
		E_1=\frac{1}{2\pi i}\int_{\Sigma^{E}}\dfrac{\left( I+\varpi(s)\right) (V^{E}-I)}{s^2}ds= t^{-1/2}H^{(1)}+\mathcal{O}(t^{-1}),\nonumber
	\end{equation}
	where
	\begin{align}
		H^{(1)}= &\sum_{ p=\pm z_2 }\frac{M^{mod}(p)A_\pm(\xi)M^{mod}(p)^{-1}}{p^2}.
	\end{align}
\end{Proposition}
\begin{proof}
	Substitute the long time asymptotic behavior of $V^{E}$, $\varpi(s)$ and Proposition \ref{asympc0} into $2\pi i(E(0)-I)$:
	\begin{align}
		&\int_{\Sigma^{E}}\dfrac{\left( I+\varpi(s)\right) (V^{E}-I)}{s}ds\nonumber\\
		&=\int_{\partial U_{(\xi)}}\dfrac{ 	M^{mod1}(s)(M^{j,\pm}(s)-I)M^{mod1}(s)^{-1}}{s}ds+\mathcal{O}(t^{-1})\nonumber\\
		&=t^{-1/2}\int_{\partial U_{(\xi)}}\dfrac{ M^{mod1}(s)A_{j,\pm}(\xi)M^{mod1}(s)^{-1}}{s(z\mp \lambda_j)}ds+\mathcal{O}(t^{-1}).
	\end{align}
	Then by residue theorem we finally arrive at the result.
\end{proof}

\section{Fast-decay    background  region  }\label{sec+}
\quad The Region  II  is corresponding to the case $\xi>1+2/c$. In this case, we introduce a new scalar function
\begin{align}
	X(z)=\sqrt{z^2-c^2},\hspace{0.3cm}\theta_+(z)=X(z)\left( \frac{1-\xi}{2}+\frac{1}{cz^2}\right),\label{theta+}
\end{align}
where   $X(z)$ is analytic on $\mathbb{C}\setminus[-c,c]$ and  take the single-valued analytic branch such that $X(z)\in i\mathbb{R}^+$ on $[-c,c]_+$.
 Then
\begin{align}
	\frac{\text{d}\theta_+}{\text{d} z}=\frac{1}{z^3	X(z)}\left( \frac{\xi-1}{2}z^4-\frac{z^2}{c}+2c \right).
\end{align}
further define $\lambda_1=\sqrt{\frac{2}{c(\xi-1)}}\in(0,1)$  satisfying   $\frac{\xi-1}{2}+\frac{1}{c\lambda_1^2}=0$. The sign of the imaginary part Im$\theta_+$ is shown in Figure \ref{figg1}.
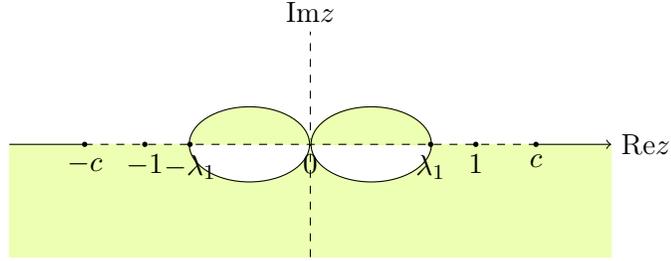
\begin{figure}[h]
	\centering
	\begin{tikzpicture}[node distance=2cm]
		\draw[lime!30, fill=lime!30] (-4,0)--(-4,-1.5)--(4,-1.5)--(4,0)--(0.8,0)--(1.6,0)arc (0:-180:0.8 and 0.5)--(0,0)--(0,0)arc (0:-180:0.8 and 0.5)--(-4,0);
		\draw[->](3,0)--(4,0)node[right]{ Re$z$};
		\draw[dashed](0,-1.5)--(0,1.5)node[above]{ Im$z$};
		\draw[lime!30, fill=lime!30] (0,0)--(0,0)arc (0:180:0.8 and 0.5)--(0,0);
		\draw [](-0.81,0) ellipse (0.8 and 0.5);
		\draw[lime!50, fill=lime!30] (0,0)--(0,0)arc (180:0:0.8 and 0.5)--(0,0);
		\draw [](0.81,0) ellipse (0.8 and 0.5);		
		\coordinate (I) at (0,0);
		\fill (I) circle (0pt) node[below] {$0$};
		\coordinate (a) at (1.6,0);
		\fill (a) circle (1pt) node[below] {$\lambda_1$};
		\coordinate (aa) at (-1.6,0);
		\fill (aa) circle (1pt) node[below] {$-\lambda_1$};
		\coordinate (b) at (2.2,0);
		\fill (b) circle (1pt) node[below] {$1$};
		\coordinate (ba) at (-2.2,0);
		\fill (ba) circle (1pt) node[below] {$-1$};
		\coordinate (c) at (3,0);
		\fill (c) circle (1pt) node[below] {$c$};
		\coordinate (ca) at (-3,0);
		\fill (ca) circle (1pt) node[below] {$-c$};
		\draw[dashed](-3,0)--(3,0);
		\draw[](-4,0)--(-3,0);
	\end{tikzpicture}
	\caption{\footnotesize  In the case $\xi>1+\frac{2}{c}$, Im$\theta_+(z)>0$ in yellow region while  Im$\theta_+(z)<0$ in white region.
		And critical line $ \text{Im }\theta(z)=0 $ is black solid line. }
	\label{figg1}
\end{figure}

Define
\begin{align}
	&\Sigma^{\pm}=\left\lbrace -1+e^{\pm\frac{3\pi i}{4}}\mathbb{R}^+\right\rbrace \cup\left\lbrace 1+e^{\pm\frac{\pi i}{4}}\mathbb{R}^+\right\rbrace, \\
 &\Omega^\pm =  \left\lbrace z; z=-1+e^{\pm\psi i}l,\ l\in\mathbb{R}^+,\  3\pi/4<\psi<\pi\right\rbrace \nonumber\\
	& \qquad\cup\left\lbrace z; z=1+e^{\pm\psi i}l,\ l\in\mathbb{R}^+,\ 0<\psi< {\pi }/{4}\right\rbrace, \\
	&\overline{\Omega}= \mathbb{C}\setminus(\Omega^+\cup\Omega^-).
\end{align}
Obviously,
\begin{align}
	&p_--\theta_+=\mathcal{O}\left(z^{-1}\right) , \text{ as }z\to\infty,\\
	&p_--\theta_+=-\frac{ci}{2}(\xi-1)+\frac{i}{2c^2}-\frac{i\xi}{2}+\mathcal{O}(z^2),\text{ as }z\to0\in\mathbb{C}^+.\label{the+0}
\end{align}
So we can use $\theta_+$ to replace $p_-$ in the exponential function. And we will utilize  these factorizations to deform the jump contours, so that the oscillating factor $e^{\pm2it\theta_+}$ are decaying in corresponding region respectively.
\begin{figure}[h]
	\centering
	\begin{tikzpicture}[node distance=2cm]
		\draw[lime!50, fill=lime!30] (-4,0)--(-4,-1.5)--(4,-1.5)--(4,0)--(0.8,0)--(1.6,0)arc (0:-180:0.8 and 0.5)--(0,0)--(0,0)arc (0:-180:0.8 and 0.5)--(-4,0);
		\draw[->](3,0)--(4,0)node[right]{ Re$z$};
		\draw[dashed](0,-1.5)--(0,1.5)node[above]{ Im$z$};
		\draw[lime!50, fill=lime!30] (0,0)--(0,0)arc (0:180:0.8 and 0.5)--(0,0);
		\draw [](-0.81,0) ellipse (0.8 and 0.5);
		\draw[lime!50, fill=lime!30] (0,0)--(0,0)arc (180:0:0.8 and 0.5)--(0,0);
		\draw [](0.81,0) ellipse (0.8 and 0.5);		
		\coordinate (I) at (0,0);
		\fill (I) circle (0pt) node[below] {$0$};
		\coordinate (a) at (1.6,0);
		\fill (a) circle (1pt) node[below] {$\lambda_1$};
		\coordinate (aa) at (-1.6,0);
		\fill (aa) circle (1pt) node[below] {$-\lambda_1$};
		\coordinate (b) at (2.2,0);
		\fill (b) circle (1pt) node[below] {$1$};
		\coordinate (ba) at (-2.2,0);
		\fill (ba) circle (1pt) node[below] {$-1$};
		\coordinate (c) at (3,0);
		\fill (c) circle (1pt) node[below] {$c$};
		\coordinate (ca) at (-3,0);
		\fill (ca) circle (1pt) node[below] {$-c$};
		\draw[dashed](-3,0)--(3,0);
		\draw[](-4,0)--(-3,0);
		\draw(2.2,0)--(3.8,1.4)node[above]{$\Sigma^+ $};
		\draw(-2.2,0)--(-3.8,1.4)node[left]{$\Sigma^+$};
		\draw(-2.2,0)--(-3.8,-1.4)node[left]{$\Sigma^-$};
		\draw(2.2,0)--(3.8,-1.4)node[right]{$\Sigma^-$};
		\draw[-latex](-3.8,-1.4)--(-3,-0.7);
		\draw[-latex](-3.8,1.4)--(-3,0.7);
		\draw[-latex](2.2,0)--(3,0.7);
		\draw[-latex](2.2,0)--(3,-0.7);
		\coordinate (C) at (-0.2,2.2);
		\coordinate (D) at (3.2,0.4);
		\fill (D) circle (0pt) node[right] {\footnotesize $\Omega^+$};
		\coordinate (J) at (-3.2,-0.4);
		\fill (J) circle (0pt) node[left] {\footnotesize $\Omega^-$};
		\coordinate (k) at (-3.2,0.4);
		\fill (k) circle (0pt) node[left] {\footnotesize $\Omega^+$};
		\coordinate (k) at (3.2,-0.4);
		\fill (k) circle (0pt) node[right] {\footnotesize $\Omega^-$};
		\coordinate (v) at (0.2,0.8);
		\fill (v) circle (0pt) node[right] {\footnotesize $\overline{\Omega} $};
	\end{tikzpicture}
	\caption{\footnotesize  The region of $\Omega^\pm$ and $\overline{\Omega}$. And $\Sigma^{(1)}=\Sigma^+\cup\Sigma^-\cup[-c,c]$.
The same as Figure \ref{figg1}, yellow region means Im$\theta_+(z)>0$ while white region means Im$\theta_+(z)<0$. }
	\label{figfj1}
\end{figure}
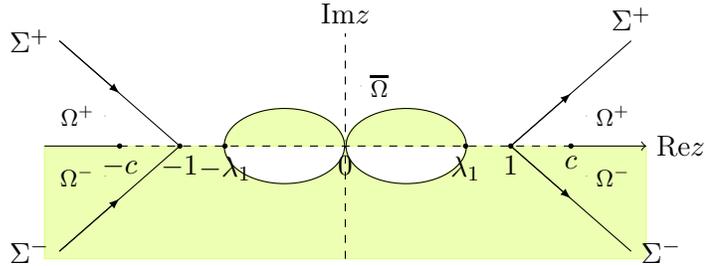

Similarly as in the above section, in this region of $\xi$, we introduce
a piecewise matrix interpolation function
\begin{equation}
G(z)=	G(z;\xi,c)=\left\{ \begin{array}{ll}
		\left(\begin{array}{cc}
			1 &\frac{ r_2e^{-2it\theta_+}}{1-r_1r_2}\\
			0 & 1
		\end{array}\right),   &\text{as } z\in\Omega^+;\\[12pt]
		\left(\begin{array}{cc}
			1 &0\\
			\frac{r_1e^{2it\theta_+}}{1-r_1r_2} & 1
		\end{array}\right),   &\text{as } z\in\Omega^-;\\
		I &\text{as } 	z \text{ in elsewhere},
	\end{array}\right..\label{funcG}
\end{equation}
Note that $1-r_1r_2=0$, so the matrix function $G(z)$ bring a new $-\frac{1}{4}$-singularity on $z=\pm c$.
We define  the new  matrix-valued   function $M^{(1)}(z)$
\begin{equation}
	M^{(1)}(z)\triangleq  M^{(1)}(z;y,t,c)=N(z)e^{it(p_--\theta_+)\sigma_3}G(z),\label{transm1}
\end{equation}
which then satisfies the following RH problem.

\begin{RHP}\label{RHP3}
	Find a matrix-valued function  $  M^{(1)}(z )$ which satisfies
	
	$\blacktriangleright$ Analyticity: $M^{(1)}(z)$ is meromorphic in $\mathbb{C}\setminus \left( \Sigma^{(1)}\cup\mathbb{R}\right) $, where
	\begin{equation}
		\Sigma^{(1)}=\Sigma_1\cup\Sigma_2\cup[-c,c] ,
	\end{equation}
	shown in Figure \ref{figfj1};

	$\blacktriangleright$ Symmetry: $M^{(1)}(z)=\sigma_2M^{(1)}(-z)\sigma_2$=$\sigma_1\overline{M^{(1)}(\bar{z})}\sigma_1$;
	
	$\blacktriangleright$ Jump condition: $M^{(1)}$  satisfies the jump condition
	\begin{equation}
		M^{(1)}_+(z)=M^{(1)}_-(z)V^{(1)}(z),\hspace{0.5cm}z \in \Sigma^{(1)}\cup\mathbb{R}, \  z\in \Sigma^{(1)}\cup\mathbb{R},
	\end{equation}
	where
	\begin{equation}
		V^{(1)}(z)=\left\{ \begin{array}{ll}
			\left(\begin{array}{cc}
				1 &\frac{ -r_2e^{-2it\theta_+}}{1-r_1r_2}\\
				0 & 1
			\end{array}\right),&\text{as } z\in\Sigma_1,\\[14pt]
			\left(\begin{array}{cc}
				1 &0\\
				\frac{r_1e^{2it\theta_+}}{1-r_1r_2} & 1
			\end{array}\right),&\text{as } z\in\Sigma_2,\\[14pt]
			(1-r_1r_2)^{\sigma_3},   &\text{as } z\in\mathbb{R}\setminus[-c,c],\\[6pt]
			\left(\begin{array}{cc}
				0 & -r_2(z-0i)\\
				r_1(z+0i) & 0
			\end{array}\right),   &\text{as } z\in[-c,c]\setminus[-1,1] ,\\[12pt]
			\left(\begin{array}{cc}
				0 & i\\
				i & 0
			\end{array}\right),   &\text{as } z\in[-1,1],
		\end{array}\right.;
	\end{equation}
	
	$\blacktriangleright$ Asymptotic behaviors: $M^{(1)}(z) = I+\mathcal{O}(z^{-1}),\hspace{0.5cm}z \rightarrow \infty;$

	$\blacktriangleright$ Singularity: $M^{(1)}(z)$ has singularity at $z=\pm1,\pm c$ with:
	\begin{align}
		&M^{(1)}(z)\sim \left(\mathcal{O}(z\mp1)^{1/4},\mathcal{O}(z\mp1)^{-1/4} \right) ,\ z\to \pm1\text{ in }\mathbb{C}^\pm,\\
		&M^{(1)}(z)\sim \left(\mathcal{O}(1),\mathcal{O}(z\mp c)^{-1/2} \right) ,\ z\to \pm c\text{ in }\mathbb{C}^+,\\
		&M^{(1)}(z)\sim \left(\mathcal{O}(z\mp c)^{-1/2},\mathcal{O}(1) \right) ,\ z\to \pm c\text{ in }\mathbb{C}^-.
	\end{align}
\end{RHP}

To deal with the jump on $\mathbb{R}$, we give a introduction of an auxiliary function $\delta(z)=\delta(z;\xi,c)$, which relies on $\xi$ and admits the following jump condition:
\begin{align*}
	&\delta_-(z)=\delta_+(z)(1-r_1r_2),\hspace*{0.5cm}z\in \mathbb{R}\setminus[-c,c];\\
	&\delta_-(z)\delta_+(z)=i[r_2]_-,\hspace*{1.2cm}z\in[-c,c] \setminus[-1,1];\\
	&\delta_-(z)\delta_+(z)=1,\hspace*{2cm}z\in[-1,1].
\end{align*}
Define the function
\begin{align}
	\log \delta(z)=&\frac{X(z)}{2\pi i}\left(\int_{-c}^{-1}+\int_{c}^{1} \right) \dfrac{\log(i[r_2]_-(s))}{(s-z)[X]_+(s)}ds -\frac{X(z)}{2\pi i}\int_{\mathbb{R}\setminus[-c,c]}\dfrac{\log(1-r_1(s)r_2(s))}{(s-z)X(s)}ds. \nonumber
\end{align}
\begin{Proposition}\label{prode1}
	The scalar function $\delta(z)$ satisfies the following properties \\
	(a) $\delta(z)$ is analytic on $\mathbb{C}\setminus\mathbb{R}$;\\
	(b) $\delta(z)$ has singularity at $z=\pm c$ with
	\begin{align}
		\delta(z)=\mathcal{O}(z-p)^{\mp 1/4},\hspace{0.3cm}z\to p\in \mathbb{C}^\pm\setminus\mathbb{R} ,\hspace{0.3cm}p=c,-c.
	\end{align}
	(c) As $z\to\infty\in\mathbb{C}\setminus\mathbb{R}$, $\delta(z)$ has limit  $\delta(\infty)$ with
	\begin{align}
		\log\delta(\infty)=&-\frac{1}{2\pi i}\left(\int_{-c}^{-1}+\int_{c}^{1} \right) \dfrac{\log(i[r_2]_-(s))}{[X]_+(s)}ds\nonumber\\
		&+\frac{1}{2\pi i}\int_{\mathbb{R}\setminus[-c,c]}\dfrac{\log(1-r_1(s)r_2(s))}{X(s)}ds.
	\end{align}
	(d) As $z\to0\in\mathbb{C}^+$,
	\begin{align}
		\delta(z)=&\exp\left\lbrace I_{\delta}^1\right\rbrace\cdot \left( 1+zI_{\delta}^2\right)
		+\mathcal{O}(z^2) .
	\end{align}
	Here,
	\begin{align}
		I_{\delta}^1=& \frac{c}{2\pi }\left(\int_{-c}^{-1}+\int_{c}^{1} \right) \dfrac{\log(i[r_2]_-(s))}{s[X]_+(s)}ds-\frac{c}{2\pi }\int_{\mathbb{R}\setminus[-c,c]}\dfrac{\log(1-r_1(s)r_2(s))}{sX(s)}ds,\label{Ide1}\\
		I_{\delta}^2=&\frac{c}{2\pi }\left(\int_{-c}^{-1}+\int_{c}^{1} \right) \dfrac{\log(i[r_2]_-(s))}{s^2[X]_+(s)}ds-\frac{c}{2\pi }\int_{\mathbb{R}\setminus[-c,c]}\dfrac{\log(1-r_1(s)r_2(s))}{s^2X(s)}ds.
	\end{align}
\end{Proposition}
\begin{proof}
	The proof of (a), (c) and (d) is trivial. And the proof of (b) is similar as \cite{lenells} and Appendix C in \cite{MonvelCMP2}.
\end{proof}
We give a new  transformation
\begin{align}
	M^{(2)}=\delta(\infty)^{-\sigma_3}M^{(1)}\delta^{\sigma_3}\label{trans2},
\end{align}
which then satisfies the following RH problem.

\begin{RHP}\label{RHP4}
	Find a matrix-valued function  $  M^{(2)}(z )$ which satisfies
	
	$\blacktriangleright$ Analyticity: $M^{(2)}(z)$ is meromorphic in $\mathbb{C}\setminus \Sigma^{(1)}$, where
	\begin{equation}
		\Sigma^{(1)}=\Sigma^+\cup\Sigma^-\cup[-c,c] ,
	\end{equation}
	shown in Figure \ref{figfj1};

	$\blacktriangleright$ Symmetry: $M^{(2)}(z)=\sigma_2M^{(2)}(-z)\sigma_2$=$\sigma_1\overline{M^{(2)}(\bar{z})}\sigma_1$;
	
	$\blacktriangleright$ Jump condition: $M^{(2)}$ has continuous boundary values $M^{(2)}_\pm(z)$ on $\Sigma^{(1)}$ and
	\begin{equation}
		M^{(2)}_+(z)=M^{(2)}_-(z)V^{(2)}(z),\hspace{0.5cm}z \in \Sigma^{(1)},
	\end{equation}
	where
	\begin{equation}
		V^{(2)}(z)=\left\{ \begin{array}{ll}
			\left(\begin{array}{cc}
				1 &\frac{ -r_2\delta^{-2}e^{-2it\theta_+}}{1-r_1r_2}\\
				0 & 1
			\end{array}\right),&\text{as } z\in\Sigma^+,\\[14pt]
			\left(\begin{array}{cc}
				1 &0\\
				\frac{r_1\delta^2e^{2it\theta_+}}{1-r_1r_2} & 1
			\end{array}\right),&\text{as } z\in\Sigma^-,\\[14pt]
			\left(\begin{array}{cc}
				0 & i\\
				i & 0
			\end{array}\right),   &\text{as } z\in[-c,c],
		\end{array}\right.
	\end{equation}
	
	$\blacktriangleright$ Asymptotic behaviors: $M^{(2)}(z) = I+\mathcal{O}(z^{-1}),\hspace{0.2cm}z \rightarrow \infty;$

	$\blacktriangleright$ Singularity: $M^{(2)}(z)$ has singularity at $z=\pm c$ with
	\begin{align}
		&M^{(2)}(z)\sim \mathcal{O}(z\mp c)^{-1/4},\ z\to \pm c\text{ in }\mathbb{C}\setminus\mathbb{R}.
	\end{align}
\end{RHP}
The jump matrix exponentially  decays to the identity matrix $I$ as $t\to\infty$ on $\Sigma^\pm_1$, which finally leads to the model RH problem.
\begin{RHP}\label{modelc}
	Find a matrix-valued function  $  M^{modc}(z )$ which satisfies 	
	
	$\blacktriangleright$ Analyticity: $  M^{modc}(z )$ is holomorphic in  $\mathbb{C}\setminus[-c,c]$;
	
	$\blacktriangleright$ Jump condition: $M^{modc}$ has continuous boundary values $M^{modc}_\pm(z)$ and
	\begin{equation}
		M^{modc}_+(z)=M^{modc}_-(z)V^{modc}(z),\hspace{0.5cm}z \in[-c,c],
	\end{equation}
	where
	\begin{equation}
		V^{modc}(z)=\left(\begin{array}{cc}
			0 & i\\
			i & 0
		\end{array}\right), z\in[-c,c];
	\end{equation}

	$\blacktriangleright$ Asymptotic behaviors
	\begin{align}
		&M^{modc}(z) = I+\mathcal{O}(z^{-1}),\hspace{0.5cm}z \rightarrow \infty;
	\end{align}

	$\blacktriangleright$ Singularity: $M^{modc}(z)$ has singularity at $z=\pm c$ with
	\begin{align}
		&M^{modc}(z)\sim \mathcal{O}(z\mp c)^{-1/4} ,\ z\to \pm c\text{ in }\mathbb{C}\setminus\mathbb{R}.
	\end{align}
\end{RHP}
We can construct the solution of model RH problem
\begin{equation}
	M^{modc}(z)=\frac{1}{\sqrt{2}}\left(\begin{array}{cc}
		\phi_c(z)+\phi_c(z)^{-1} & \phi_c(z)-\phi_c(z)^{-1}\\
		\phi_c(z)-\phi_c(z)^{-1} & \phi_c(z)+\phi_c(z)^{-1}
	\end{array}\right).
\end{equation}
As $z\to 0\in\mathbb{C}^+$,
\begin{equation}
	M^{modc}(z)=\sqrt{2}\left(\begin{array}{cc}
		0 & i\\
		i & 0
	\end{array}\right)+\frac{zi}{\sqrt{2}c}I+\mathcal{O}(z^2).
\end{equation}
Considering transformation
\begin{equation}
	E(z)=M^{(2)}(M^{modc}(z))^{-1},
\end{equation}
which has  jump matrix exponentially  decaying to the identity matrix $I$ as $t\to\infty$ on $\Sigma^\pm_1$. Then its    existence and uniqueness  can be shown  by  a  small-norm RH problem with
\begin{equation}
	E(z)=I+\mathcal{O}(t^{-2}).\label{asyEr1}
\end{equation}

\section{The first-type genus-2 elliptic wave  region     } \label{sec5}

 In  the  Region III, we need to introduce a new g-function  defined on genus 2
Riemann surface which has real  branch points $\pm 1$, $\pm c$ and $\pm z_0$ with $1<z_0<c$. And the  canonical homology basis $\left\lbrace a_j,b_j \right\rbrace_{j=1}^2 $ is shown in Figure \ref{figab}. Note that this region    contains  two  cases   \\
$3/4<\xi<1$:\ \ \ (1) $ 2<c^2<4$, $ 1-\frac{2(c^2-2)}{c^4}<\xi<1$;  \quad  \ \   (2)  $c^2>4$, $\xi_m<\xi<1$;\\
$1<\xi<+\infty$: \ (3)  $c^2<2$, $1+\frac{2(2-c^2)}{c^4}<\xi < 1+2/c $;\quad  (4)  $c^2>2$, $ 1< \xi < 1+2/c $. \\
In this two different cases, $g$ has different property. So after we proving the basic property of $g$, we will discuss this two different cases separately. Here, the definition of $\xi_m$ is means the
critical condition of $\xi$ that stationary phase point $z_2$ of the $g$-function in case (1) merge $c$. The existence of $\xi_m$ is given in the  Appendix \ref{xim}.
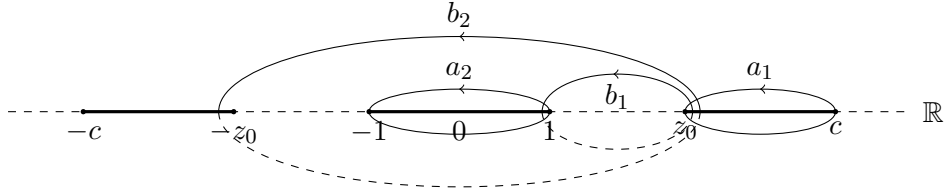
\begin{figure}[h]
	\centering
	\begin{tikzpicture}[node distance=2cm]
		\draw[dashed](-6,0)--(6,0)  node  at (6.3,0)  {$\mathbb{R}$};
		\draw(3.2,0)arc (0:180:3.2 and 1);
		\draw[dashed](-3.2,0)arc (180:360:3.2 and 1);
		\draw(1.1,0)arc (180:0:1 and 0.5);
		\draw [->](0.01,1)--(0,1);
		\draw[dashed](3.1,0)arc (0:-180:1 and 0.5);
		\draw [->](2.06,0.5)--(2.05,0.5);
		\draw [](4,0) ellipse (1 and 0.3);
		\draw [->](4.01,0.3)--(4,0.3);
		\draw [](0,0) ellipse (1.2 and 0.3);
		\draw [->](0.01,0.3)--(0,0.3);	
		\coordinate (a1) at (0,0.3);
		\fill (a1) circle (0pt) node[above] {$a_2$};
		\coordinate (a2) at (4,0.3);
		\fill (a2) circle (0pt) node[above] {$a_1$};
		\coordinate (b1) at (0,1);
		\fill (b1) circle (0pt) node[above] {$b_2$};
		\coordinate (r) at (2.1,0.5);
		\fill (r) circle (0pt) node[below] {$b_1$};	
		\coordinate (I) at (0,0);
		\fill (I) circle (0pt) node[below] {$0$};
		\coordinate (a) at (3,0);
		\fill (a) circle (1pt) node[below] {$z_0$};
		\coordinate (aa) at (-3,0);
		\fill (aa) circle (1pt) node[below] {$-z_0$};
		\coordinate (b) at (1.2,0);
		\fill (b) circle (1pt) node[below] {$1$};
		\coordinate (ba) at (-1.2,0);
		\fill (ba) circle (1pt) node[below] {$-1$};
		\coordinate (c) at (5,0);
		\fill (c) circle (1pt) node[below] {$c$};
		\coordinate (ca) at (-5,0);
		\fill (ca) circle (1pt) node[below] {$-c$};
		\draw[very thick](-5,0)--(-3,0);
		\draw[very thick](5,0)--(3,0);
		\draw[very thick](-1.2,0)--(1.2,0);
	\end{tikzpicture}
	\caption{\footnotesize  The  canonical homology basis $\left\lbrace a_j,b_j \right\rbrace_{j=1}^2 $ of the genius 2 Riemann surface. }
	\label{figab}
\end{figure}
 \subsection{Constructing the $g$-function  }
 To construct the g-function, we first introduce:
\begin{align}
	Y(z)=\left[\dfrac{z^2-z_0^2}{(z^2-1)(z^2-c^2)} \right]^{1/2},
\end{align}
where $z_0$ is a real number in $(1,c)$ and the branch of the square root is such that $Y(z)\in i\mathbb{R}^+$ for $z\in[z_0,c]$. And the d$g$ is the  derivative of g-function
\begin{align}
	\text{d}g=\dfrac{	Y(z)}{z^3}\left[ \frac{1-\xi}{2}z^4-\frac{c}{z_0}\left( 1+\frac{1}{c^2}-\frac{1}{z_0^2}\right) z^2+\frac{2c}{z_0}\right] \text{d}z.
\end{align}
Here, $\text{d}\hat{g}$ is a meromorphic differential defined on the 2-genus Riemann surface, with d$g$ on the upper sheet and $-$d$g$ on the lower sheet.
Similarly, simply calculation shows that
\begin{align*}
	&\frac{\text{d}g}{\text{d}z}-\frac{\text{d}p_-}{\text{d}z}=\mathcal{O}(  z^{-2}),\text{ as }z\to\infty;\\
	&\frac{\text{d}g}{\text{d}z}-\frac{\text{d}p_-}{\text{d}z}=\mathcal{O}(z),\text{ as }z\to0\in\mathbb{C}^+.
\end{align*}
Denote $\Sigma^{mod}$ as the union of three branch cuts:
\begin{align}
	\Sigma^{mod}=[-c,-z_0]\cup[-1,1]\cup[z_0,c].
\end{align}
Thus  the g-function is given by
\begin{align}
	g(z)=g(z;\xi,c)=\int_{c}^{z}\text{d}g,\ \ z\in\mathbb{C}\setminus\Sigma^{mod}.
\end{align}

\begin{Proposition}\label{prog}	
		There exist  a real number $z_0=z_0(\xi,c)$  in $(1,c)$ such that the function $g(z)$ defined above has the following properties
\begin{itemize}
	\item[{\rm (a)}]  The $a$-period of $g(z)$ is zero and the $b$-period of $g(z)$ is real;

	\item[{\rm (b)}] $g(z)$ satisfies the following jump conditions across $[-c,c]$:
		\begin{align}
			&g_-(z)+g_+(z)=0,\hspace{0.5cm}z\in(z_0,c),\\
			&g_-(z)-g_+(z)=0,\hspace{0.5cm}z\in(1,z_0)\cup(-z_0,-1),\\
			&g_-(z)+g_+(z)=B_1,\hspace{0.5cm}z\in(-1,1),\\
			&g_-(z)+g_+(z)=B_2,\hspace{0.5cm}z\in(-c,-z_0),
		\end{align}
	where  $B_j=B_j(\xi,c)=\frac{1}{2}\oint_{b_j}dg$ is real;

	\item[{\rm (c)}] $g(z)$ has another phase point $z_1=z_1(\xi)\in(z_0,c)$ which is the solution of equation $\frac{\xi-1}{2}z^4+\frac{c}{z_0}\left( 1+\frac{1}{c^2}-\frac{1}{z_0^2}\right) z^2-\frac{2c}{z_0}=0$;

	\item[{\rm (a)}] In Case (1), (2) with  $c^2>2$, $1>\xi$, $g(z)$ has another phase point $z_2=z_2(\xi)\in(c,+\infty)$, $z_2>z_1$ , which also is  a  solution of equation $\frac{\xi-1}{2}z^4+\frac{c}{z_0}\left( 1+\frac{1}{c^2}-\frac{1}{z_0^2}\right) z^2-\frac{2c}{z_0}=0$. When $c>2$, as $\xi\to \xi_m<1-\frac{2(c^2-2)}{c^4}$, $z_2(\xi)$ decreases to $c$.
\end{itemize}
\end{Proposition}

\begin{proof}
From the symmetry of d$g$, $a_2$-period of $g(z)$ is zero. Rewrite the function $	Y(z)$ as $Y(z;z_0)$. Let $F(s)$ be a function defined on $\mathbb{R}$ with
\begin{align}
	F(s)=\int_{s}^{c}\frac{[Y]_+(z;s)}{z^3}\left[ \frac{\xi-1}{2}z^4+\frac{c}{s}\left( 1+\frac{1}{c^2}-\frac{1}{s^2}\right) z^2-\frac{2c}{s}\right] \text{d}z.
\end{align}
Then we have $F(c)=0$ and
$$F(1)=\int_{1}^{c}\frac{1}{z^3}\left[ \left(z^2-c^2 \right)^{-1/2}\right] _+\left(  \frac{\xi-1}{2}z^4+\frac{z^2}{c}-2c\right)  \text{d}z=-\theta_+(1^+).$$
And we consider the $s$-derivative of  $F$  at $s=c$,
\begin{align}
	\frac{\text{d}F}{\text{d}s}(c)=-\frac{i}{c^3}\left(c^2-1\right)^{-1/2}\left(  \frac{\xi-1}{2}c^4+c^2-2\right).
\end{align}
In the case $\xi<1$, obviously, $F(1)\in i\mathbb{R}^-$. Thus, when $c^2>2$, $\xi>-\frac{2(c^2-2)}{c^4}+1$, $\frac{\text{d}F}{\text{d}s}(c)\in i\mathbb{R}^-$.
And in the case  $1<\xi<\frac{2}{c}+1$, from the property of $\theta_+$ in above section, we have that $F(1)\in i\mathbb{R}^+$. While when $c^2<2$, $1+\frac{2(2-c^2)}{c^4}<\xi<\frac{2}{c}+1$ and $c^2>2$, $1<\xi<\frac{2}{c}+1$, we have that $\frac{\text{d}F}{\text{d}s}(c)\in i\mathbb{R}^+$. So there must exist $z_0\in(1,c)$, such that $F(z_0)=0$. This also implies that there exist $z_1\in(z_0,c) $ such  that $f(z_1^2)=0$ with $$f(x)=\frac{\xi-1}{2}x^2+\frac{c}{z_0}\left( 1+\frac{1}{c^2}-\frac{1}{z_0^2}\right) x-\frac{2c}{z_0}.$$
By simply calculating  the $a_1$-period of $g(z)$ can be  zero and  the both $b$-period are real. Obviously, $f(0)<0$. So in the $\xi>1$ case, $f(x)$ only has one zero $z_1$ on $\mathbb{R}^+$.
And in the $\xi<1$ case, we note another real solution of $f(z)=0$ is $z_2^2$.

In addition, in the $\xi<1$ case, simple calculation gives that
\begin{align}
	&\dfrac{\partial 	F}{\partial s}(s;\xi)=-\int_{s}^{c}\frac{z(1-\xi)}{2\sqrt{(z^2-1)(z^2-c^2)(z^2-s^2)}s^3}(s^2-z_1^2)(s^2-z_2^2)dz,\\
	&\dfrac{\partial 	F}{\partial  \left( 1-\xi\right)/2 }(s;\xi)=\int_{s}^{c}\frac{z\sqrt{z^2-s^2}}{\sqrt{(z^2-1)(z^2-c^2)}}dz.
\end{align}
So when $\xi$ decrease from $1$, $z_0(\xi)$ increase in $(1,c)$ while  $z_2$ as a solution of $f(x)=0$ decrease. When $z_2$ merge $c$, we denote this critical condition as $\xi_m$.
\end{proof}
Next, because $g$ will have different sign table in $\xi>1$ and $\xi<1$, we will discuss $g$-function according to it. Denote constant
\begin{align}
	g(\infty)=\lim_{z\to \infty}g(z)-p_-(z).\label{ginf}
\end{align}

\subsection{ Opening the jump in the region  $1<\xi <+\infty$ }

\quad In this region, we give the signature table  of Im$g$ is given in Figure \ref{figdg}.
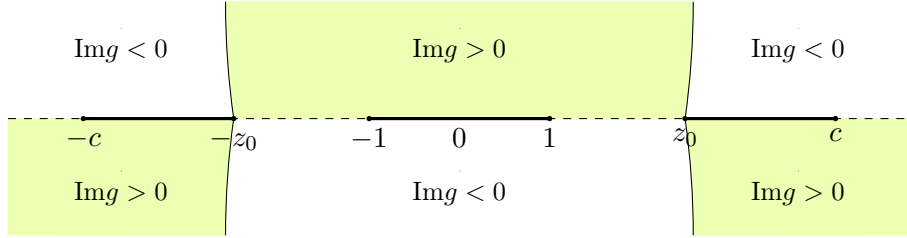
\begin{figure}[h]
	\centering
	\begin{tikzpicture}[node distance=2cm]
		\draw[lime!30, fill=lime!30](-6,0)--(-3,0)arc (165:180:3.2 and 6)--(-6,-1.55)--(-6,0);
		\draw[lime!30, fill=lime!30](6,0)--(3,0)arc (15:0:3.2 and 6)--(6,-1.55)--(6,0);
		\draw[lime!30, fill=lime!30](3,0)arc (-15:0:3.2 and 6)--(-3,1.55)--(-3,0)--(3,0);
		\draw[lime!30, fill=lime!30](-3,0)arc (195:180:3.2 and 6)--(3,1.55)--(3,0)--(-3,0);
		\draw[dashed](-6,0)--(6,0);
		\draw(3,0)arc (-15:0:3.2 and 6);
		\draw(3,0)arc (15:0:3.2 and 6);
		\draw[](-3,0)arc (165:180:3.2 and 6);
		\draw[](-3,0)arc (195:180:3.2 and 6);
		\coordinate (a1) at (0,1.2);
		\fill (a1) circle (0pt) node[below] {\small $\text{Im}g>0$};
		\coordinate (e1) at (-4.5,-0.7);
		\fill (e1) circle (0pt) node[below] {\small $\text{Im}g>0$};
		\coordinate (o1) at (4.5,-0.7);
		\fill (o1) circle (0pt) node[below] {\small $\text{Im}g>0$};	
		\coordinate (w1) at (0,-0.7);
		\fill (w1) circle (0pt) node[below] {\small $\text{Im}g<0$};
		\coordinate (e2) at (-4.5,1.2);
		\fill (e2) circle (0pt) node[below] {\small $\text{Im}g<0$};
		\coordinate (o2) at (4.5,1.2);
		\fill (o2) circle (0pt) node[below] {\small $\text{Im}g<0$};				
		\coordinate (I) at (0,0);
		\fill (I) circle (0pt) node[below] {$0$};
		\coordinate (a) at (3,0);
		\fill (a) circle (1pt) node[below] {$z_0$};
		\coordinate (aa) at (-3,0);
		\fill (aa) circle (1pt) node[below] {$-z_0$};
		\coordinate (b) at (1.2,0);
		\fill (b) circle (1pt) node[below] {$1$};
		\coordinate (ba) at (-1.2,0);
		\fill (ba) circle (1pt) node[below] {$-1$};
		\coordinate (c) at (5,0);
		\fill (c) circle (1pt) node[below] {$c$};
		\coordinate (ca) at (-5,0);
		\fill (ca) circle (1pt) node[below] {$-c$};
		\draw[very thick](-5,0)--(-3,0);
		\draw[very thick](5,0)--(3,0);
		\draw[very thick](-1.2,0)--(1.2,0);
	\end{tikzpicture}
	\caption{\footnotesize  In the purpler region, Im$g$>0 while in the white region, Im$g$<0. }
	\label{figdg}
\end{figure}
To open the  jump  contour $\mathbb{R}$,  we define
\begin{align*}
	\Sigma^{\pm}_1=&\left\lbrace z= -z_1+e^{\pm\frac{3\pi i}{4}}\mathbb{R}^+\right\rbrace \cup\left\lbrace z= z_1+e^{\pm\frac{\pi i}{4}}\mathbb{R}^+\right\rbrace ,\\
	\Sigma_{2}^\pm=&\left\lbrace z=-z_0+e^{\pm\psi i}l,\ l\in(0,\frac{z_0-1}{2\cos\psi})\right\rbrace \cup\left\lbrace z=z_0+e^{(\pi\mp\psi) i}l,\ l\in(0,\frac{z_0-1}{2\cos\psi})\right\rbrace\nonumber\\
	&\cup\left\lbrace z=1+e^{\psi i}l,\ l\in(0,\frac{z_0-1}{2\cos\psi})\right\rbrace \cup\left\lbrace z=-1+e^{(\pi-\psi) i}l,\ l\in(0,\frac{z_0-1}{2\cos\psi})\right\rbrace,
\end{align*}
where  $ \psi < \pi/4$  is chosen as a  small enough positive constant such that $\Sigma_{3}$ is  contained in the region of Im$g>0$. And further define
opened domains
\begin{align*}
	\Omega_1^\pm=&\left\lbrace z:  z=-z_1+e^{\pm \phi i}l,\ l\in\mathbb{R}^+,\ \pi>\phi>\frac{3\pi }{4}\right\rbrace \nonumber\\
	&\cup\left\lbrace z; z=z_1+e^{\pm \phi i}l,\ l\in\mathbb{R}^+,\ 0<\phi<\frac{\pi }{4}\right\rbrace ,\\
	\Omega_2^\pm=&\left\lbrace z:  z=-z_0+e^{\pm \phi i}l,\ l\in(0,\frac{z_0-1}{2\cos\psi}),\ 0<\phi<\psi\right\rbrace \nonumber\\
	&\cup\left\lbrace z:  z=z_0+e^{\pm\phi i}l,\ l\in(0,\frac{z_0-1}{2\cos\psi}),\ \pi>\phi>\pi\mp\psi\right\rbrace ,\nonumber\\
	&\cup\left\lbrace z:  z=1+e^{\pm\phi i}l,\ l\in(0,\frac{z_0-1}{2\cos\psi}),\ 0<\phi<\psi\right\rbrace \nonumber\\
	&\cup\left\lbrace z:  z=-1+e^{\pm\phi i}l,\ l\in(0,\frac{z_0-1}{2\cos\psi}),\ \pi>\phi>\pi\mp\psi\right\rbrace
\end{align*}
 Now we use $g$ to replace $p_-$ in the exponential function. And we will utilize  these factorizations to deform the jump contours, so that the oscillating factor $e^{\pm2it\theta_+}$ are decaying in corresponding region respectively.
\begin{figure}[h]
	\centering
	\begin{tikzpicture}[node distance=2cm]
		\draw[lime!30, fill=lime!30](-6,0)--(-3,0)arc (165:180:3.2 and 6)--(-6,-1.55)--(-6,0);
	\draw[lime!30, fill=lime!30](6,0)--(3,0)arc (15:0:3.2 and 6)--(6,-1.55)--(6,0);
	\draw[lime!30, fill=lime!30](3,0)arc (-15:0:3.2 and 6)--(-3,1.55)--(-3,0)--(3,0);
	\draw[lime!30, fill=lime!30](-3,0)arc (195:180:3.2 and 6)--(3,1.55)--(3,0)--(-3,0);
		\draw[dashed](0,-1.5)--(0,1.5)node[above]{ Im$z$};
		\draw[dashed](-6,0)--(6,0)node[right]{ Re$z$};
		\coordinate (I) at (0,0);
		\fill (I) circle (0pt) node[below] {$0$};
		\coordinate (a) at (3,0);
		\fill (a) circle (1pt) node[below] {$z_0$};
		\coordinate (aa) at (-3,0);
		\fill (aa) circle (1pt) node[below] {$-z_0$};
		\coordinate (b) at (1,0);
		\fill (b) circle (1pt) node[below] {$1$};
		\coordinate (ba) at (-1,0);
		\fill (ba) circle (1pt) node[below] {$-1$};
		\coordinate (c) at (5.5,0);
		\fill (c) circle (1pt) node[below] {$c$};
		\coordinate (ca) at (-5.5,0);
		\fill (ca) circle (1pt) node[below] {$-c$};
		\coordinate (y) at (4.5,0);
		\fill (y) circle (1pt) node[below] {$z_1$};
		\coordinate (cy) at (-4.5,0);
		\fill (cy) circle (1pt) node[below] {$-z_1$};
		\draw(4.5,0)--(5.7,1.4)node[above]{$\Sigma_1^+$};
		\draw(-4.5,0)--(-5.7,1.4)node[left]{$\Sigma_1^+$};
		\draw(-4.5,0)--(-5.7,-1.4)node[left]{$\Sigma_1^-$};
		\draw(4.5,0)--(5.7,-1.4)node[right]{$\Sigma_1^-$};
		\draw[-latex](-5.7,-1.4)--(-5.1,-0.7);
		\draw[-latex](-5.7,1.4)--(-5.1,0.7);
		\draw[-latex](4.5,0)--(5.1,0.7);
		\draw[-latex](4.5,0)--(5.1,-0.7);
		\draw(-3,0)--(-2,0.6)node[above]{$\Sigma_2^+$};
		\draw(-1,0)--(-2,0.6);
		\draw(3,0)--(2,0.6)node[above]{$\Sigma_2^+$};
		\draw(1,0)--(2,0.6);
		\draw(-3,0)--(-2,-0.6)node[below]{$\Sigma_2^-$};
		\draw(-1,0)--(-2,-0.6);
		\draw(3,0)--(2,-0.6)node[below]{$\Sigma_2^-$};
		\draw(1,0)--(2,-0.6);
		\draw[-latex](-2,0.6)--(-1.5,0.3);
		\draw[-latex](-3,0)--(-2.5,0.3);
		\draw[-latex](-2,-0.6)--(-1.5,-0.3);
		\draw[-latex](-3,0)--(-2.5,-0.3);
		\draw[-latex](1,0)--(1.5,0.3);
		\draw[-latex](2,0.6)--(2.5,0.3);
		\draw[-latex](1,0)--(1.5,-0.3);
		\draw[-latex](2,-0.6)--(2.5,-0.3);
		\coordinate (r) at (-2,0.02);
		\fill (r) circle (0pt) node[below] {\footnotesize$\Omega_2^-$};
		\coordinate (r1) at (2,0.02);
		\fill (r1) circle (0pt) node[below] {\footnotesize$\Omega_2^-$};
		\coordinate (hh) at (-2,-0.02);
		\fill (hh) circle (0pt) node[above] {\footnotesize$\Omega_2^+$};
		\coordinate (hr) at (2,-0.02);
		\fill (hr) circle (0pt) node[above] {\footnotesize$\Omega_2^+$};
		\coordinate (C) at (-0.2,2.2);
		\coordinate (D) at (5.2,0.6);
		\fill (D) circle (0pt) node[right] {\footnotesize $\Omega_1^+$};
		\coordinate (J) at (-5.2,-0.6);
		\fill (J) circle (0pt) node[left] {\footnotesize $\Omega_1^-$};
		\coordinate (k) at (-5.2,0.6);
		\fill (k) circle (0pt) node[left] {\footnotesize $\Omega_1^+$};
		\coordinate (k) at (5.2,-0.6);
		\fill (k) circle (0pt) node[right] {\footnotesize $\Omega_1^-$};
		\coordinate (v) at (0.2,1.2);
		\fill (v) circle (0pt) node[right] {\footnotesize $ \overline{\Omega} $};
		\draw[-latex](-3.73,0)--(-3.72,0);
		\draw[-latex](3.73,0)--(3.74,0);
		\draw[-latex](-0.01,0)--(0.01,0);
	\end{tikzpicture}
	\caption{\footnotesize  The opened jump contours $\Sigma_j^\pm, j=1,2$  and opened  regions    $\Omega_j^\pm$, $j=1,2$.
 The same as Figure \ref{figdg}, purple region means Im$g(z)>0$ while white region means Im$g(z)<0$. }
	\label{figfj2}
\end{figure}
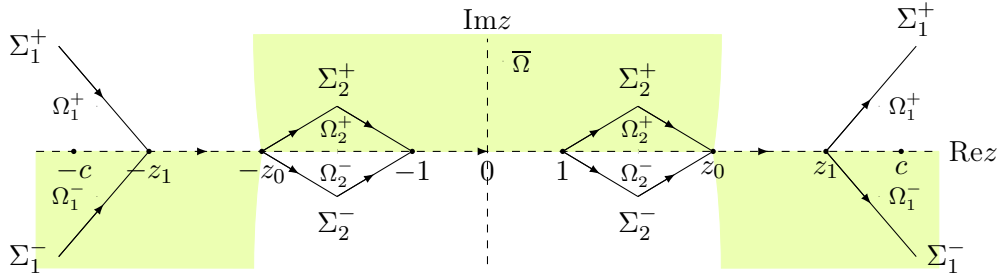

In this region of $\xi,c$, we introduce
a piecewise matrix interpolation function
\begin{equation}
	G(z)= G (z;\xi,c)=\left\{ \begin{array}{ll}
		\left(\begin{array}{cc}
			1 &\frac{ r_2e^{-2itg}}{1-r_1r_2}\\
			0 & 1
		\end{array}\right),   &\text{as } z\in\Omega_1^+;\\[12pt]
		\left(\begin{array}{cc}
			1 &0\\
			\frac{r_1e^{2itg}}{1-r_1r_2} & 1
		\end{array}\right),   &\text{as } z\in\Omega_1^-;\\[12pt]
	\left(\begin{array}{cc}
		1 & 0\\
		-r_1e^{2itg} & 1
	\end{array}\right),   &\text{as } z\in\Omega_2^+;\\[12pt]
	\left(\begin{array}{cc}
		1 & -r_2e^{-2itg}\\
		0 & 1
	\end{array}\right),   &\text{as } z\in\Omega_2^-;\\
		I &\text{as } 	z \text{ in elsewhere},
	\end{array}\right. \label{funcG2}
\end{equation}
Same as above section,  $ G (z)$ bring a new singularity.
We define  the new  matrix-valued   function $M^{(1)}(z)$
\begin{equation}
	M^{(1)}(z)\triangleq  e^{itg(\infty)\sigma_3}N(z)e^{it(p_--g)\sigma_3} G (z),\label{transm21}
\end{equation}
which then satisfies the following RH problem.

\begin{RHP}\label{RHP5}
	Find a matrix-valued function  $  M^{(1)}(z )$ which satisfies
	
	$\blacktriangleright$ Analyticity: $M^{(1)}(z)$ is meromorphic in $\mathbb{C}\setminus \left( \Sigma^{(1)}\cup(-\infty,-z_1)\cup(z_1,\infty)\right) $, where
	\begin{equation}
		\Sigma^{(1)}= \left(\cup_{j=1}^2\Sigma_j^\pm\right) \cup[-1,1]\cup[-z_1,-z_0]\cup[z_0,z_1],
	\end{equation}
	see  Figure \ref{figfj2};

	$\blacktriangleright$ Symmetry: $M^{(1)}(z)=\sigma_2M^{(1)}(-z)\sigma_2$=$\sigma_1\overline{M^{(1)}(\bar{z})}\sigma_1$;
	
	$\blacktriangleright$ Jump condition: $M^{(1)}$ has continuous boundary values $M^{(1)}_\pm(z)$ on the contour $ \Sigma^{(1)}\cup(-\infty,-z_1)\cup(z_1,\infty)$ and
	\begin{equation}
		M^{(1)}_+(z)=M^{(1)}_-(z)\hat{V}^{(1)}(z),\hspace{0.5cm}z \in  \Sigma^{(1)}\cup(-\infty,-z_1)\cup(z_1,\infty),
	\end{equation}
	where
	\begin{equation}
		\hat{V}^{(1)}(z)=\left\{ \begin{array}{ll}
			\left(\begin{array}{cc}
				1 &\frac{ -r_2e^{-2itg}}{1-r_1r_2}\\
				0 & 1
			\end{array}\right),&\text{as } z\in\Sigma_1^+,\\[14pt]
			\left(\begin{array}{cc}
				1 &0\\
				\frac{r_1e^{2itg}}{1-r_1r_2} & 1
			\end{array}\right),&\text{as } z\in\Sigma_1^-,\\[14pt]
			\left(\begin{array}{cc}
				1 & 0\\
				r_1e^{2itg} & 1
			\end{array}\right),   &\text{as } z\in\Sigma_2^+;\\[12pt]
			\left(\begin{array}{cc}
				1 & -r_2e^{-2itg}\\
				0 & 1
			\end{array}\right),   &\text{as } z\in\Sigma_2^-;\\[12pt]
		(1-r_1r_2)^{\sigma_3},   &\text{as } z\in\mathbb{R}\setminus[-c,c],\\[6pt]
		\left(\begin{array}{cc}
			0 & -r_2(z-0i)e^{-itB_2}\\
			r_1(z+0i)e^{itB_2} & 0
		\end{array}\right),   &\text{as } z\in[-c,-z_1] ,\\[12pt]
		\left(\begin{array}{cc}
		0 & -r_2(z-0i)e^{-itB_2}\\
		r_1(z+0i)e^{itB_2} & e^{it(g_--g_+)}
		\end{array}\right),   &\text{as } z\in[-z_1,-z_0] ,\\[12pt]
			\left(\begin{array}{cc}
			0 & ie^{-itB_1}\\
			ie^{itB_1} & 0
		\end{array}\right),   &\text{as } z\in[-1,1],\\[12pt]	
	\left(\begin{array}{cc}
		0 & -r_2(z-0i)\\
		r_1(z+0i) & e^{-2itg_+}
	\end{array}\right),   &\text{as } z\in[z_0,z_1] ,\\[12pt]
	\left(\begin{array}{cc}
				0 & -r_2(z-0i)\\
				r_1(z+0i) & 0
			\end{array}\right),   &\text{as } z\in[z_1,c],	
		\end{array}\right.;
	\end{equation}
	
	$\blacktriangleright$ Asymptotic behaviors:
	\begin{align}
		&M^{(1)}(z) = I+\mathcal{O}(z^{-1}),\hspace{0.5cm}z \rightarrow \infty;
	\end{align}
	
	$\blacktriangleright$ Singularity: $M^{(1)}(z)$ has singularity at $z=\pm1,\pm c$ with:
	\begin{align}
		&M^{(1)}(z)\sim \mathcal{O}(z\mp1)^{-1/4} ,\ z\to \pm1\text{ in }\mathbb{C}\setminus \Sigma^{(1)},\\
		&M^{(1)}(z)\sim \left(\mathcal{O}(1),\mathcal{O}(z\mp c)^{-1/2} \right) ,\ z\to \pm c\text{ in }\mathbb{C}^+,\\
		&M^{(1)}(z)\sim \left(\mathcal{O}(z\mp c)^{-1/2},\mathcal{O}(1) \right) ,\ z\to \pm c\text{ in }\mathbb{C}^-.
	\end{align}
\end{RHP}

To deal with the jump on $\mathbb{R}$, we give a introduction of an auxiliary function $D(z)$, which admits the following jump condition:
\begin{align*}
	&D_-(z)=D+(z)(1-r_1r_2),\hspace*{0.5cm}z\in \mathbb{R}\setminus[-c,c];\\
	&D_-(z)D_+(z)=i[r_2]_-,\hspace*{1.2cm}z\in[-c,c] \setminus[-z_0,z_0];\\
	&D_-(z)D_+(z)=1,\hspace*{2cm}z\in[-1,1].
\end{align*}
Define $Y_3(z)= (z^2-1)(z^2-c^2)Y(z),$ and
\begin{align}
	\log D(z)=&\frac{Y_3(z)}{2\pi i}\left(\int_{-c}^{-z_0}+\int_{c}^{z_0} \right) \dfrac{\log(i[r_2]_-(s))}{(s-z)[Y_3]_+(s)}ds\nonumber\\
	&-\frac{Y_3(z)}{2\pi i}\int_{\mathbb{R}\setminus[-c,c]}\dfrac{\log(1-r_1(s)r_2(s))}{(s-z)Y_3(s)}ds.
\end{align}
\begin{Proposition}\label{proD}
	The scalar function $D(z)$ satisfies the following properties  \\
	(a) $D(z)$ is analytic on $\mathbb{C}\setminus\left( (-\infty,-z_0)\cup[-1,1]\cup(z_0,\infty)\right) $;\\
	(b) $D(z)$ has singularity at $z=\pm c$ with:
	\begin{align}
		D(z)=\mathcal{O}(z-p)^{\mp 1/4},\hspace{0.3cm}z\to p\in \mathbb{C}^\pm\setminus\mathbb{R} ,\hspace{0.3cm}p=c,-c.
	\end{align}
	(c) As $z\to\infty\in\mathbb{C}\setminus\mathbb{R}$, $D(z)$ has limit  $D_\infty(z)$ with
	\begin{align}
		\log D_\infty(z)=&-\frac{1}{2\pi i}\left(\int_{-c}^{-z_0}+\int_{c}^{z_0} \right) \dfrac{\log(i[r_2]_-(s))}{[X]_+(s)}\left(z^2+zs+s^2-\frac{1+c^2+z_0^2}{2} \right) ds\nonumber\\
		&+\frac{1}{2\pi i}\int_{\mathbb{R}\setminus[-c,c]}\dfrac{\log(1-r_1(s)r_2(s))}{X(s)}\left(z^2+zs+s^2-\frac{1+c^2+z_0^2}{2} \right)ds.\nonumber
	\end{align}
(d) As $z\to0\in\mathbb{C}^+$,
\begin{align}
	D_\infty(z)=&D_\infty(0)\left(1 +D_\infty^{(1)}z\right) +\mathcal{O}(z^2),\\
	D(z)=&D(0)\left( 1+D^{(1)}z\right) +\mathcal{O}(z^2),
\end{align}
where
{\small\begin{align}
	&\log \left( D_\infty(0)\right) = \frac{1}{2\pi i}\left(\int_{-c}^{-z_0}+\int_{c}^{z_0} \right) \dfrac{\log(i[r_2]_-(s))}{[X]_+(s)}\left(\frac{1+c^2+z_0^2}{2} -s^2\right) ds\nonumber\\
	&\qquad \qquad\qquad+\frac{1}{2\pi i}\int_{\mathbb{R}\setminus[-c,c]}\dfrac{\log(1-r_1(s)r_2(s))}{X(s)}\left(s^2-\frac{1+c^2+z_0^2}{2} \right)ds,\nonumber\\[5pt]
&	D_\infty^{(1)}= -\frac{1}{2\pi i}\left(\int_{-c}^{-z_0}+\int_{c}^{z_0} \right) \dfrac{s\log(i[r_2]_-(s))}{[X]_+(s)}ds +\frac{1}{2\pi i}\int_{\mathbb{R}\setminus[-c,c]}\dfrac{s\log(1-r_1(s)r_2(s))}{X(s)}ds,\nonumber\\[5pt]
&	\log \left( D(0)\right) = \frac{cz_0}{2\pi }\left(\int_{-c}^{-z_0}+\int_{c}^{z_0} \right) \dfrac{\log(i[r_2]_-(s))}{s[Y_3]_+(s)}ds -\frac{cz_0}{2\pi }\int_{\mathbb{R}\setminus[-c,c]}\dfrac{\log(1-r_1(s)r_2(s))}{sY_3(s)}ds,\nonumber\\[5pt]
	&D^{(1)}= \frac{cz_0}{2\pi }\left(\int_{-c}^{-z_0}+\int_{c}^{z_0} \right) \dfrac{\log(i[r_2]_-(s))}{s^2[Y_3]_+(s)}ds -\frac{cz_0}{2\pi }\int_{\mathbb{R}\setminus[-c,c]}\dfrac{\log(1-r_1(s)r_2(s))}{s^2Y_3(s)}ds.\nonumber
\end{align}}
\end{Proposition}
By using $D(z)$,   we define  a new matrix function
\begin{align}
M^{(2)}=D_\infty^{-\sigma_3}M^{(1)}D^{\sigma_3}\label{trans22},
\end{align}
which then satisfies the following RH problem.

\begin{RHP}\label{RHP6}
	Find a matrix-valued function  $  M^{(2)}(z )$ which satisfies
	
	$\blacktriangleright$ Analyticity: $M^{(2)}(z)$ is meromorphic in $\mathbb{C}\setminus \Sigma^{(1)}$;

	$\blacktriangleright$ Symmetry: $M^{(2)}(z)=\sigma_2M^{(2)}(-z)\sigma_2$=$\sigma_1\overline{M^{(2)}(\bar{z})}\sigma_1$;
	
	$\blacktriangleright$ Jump condition: $M^{(2)}$ has continuous boundary values $M^{(2)}_\pm(z)$ on $\Sigma^{(1)}$ and
	\begin{equation}
		M^{(2)}_+(z)=M^{(2)}_-(z)V^{(2)}(z),\hspace{0.5cm}z \in \Sigma^{(1)},
	\end{equation}
	where
	\begin{equation}
		V^{(2)}(z)=\left\{ \begin{array}{ll}
			\left(\begin{array}{cc}
				1 &\frac{ -r_2D^{-2}e^{-2itg}}{1-r_1r_2}\\
				0 & 1
			\end{array}\right),&\text{as } z\in\Sigma_1^+,\\[14pt]
			\left(\begin{array}{cc}
				1 &0\\
				\frac{r_1D^2e^{2itg}}{1-r_1r_2} & 1
			\end{array}\right),&\text{as } z\in\Sigma_1^-,\\[14pt]
			\left(\begin{array}{cc}
				1 & 0\\
				r_1D^2e^{2itg} & 1
			\end{array}\right),   &\text{as } z\in\Sigma_2^+;\\[12pt]
			\left(\begin{array}{cc}
				1 & -r_2D^{-2}e^{-2itg}\\
				0 & 1
			\end{array}\right),   &\text{as } z\in\Sigma_2^-;\\[12pt]
			\left(\begin{array}{cc}
				0 & ie^{-itB_2}\\
				ie^{itB_2} & 0
			\end{array}\right),   &\text{as } z\in[-c,-z_1] ,\\[12pt]
			\left(\begin{array}{cc}
				0 & ie^{-itB_2}\\
				ie^{itB_2} & \frac{D_-}{D_+}e^{itB_2}e^{-2itg_+}
			\end{array}\right),   &\text{as } z\in[-z_1,-z_0] ,\\[12pt]
			\left(\begin{array}{cc}
				0 & ie^{-itB_1}\\
				ie^{itB_1} & 0
			\end{array}\right),   &\text{as } z\in[-1,1],\\[12pt]	
			\left(\begin{array}{cc}
				0 & i\\
				i & \frac{D_-}{D_+}e^{-2itg_+}
			\end{array}\right),   &\text{as } z\in[z_0,z_1] ,\\[12pt]
			\left(\begin{array}{cc}
				0 & i\\
				i & 0
			\end{array}\right),   &\text{as } z\in[z_1,c],	
		\end{array}\right.
	\end{equation}
	
	$\blacktriangleright$ Asymptotic behaviors: $M^{(2)}(z) = I+\mathcal{O}(z^{-1}),\hspace{0.2cm}z \rightarrow \infty;$
	
	$\blacktriangleright$ Singularity: $M^{(2)}(z)$ has singularity at $z=\pm c$ with
	\begin{align}
		&M^{(2)}(z)\sim \mathcal{O}(z\mp c)^{-1/4} ,\ z\to \pm c,\ \pm 1\text{ in }\mathbb{C}\setminus\mathbb{R}.
	\end{align}
\end{RHP}

Away from $\mathbb{R}$, the jump $\hat{V}^{(2)}(z)$ exponentially approaches the identity matrix as $t\to\infty$.
So we expect to only consider the jump on $\mathbb{R}$. To arrive at this goal, we
denote $ U(\xi)$ as the union set of neighborhood of $\pm z_0$:
\begin{equation}
	U(\xi)=U(\pm z_0),\ U(\pm z_0)= \left\lbrace z:|z\mp z_0|\leq \varrho \right\rbrace.
\end{equation}
Here, $\varrho$ is a small positive constant such that $\varrho<\min\left\lbrace \frac{z_0-1}{3}, \frac{z_1-z_0}{3} \right\rbrace $. In the case of \cite{YYLmch},
 $\theta(z_1)$ is in $\mathbb{R}$ such that as $t\to\infty$, the term $\theta''(z_1)(z-z_1)^2$ in the Taylor expansion of $\theta$ at $z=z_1$ is dominating. But in this paper, the phase point $\pm z_1$ is on the cut with Im$g(z_1)_+<0$ which means the exponential function in $\hat{V}^{(2)}(z)$  also decays exponentially on $\pm z_1$. In fact,  it is  also decays exponentially on $(z_0,z_1]\cup[-z_1,z_0)$.

Thus, the jump matrix $V^{(2)}(z)$    uniformly goes to  $I$  on     $\Sigma^{(1)}\setminus  U(\xi)$.
So outside the $U(\xi)$ there is only exponentially small error (in $t$) by completely ignoring the jump condition of  $M^{(2)}(z)$.
And this proposition enlightens  us to construct the solution $M^{(2)}(z)$ as follow
\begin{equation}
M^{(2)}(z)=M^{(2)}(z;\xi,c)=\left\{\begin{array}{ll}
		E(z;\xi,c)M^{mod}(z;\xi,c) & z\notin U(\xi)\\
		E(z;\xi,c)M^{lo,+}(z;\xi,c)  &z\in U(z_0)\\
		E(z;\xi,c)M^{lo,-}(z;\xi,c)  &z\in U(-z_0)\\
	\end{array}\right..\label{transm4}
\end{equation}
Here  $M^{mod}(z)$ is the model  RH problem  on the Riemann surface, which solution is given by theta function in Subsection \ref{secmod}. $M^{lo,\pm}(z)$ are local model of $\pm z_0$ which solution can be expressed in terms of Airy functions shown in Subsection \ref{seclo}. And  $E(z;\xi,c)$ is the error function, which will be discussed in subsection \ref{secer}.
 by the small-norm RH problem theory.

\subsubsection{Model RH problem on Riemann surface}\label{secmod}

\quad We  consider  the following model RH problem with  its  jump matrix   on $\mathbb{R}$.

\begin{RHP}\label{RHP7}
	Find a matrix-valued function $M^{mod}(z)$  with following identities:
	
	$\blacktriangleright$ Analyticity: $M^{mod}(z)$ is analytical  in $\mathbb{C}\setminus  \Sigma^{cut} $, with $$\Sigma^{cut}=[-c,-z_0]\cup[-1,1]\cup[z_0,c];$$

	$\blacktriangleright$ Asymptotic behaviors: $M^{mod}(z) \sim I+\mathcal{O}(z^{-1}),\hspace{0.2cm}|z| \rightarrow \infty;$

	$\blacktriangleright$ Jump condition: $M^{mod}(z)$  satisfies the jump relation
	$$M^{mod}_+(z)=M^{mod}_-(z)V^{mod}(z), \ \ \ z\in \Sigma^{cut},$$
	where the jump matrix $V^{mod}(z) $ is given by
	\begin{equation}
		\hat{V}^{mod}(z)=\left\{ \begin{array}{ll}
			\left(\begin{array}{cc}
				0 & ie^{-itB_2}\\
				ie^{itB_2} & 0
			\end{array}\right),   &\text{as } z\in[-c,-z_0] ,\\[12pt]
			\left(\begin{array}{cc}
				0 & ie^{-itB_1}\\
				ie^{itB_1} & 0
			\end{array}\right),   &\text{as } z\in[-1,1],\\[12pt]	
			\left(\begin{array}{cc}
				0 & i\\
				i & 0
			\end{array}\right),   &\text{as } z\in[z_0,c].	
		\end{array}\right.
	\end{equation}

$\blacktriangleright$ Singularity: $M^{mod}(z)$ has singularity at $z=\pm c$ with
\begin{align}
	&M^{mod}(z)\sim \mathcal{O}(z\mp p)^{-1/4} ,\ z\to p=\pm c,\ \pm 1,\ \pm z_0\text{ in }\mathbb{C}\setminus\mathbb{R}.
\end{align}
\end{RHP}
The  solution  $M^{mod}$ of the model RH problem     can be characterized  by  $\vartheta$ function on the Riemann surface with genus-2. We   define
\begin{align}
	\mathcal{N}(z)&=\frac{1}{2}\left(\begin{array}{cc}
		\kappa(z)+\kappa(z)^{-1}&\kappa(z)-\kappa(z)^{-1}\\
		\kappa(z)-\kappa(z)^{-1}&\kappa(z)+\kappa(z)^{-1}\\
	\end{array}\right),
\end{align}
where
\begin{align}
	\kappa(z)&=\left[ \frac{(z-c)(z-1)(z+z_0)}{(z-z_0)(z+1)(z+c)}\right] ^{\frac{1}{4}}, \ \  z\in\mathbb{C}\setminus\Sigma^{cut},\nonumber\\
	\kappa(z)&=1+\mathcal{O}(z^{-2}), \ \ z\to\infty.\nonumber
\end{align}
Let $\omega_i,a_i,b_i,i=1,2$ denote the standard holomorphic differentials, canonical $a,b$ periods  on the  genus 2 Riemann surface $\mathcal{M}$ which is covered by two sheets of $\mathbb{C}\setminus\Sigma^{cut}$.  And matrix $\tilde{B}\in GL_2(\mathbb{C})$, $\tilde{B}_{ij}=\oint_{b_j}\omega_i$, \  $i,j=1,2$.
Considering the Abel map
\begin{align}
	\mathcal{A}:\mathcal{M}&\to \mathbb{C}^2/\tilde{B}M+N, \ \ M, N\in \mathbb{Z}^2\\
	P&\mapsto\left(\int_{c}^{P}\omega_1, \ \int_{c}^{P}\omega_2\right)^T.
\end{align}
The $\vartheta$ function is defined by
\begin{align}
	\vartheta(z)=\sum_{n \in\mathbb{Z}^2}\exp(\pi i \langle B n, n \rangle+2\pi i\langle n, z\rangle )
\end{align}
which satisfies
\begin{align}
	\vartheta(z\pm e_j)&=	\vartheta(z),\\	\vartheta(z\pm B e_j)&=\exp(\mp2\pi iu_j-\pi iB_{jj})\vartheta (z).
\end{align}
According to \cite{BA},  there is a constant $\mathcal{K}\in \mathbb{C}^2$ such that  for arbitrary divisor $\mathcal{P}_0$,
\begin{align*}
	\vartheta(\mathcal{A}(P)-\mathcal{A}(\mathcal{P}_0)-\mathcal{K}) := \vartheta(\mathcal{A}(P)-K)
\end{align*}
has  two  zeros $P_1, P_2$ on $\mathcal{M}$ with $P_1+P_2=\mathcal{P}_0$.

Observing if $P_1,P_2$ is zeros of $\mathcal{N}_{11},\mathcal{N}_{22}$ only if $P'_1,P'_2$ is zeros of $\mathcal{N}_{12},
\mathcal{N}_{21}$ where $P'_1,P'_2$ are the same point of $P_1,P_2$ on the other sheet. Therefore, $\vartheta(\mathcal{A}(P)-K),\
\vartheta(\mathcal{A}(P)+K)$
have the same zeros as $\mathcal{N}_{11},\ \mathcal{N}_{22}$ and $\mathcal{N}_{12},\ \mathcal{N}_{21}$ have respectively at the time of $\mathcal{P}_0=P_1+P_2$.
We then can show that  the RHP \ref{RHP7} admits the solution
\begin{align}
	M^{mod}(z)&= F(\infty) \left(\begin{array}{cc}
	\mathcal{N}_{11}\frac{\vartheta(\mathcal{A}(z)-K+C)}{\vartheta(\mathcal{A}(z)-K)}&
	\mathcal{N}_{12}\frac{\vartheta(-\mathcal{A}(z)-K+C)}{\vartheta( \mathcal{A}(z)+K)}\\[10pt]
	\mathcal{N}_{21}\frac{\vartheta(\mathcal{A}(z)+K+C)}{\vartheta(\mathcal{A}(z)+K)}&
	\mathcal{N}_{22}\frac{\vartheta(-\mathcal{A}(z)+K+C)}{\vartheta(-\mathcal{A}(z)+K)}\\
\end{array}\right),\label{Mmod0}
\end{align}
where
$$ F(\infty)= \frac{1}{2}{\rm diag}  \left(  \frac{\vartheta(\mathcal{A}(\infty)-K)}{\vartheta(\mathcal{A}(\infty)-K+C)},
\frac{\vartheta(\mathcal{A}(\infty)-K)}{\vartheta(\mathcal{A}(\infty)-K-C)} \right).  $$
Noting that  as $z\to0\in\mathbb{C}^+$,
\begin{align}
	\mathcal{N}(z)=\frac{\sqrt{2}}{2}
	\left(\begin{array}{cc}
	1 & i\\
	i & 1\\
	\end{array}\right)+z\frac{\sqrt{2}}{4}\left(\frac{1}{z_0}-\frac{1}{c}-1 \right) \left(\begin{array}{cc}
	i & 1\\
	1 & i\\
\end{array}\right)+\mathcal{O}(z^2),
\end{align}
then by using (\ref{Mmod0}),  we have
\begin{align}
		M^{mod}(z)=M^{mod}(0)+M^{mod}_1z+\mathcal{O}(z^2),
\end{align}
where
\begin{align}
&M^{mod}_1=\frac{\sqrt{2}}{4}P(\infty)\left(\frac{1}{z_0}-\frac{1}{c}-1 \right) \left(\begin{array}{cc}
	i & 1\\
	1 & i\\
\end{array}\right) G(0)+  \frac{\sqrt{2}}{2}
	\left(\begin{array}{cc}
	1 & i\\
	i & 1\\
	\end{array}\right) G_z(0),\nonumber
\end{align}
and
$$G(z)=\left(\begin{array}{cc}
	 \frac{\vartheta(\mathcal{A}(z)-K+C)}{\vartheta(\mathcal{A}(z)-K)}&
	 \frac{\vartheta(-\mathcal{A}(z)-K+C)}{\vartheta(-\mathcal{A}(z)-K)}\\[10pt]
 \frac{\vartheta(\mathcal{A}(z)+K+C)}{\vartheta(\mathcal{A}(z)+K)}&
 \frac{\vartheta(-\mathcal{A}(z)+K+C)}{\vartheta(-\mathcal{A}(z)+K)}\\
\end{array}\right).  $$

\subsubsection{Localized  RH problem near phase points }\label{seclo}

\quad Although the exponential function in $\hat{V}^{(2)}(z)$  decays exponentially on $\Sigma^{(1)} \setminus\Sigma^{cut}$ as $t\to\infty$. This decay is not uniform with respect to $z$ as $z$ approaches $\Sigma^{cut}$. Thus, on the parts of the contour $\Sigma^{(1)}$,
that lie near $\Sigma^{cut}$, we need to introduce local solutions that are better approximations of  $M^{(2)}(z)$
than $M^{mod}(z)$ is. Using these local approximations, we can derive appropriate error estimates as well as higher order asymptotics beyond the $\mathcal{O}(1)$ term. In this subsection, we only give the details of the model around $z_0$. We consider $M^{lo,+}(z)$ here as an example. Firstly, we denote $P(z_0)$ is the neighborhood of $z_0$ in the Riemann surface corresponding to $U(z_0)$. And let $\omega$ be  the local coordinate of $P(z_0)$. It is a analytic homeomorphism. Observing  $g$ is an holomorphic function  in $P(z_0)$.  And $z_0$ is its zero of order 3. Therefore, there exist a holomorphic homeomorphism  $f_+$ on $P(z_0)$ such that
\begin{equation}
	g(\omega)=-\frac{2i}{3}f_+^3(\omega).
\end{equation}
Here, on the complex plane,
because of $g_+(\left( z_0,c\right)) \subset i\mathbb{R}^-$, we can choose $f_+\in \mathbb{R}^+ $ in $z\in\left( z_0,c\right) \cap U(z_0)$.
Let
\begin{equation}
	\lambda_+(\omega)=t^\frac{2}{3}f_+^2(\omega).\label{lam}
\end{equation}
Because of $g(\omega)=-g(-\omega)$, $z-z_0\mapsto\lambda_+$ is a  holomorphic homeomorphism from $U(z_0)$ to
a neighborhood  of zero. And
\begin{equation}
	\frac{4}{3}\lambda_+^{\frac{3}{2}}=2itg_+, \ \ z\in\left( z_0,c\right) \cap U(z_0),
\end{equation}
where $(\cdot)^{\frac{3}{2}}$ is the same as in the Airy model in the Apeendix \ref{appairy}.

From the definition of $\lambda_+$, we have
\begin{eqnarray}
	\frac{4}{3}\lambda_+^{\frac{3}{2}}=2itg,\ \ \text{Im}z>0,\ \ \ 	\frac{4}{3}\lambda_+^{\frac{3}{2}}=-2itg,\ \ \text{Im}z<0.
\end{eqnarray}
We choose the jump contour $\Sigma_3,\Sigma_4$ satisfying
\begin{eqnarray}
	\lambda_+:\Sigma_3\cap U(z_0)\mapsto e^{\frac{2i\pi}{3}}\mathbb{R}^+,\\
	\lambda_+:\Sigma_4\cap U(z_0)\mapsto e^{\frac{4i\pi}{3}}\mathbb{R}^+.
\end{eqnarray}
Define $M^{lo,+}(z)$ as follow
\begin{align}
	M^{lo,+}(z)=M^{mod}H(z;z_0)^{-1}N^{-1}\lambda_+^{\frac{\sigma_3}{4}}m^{Ai}(\lambda_+)H(z;z_0).\label{Mlo+}
\end{align}
The definition of $N$ comes from (\ref{asyairy}) and
\begin{align}
	H(z;z_0)=\left\{
	\begin{array}{ll}
		D^{\sigma_3}r_1^{\frac{\sigma_3}{2}},\hspace{1cm}  z-z_0\in\mathbb{C}^+\cap U(z_0),\\
		\sigma_3\sigma_1D^{\sigma_3}r_2^{-\frac{\sigma_3}{2}}, \ z-z_0\in\mathbb{C}^-\cap U(z_0).
	\end{array}
	\right.
\end{align}
Moreover, $M^{mod}H(z;z_0)^{-1}N^{-1}\lambda_+^{\frac{\sigma_3}{4}}$ is a analytic and  invertible function in $U(z_0)$. Similarly,
\begin{align}
	M^{lo,-}(z)=M^{mod}H(z;-z_0)^{-1}N^{-1}\lambda_-^{\frac{\sigma_3}{4}}m^{Ai}(\lambda_-)H(z;-z_0).\label{Mlo-}
\end{align}
with
\begin{align}
	H(z;-z_0)=\left\{
	\begin{array}{ll}
		e^{\frac{itB_2}{2}\sigma_3}\sigma_3D^{\sigma_3}r_1^{\frac{\sigma_3}{2}},\hspace{0.2cm}  z-z_0\in\mathbb{C}^+\cap U(-z_0),\\
		e^{\frac{itB_2}{2}\sigma_3}\sigma_1D^{\sigma_3}r_2^{-\frac{\sigma_3}{2}},z-z_0\in\mathbb{C}^-\cap U(-z_0),
	\end{array}
	\right.
\end{align}
\begin{equation}
	g=-\frac{2i}{3}f_-^3+\frac{B_2}{2},\hspace{1cm}	\lambda_-=t^\frac{2}{3}f_-^2,
\end{equation}
where $f_-\in \mathbb{R}+$ on $z\in (-c,-z_0)\cap U(-z_0)$ and $M^{mod}H(z;-z_0)^{-1}N^{-1}\lambda_-^{\frac{\sigma_3}{4}}$ is a analytic and  invertible function in $U(-z_0)$.
Moreover, as $z\to z_0$ in $\mathbb{C}\setminus[z_0,c]$, $\lambda_+$ has expending
\begin{align}
	\lambda_+=(\tilde{A}_+)^{\frac{2}{3}}(z-z_0)+\frac{2\tilde{B}_+}{3\tilde{A}_+^{\frac{1}{3}}}(z-z_0)^2+\mathcal{O}((z-z_0)^3),\\
	\lambda_-=(\tilde{A}_-)^{\frac{2}{3}}(z+z_0)+\frac{2\tilde{B}_-}{3\tilde{A}_-^{\frac{1}{3}}}(z+z_0)^2+\mathcal{O}((z+z_0)^3),
\end{align}
where
{\footnotesize\begin{align}
	\tilde{A}_+=&\frac{t}{2} \left(\frac{2z_0}{(z_0^2-1)(c^2-z_0^2)}\right)^{\frac{1}{2}}\left( \frac{\xi-1}{2}z_0+\frac{c}{z_0^2}(1+\frac{1}{c^2})-\frac{3c}{z_0^4}\right),\nonumber\\[6pt]
	\tilde{B}_+=&\frac{3t}{10}\left[-\frac{1}{2}(\frac{2z_0}{(z_0^2-1)(c^2-z_0^2)})^{-\frac{1}{2}}\frac{-7z_0^4+3(1+c^2)z_0^2+c^2}{(z_0^2-1)^2(c^2-z_0^2)^2}
	 \left(\frac{\xi-1}{2}z_0+\frac{c}{z_0^2}(1+\frac{1}{c^2})-\frac{3c}{z_0^4}\right)\right.\nonumber\\[4pt]
	&\left.+(\frac{2z_0}{(z_0^2-1)(c^2-z_0^2)})^{\frac{1}{2}}(\frac{\xi-1}{2}-\frac{c}{z_0^3}(1+\frac{1}{c^2}-\frac{1}{z_0^2})+\frac{6c}{z_0^5})\right],\nonumber\\[6pt]
	\tilde{A}_-=&-\frac{it}{2}\left(\frac{2z_0}{(z_0^2-1)(c^2-z_0^2)}\right)^{\frac{1}{2}}\left(-\frac{\xi-1}{2}z_0+\frac{c}{z_0^2}(1+\frac{1}{c^2})-\frac{3c}{z_0^4}\right),\nonumber\\[4pt]
	\tilde{B}_-=&\frac{3t}{10}\left[-\frac{i}{2}(\frac{2z_0}{(z_0^2-1)(c^2-z_0^2)})^{-\frac{1}{2}}\frac{-7z_0^4+3(1+c^2)z_0^2+c^2}{(z_0^2-1)^2(c^2-z_0^2)^2}
    (-\frac{\xi-1}{2}z_0+\frac{c}{z_0^2}(1+\frac{1}{c^2})-\frac{3c}{z_0^4})\right.\nonumber\\[4pt]
	&\left.-i(\frac{2z_0}{(z_0^2-1)(c^2-z_0^2)})^{\frac{1}{2}}(\frac{\xi-1}{2}+\frac{c}{z_0^3}(1+\frac{1}{c^2}-\frac{1}{z_0^2})-\frac{6c}{z_0^5})\right].\label{expf}
\end{align}}

\subsubsection{The small norm RH problem  for error function   }\label{secer}
\quad In this subsection,  we consider the error matrix-function $E(z;\xi,c)$.

\begin{RHP}\label{RHPE}
 Find a matrix-valued function $E(z;\xi,c)$  with following properties:

$\blacktriangleright$ Analyticity: $E(z;\xi,c)$ is analytical  in $\mathbb{C}\setminus  \Sigma^{ E } $, where
$$\Sigma^{ E }= \partial U_{ \xi }\cup
\left[ \Sigma^{(2)}\setminus \left( U_{(\xi)}\cup[-c,-z_1]\cup[-1,1]\cup[z_1,c]\right) \right] ;$$

$\blacktriangleright$ Asymptotic behaviors:
\begin{align}
	&E(z;\xi,c) \sim I+\mathcal{O}(z^{-1}),\hspace{0.5cm}|z| \rightarrow \infty;
\end{align}

$\blacktriangleright$ Jump condition: $E(z;\xi,c)$ has continuous boundary values $E_\pm(z;\xi,c)$ on $\Sigma^{E}$ satisfying
$$E_+(z;\xi,c)=E_-(z;\xi,c)V^{ E }(z),$$
where the jump matrix $V^{E}(z)$ is given by
\begin{equation}
	V^{ E }(z)=\left\{\begin{array}{llll}
		M^{mod}(z)V^{(2)}(z)M^{mod}(z)^{-1}, & z\in \Sigma^{ E }\setminus \partial U_{ \xi },\\[4pt]
	M^{lo,\pm}(z)M^{mod}(z)^{-1},  & z\in \partial U_{\pm z_0},
	\end{array}\right. \label{deVE}
\end{equation}
which is  shown in  Figure \ref{figE}.
\end{RHP}
The  above  RH problem, by (\ref{Mlo+}), (\ref{Mlo-}) and (\ref{asymlo}), satisfies
\begin{align*}
	\parallel V^{E}(z)-I\parallel_2 \lesssim \mathcal{O}(t^{-1}).
\end{align*}

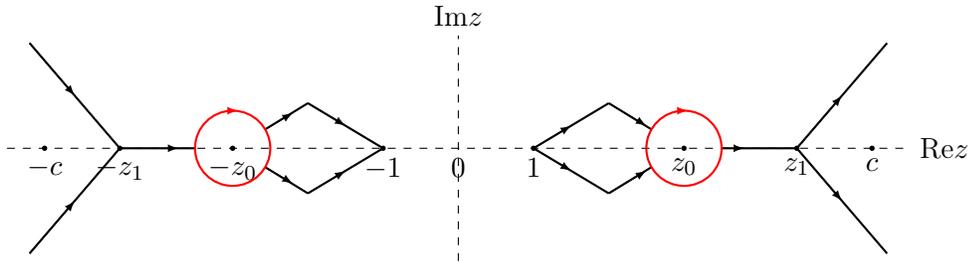
\begin{figure}[H]
	\centering
		\begin{tikzpicture}
			\draw[dashed](0,-1.5)--(0,1.5)node[above]{ Im$z$};
			\draw[dashed](-6,0)--(6,0)node[right]{ Re$z$};
			\coordinate (I) at (0,0);
			\fill (I) circle (0pt) node[below] {$0$};
			\coordinate (a) at (3,0);
			\fill (a) circle (1pt) node[below] {$z_0$};
			\coordinate (aa) at (-3,0);
			\fill (aa) circle (1pt) node[below] {$-z_0$};
			\coordinate (b) at (1,0);
			\fill (b) circle (1pt) node[below] {$1$};
			\coordinate (ba) at (-1,0);
			\fill (ba) circle (1pt) node[below] {$-1$};
			\coordinate (c) at (5.5,0);
			\fill (c) circle (1pt) node[below] {$c$};
			\coordinate (ca) at (-5.5,0);
			\fill (ca) circle (1pt) node[below] {$-c$};
			\coordinate (y) at (4.5,0);
			\fill (y) circle (1pt) node[below] {$z_1$};
			\coordinate (cy) at (-4.5,0);
			\fill (cy) circle (1pt) node[below] {$-z_1$};
			\draw[thick](4.5,0)--(5.7,1.4);
			\draw[thick](-4.5,0)--(-5.7,1.4);
			\draw[thick](-4.5,0)--(-5.7,-1.4);
			\draw[thick](4.5,0)--(5.7,-1.4);
			\draw[-latex](-5.7,-1.4)--(-5.1,-0.7);
			\draw[-latex](-5.7,1.4)--(-5.1,0.7);
			\draw[-latex](4.5,0)--(5.1,0.7);
			\draw[-latex](4.5,0)--(5.1,-0.7);
			\draw[thick](3.5,0)--(4.5,0);
			\draw[thick](-3.5,0)--(-4.5,0);
			\draw[thick](-2.57,0.25)--(-2,0.6);
			\draw[thick](-1,0)--(-2,0.6);
			\draw[thick](2.57,0.25)--(2,0.6);
			\draw[thick](1,0)--(2,0.6);
			\draw[thick](-2.57,-0.25)--(-2,-0.6);
			\draw[thick](-1,0)--(-2,-0.6);
			\draw[thick](2.57,-0.25)--(2,-0.6);
			\draw[thick](1,0)--(2,-0.6);
			\draw[-latex](-2,0.6)--(-1.5,0.3);
			\draw[-latex](-2.57,0.25)--(-2.2,0.47);
			\draw[-latex](-2,-0.6)--(-1.5,-0.3);
			\draw[-latex](-2.57,-0.25)--(-2.2,-0.47);
			\draw[-latex](1,0)--(1.5,0.3);
			\draw[-latex](2,0.6)--(2.5,0.3);
			\draw[-latex](1,0)--(1.5,-0.3);
			\draw[-latex](2,-0.6)--(2.5,-0.3);
			\draw[-latex](-3.73,0)--(-3.72,0);
			\draw[-latex](3.73,0)--(3.8,0);
			\draw[thick,red](3,0) circle (0.5);
			\draw[thick,red](-3,0) circle (0.5);
			\draw[-latex,red](3,0.5)--(3.05,0.5);
			\draw[-latex,red](-3,0.5)--(-2.93,0.5);
		\end{tikzpicture}
	\caption{  The jump contour $\Sigma^{E}$ for the $E(z;\xi,c)$. The red circles are $U(\xi)$. }
	\label{figE}
\end{figure}

%
Similar with the discussion in Section \ref{erroranalysis},  the RHP  \ref{RHPE} admits  a unique  solution, which
  can be given by
\begin{equation}
	E(z;\xi,c)=I+\frac{1}{2\pi i}\int_{\Sigma^{ E }}\dfrac{\left( I+\varpi(s)\right) (V^{E}(s)-I)}{s-z}ds,\label{Ez}
\end{equation}
where the $\varpi\in L^\infty(\Sigma^{ E })$ is the unique solution of following equation
\begin{equation}
	(1-C_E)\varpi=C_E\left(I \right).
\end{equation}
In order to reconstruct the solution $u(y,t)$ of (\ref{mch}), we need the asymptotic behavior of $E(z;\xi,c)$ as $z\to 0\in \mathbb{C}^+$ and the long time asymptotic behavior of $E(0)$. Note that when we estimate its  asymptotic behavior, we only need to consider the calculation on $\partial U(\xi)$ because it  approach zero exponentially on other boundary. For convenience, we  denote
\begin{align}
	F^\pm=M^{mod}H(z;\pm z_0)^{-1}N^{-1}f_\pm^{\frac{1}{2}\sigma_3}=[F^\pm_{ij}]_{2\times2},
\end{align}
with $(F^\pm)^{-1}\triangleq[F^{\pm,*}_{ij}]_{2\times2}$.
\begin{Proposition}\label{asyE}
	As $z\to 0\in \mathbb{C}^+$, we have
	\begin{align}
		E(z;\xi,c)=E(0)+E_1z+\mathcal{O}(z^2),
	\end{align}
	where
	\begin{align}
	E(0)=I+\frac{1}{2\pi i}\int_{\Sigma^{E}}\dfrac{\left( I+\varpi(s)\right) (V^{ E }-I)}{s}ds,
	\end{align}
	with long time asymptotic behavior
	\begin{equation}
		E(0)=I+t^{-1}H^{(0)}+\mathcal{O}(t^{-2}).\label{E0t}
	\end{equation}
	And
	\begin{align}
		H^{(0)}=H^{(0)}(\xi,c)=&\sum_{ p=\pm z_0 }\frac{\text{d}}{\text{d}z}\left(\frac{1 }{z}F^{\pm}\left(\begin{array}{cc}
			0&-\frac{5i}{48}\tilde{A}_\pm^{-\frac{4}{3}}\\
			0&0\\
		\end{array}\right)(F^\pm)^{-1}\right)(p) \nonumber\\
		&+\sum_{ p=\pm z_0 }\frac{1 }{p}\left( F^{\pm}\left(\begin{array}{cc}
			0&\frac{5}{36}\tilde{B}_\pm\tilde{A}^{-\frac{7}{3}}_\pm\\
			-\frac{7i}{48}\tilde{A}^{-\frac{2}{3}}_\pm&0\\
		\end{array}\right)(F^\pm)^{-1}\right)(p).\nonumber
	\end{align}
Here, $\tilde{B}_\pm$, $\tilde{A}_\pm$ is shown in (\ref{expf}).
 And
	\begin{equation}
		E_1=\frac{1}{2\pi i}\int_{\Sigma^{E}}\dfrac{\left( I+\varpi(s)\right) (V^{ E }-I)}{s^2}ds,
	\end{equation}
	satisfying long time asymptotic behavior condition
	\begin{equation}
		E_1=t^{-1}H^{(1)}+\mathcal{O}(t^{-2}),\label{E1t}
	\end{equation}
	where
	\begin{align}
		H^{(1)}=H^{(1)}(\xi,c)=&\sum_{ p=\pm z_0 }\frac{\text{d}}{\text{d}z}\left(\frac{1 }{z^2}F^{\pm}\left(\begin{array}{cc}
			0&-\frac{5i}{48}\tilde{A}_\pm^{-\frac{4}{3}}\\
			0&0\\
		\end{array}\right)(F^\pm)^{-1}\right)(p) \nonumber\\
		&+\sum_{ p=\pm z_0 }\frac{1 }{p^2}\left( F^{\pm}\left(\begin{array}{cc}
			0&\frac{5}{36}\tilde{B}_\pm\tilde{A}^{-\frac{7}{3}}_\pm\\
			-\frac{7i}{48}\tilde{A}^{-\frac{2}{3}}_\pm&0\\
		\end{array}\right)(F^\pm)^{-1}\right)(p).\nonumber
	\end{align}
\end{Proposition}
\begin{proof}
	By using  expansion  of $V^{ E }$  and $\varpi(s)=\mathcal{O}(t^{-1})$ , we have
	\begin{align}
		&\int_{\Sigma^{E}}\dfrac{\left( I+\varpi(s)\right) (V^{E}-I)}{s}ds\nonumber\\
		&=\frac{1}{t}\int_{\partial U(\pm z_0)}\dfrac{ M^{mod}H(z;\pm z_0)^{-1} m^{Ai}_1H(z;\pm z_0)(M^{mod})^{-1}}{sf_\pm^{3}}ds+\mathcal{O}(t^{-2}).
	\end{align}
Rewrite the as
\begin{align*}
	M^{mod}H(z;\pm z_0)^{-1} m^{Ai}_1H(z;\pm z_0)(M^{mod})^{-1}=F^\pm f_\pm^{-\frac{\sigma_3}{2}}Nm^{Ai}_1N^{-1}f_\pm^{\frac{\sigma_3}{2}}(F^\pm)^{-1}.
\end{align*}
Here, from (\ref{expf}) and the definition of $N$, $m^{Ai}_1$ in Appendix \ref{appairy}, as $z\to z_0\in \mathbb{C}^+$, we have
\begin{align}
	&f_\pm^{-3}f_\pm^{-\frac{\sigma_3}{2}}Nm^{Ai}_1N^{-1}f_\pm^{\frac{\sigma_3}{2}}=\nonumber\\
	&\frac{1}{(z\mp z_0)^2}\left(\begin{array}{cc}
		0&-\frac{5i}{48}\tilde{A}_\pm^{-\frac{4}{3}}\\
		0&0\\
	\end{array}\right)+\frac{1}{(z\mp z_0)}
\left(\begin{array}{cc}
	0&\frac{5}{36}\tilde{B}_\pm\tilde{A}^{-\frac{7}{3}}_\pm\\
	-\frac{7i}{48}\tilde{A}^{-\frac{2}{3}}_\pm&0\\
\end{array}\right)+\mathcal{O}(1).
\end{align}
Then by residue theorem we finally arrive at the result.
\end{proof}
\subsection{Opening the jump in the region $3/4<\xi<1$  }\label{subsecg2}
\quad
This region  include  two cases:   (1) $ 2< c^2<4$,  $ 1-\frac{2(c^2-2)}{c^4} <\xi <1$;\quad (2)  $c^2>4$, $\xi_m<\xi<1$.
We introduce a $g$ function by
 \begin{align*}
	\text{d}g=\dfrac{	Y(z)}{z^3}\left[ \frac{1-\xi}{2}z^4-\frac{c}{z_0}\left( 1+\frac{1}{c^2}-\frac{1}{z_0^2}\right) z^2+\frac{2c}{z_0}\right] \text{d}z
\end{align*}
will have another zero on $\mathbb{R}$ except on cut. It means $g$ have three pairs of stationary phase points on $\mathbb{R}$, which will gives  more contribution as $t\to\infty$. Consider the equation
\begin{align*}
	\frac{1-\xi}{2}z^4-\frac{c}{z_0}\left( 1+\frac{1}{c^2}-\frac{1}{z_0^2}\right) z^2+\frac{2c}{z_0}=0.
\end{align*}
It has two pairs of zeros on $\mathbb{R}$: $\pm z_1\in(z_0,c)$, $\pm z_2$. Note that,
$z_1^2z_2^2=\frac{4c}{z_0(1-\xi)}$.
For a given $c$, when $\xi$  decreases from $1$, $z_2$ as a function of $\xi$ decreases from $+\infty$. We denote $\xi_m$ as the
critical condition of $\xi$ that stationary phase point $z_2$ merge $c$. Under this case, the sign table of Im$g$ has the  following figure:
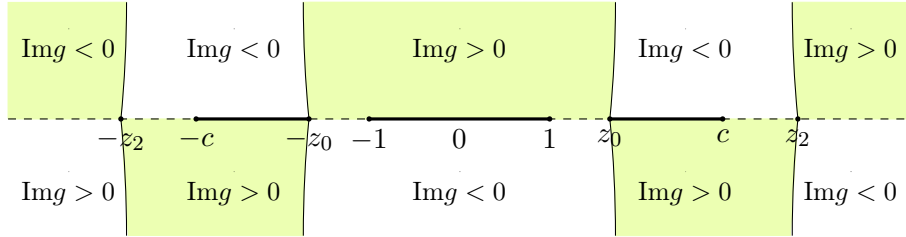
\begin{figure}[h]
	\centering
	\begin{tikzpicture}[node distance=2cm]
		\draw[lime!30, fill=lime!30](-4.5,0)--(-2,0)arc (165:180:2.2 and 6)--(-4.4,-1.55)--(-4.5,0);
		\draw[lime!30, fill=lime!30](4.5,0)--(2,0)arc (15:0:2.2 and 6)--(4.4,-1.55)--(4.5,0);
		\draw[lime!30, fill=lime!30](2,0)arc (-15:0:2.2 and 6)--(-2,1.55)--(-2,0)--(2,0);
		\draw[lime!30, fill=lime!30](-2,0)arc (195:180:2.2 and 6)--(2,1.55)--(2,0)--(-2,0);
		\draw[lime!30, fill=lime!30](-6,0)--(-4.5,0)arc (-15:0:2.2 and 6)--(-6,1.55)--(-6,0);
		\draw[lime!30, fill=lime!30](6,0)--(4.5,0)arc  (195:180:2.2 and 6)--(6,1.55)--(6,0);
		\draw[dashed](-6,0)--(6,0);
		\draw(2,0)arc (-15:0:2.2 and 6);
		\draw(2,0)arc (15:0:2.2 and 6);
		\draw[](-2,0)arc (165:180:2.2 and 6);
		\draw[](-2,0)arc (195:180:2.2 and 6);
		\draw[](4.5,0)arc (165:180:2.2 and 6);
		\draw[](4.5,0)arc (195:180:2.2 and 6);
		\draw[](-4.5,0)arc (-15:0:2.2 and 6);
		\draw[](-4.5,0)arc (15:0:2.2 and 6);
		\coordinate (a1) at (0,1.2);
		\fill (a1) circle (0pt) node[below] {\small $\text{Im}g>0$};
		\coordinate (e1) at (-3,-0.7);
		\fill (e1) circle (0pt) node[below] {\small $\text{Im}g>0$};
		\coordinate (o1) at (3,-0.7);
		\fill (o1) circle (0pt) node[below] {\small $\text{Im}g>0$};	
		\coordinate (w1) at (0,-0.7);
		\fill (w1) circle (0pt) node[below] {\small $\text{Im}g<0$};
		\coordinate (e2) at (-3,1.2);
		\fill (e2) circle (0pt) node[below] {\small $\text{Im}g<0$};
		\coordinate (o2) at (3,1.2);
		\fill (o2) circle (0pt) node[below] {\small $\text{Im}g<0$};				
		\coordinate (I) at (0,0);
		\fill (I) circle (0pt) node[below] {$0$};
		\coordinate (a) at (2,0);
		\fill (a) circle (1pt) node[below] {$z_0$};
		\coordinate (aa) at (-2,0);
		\fill (aa) circle (1pt) node[below] {$-z_0$};
		\coordinate (b) at (1.2,0);
		\fill (b) circle (1pt) node[below] {$1$};
		\coordinate (ba) at (-1.2,0);
		\fill (ba) circle (1pt) node[below] {$-1$};
		\coordinate (c) at (3.5,0);
		\fill (c) circle (1pt) node[below] {$c$};
		\coordinate (ca) at (-3.5,0);
		\fill (ca) circle (1pt) node[below] {$-c$};
		\coordinate (cr) at (4.5,0);
		\fill (cr) circle (1pt) node[below] {$z_2$};
		\coordinate (car) at (-4.5,0);
		\fill (car) circle (1pt) node[below] {$-z_2$};
		\draw[very thick](-3.5,0)--(-2,0);
		\draw[very thick](3.5,0)--(2,0);
		\draw[very thick](-1.2,0)--(1.2,0);
		\coordinate (ww) at (5.2,1.2);
		\fill (ww) circle (0pt) node[below] {\small $\text{Im}g>0$};	
		\coordinate (www) at (5.2,-0.7);
		\fill (www) circle (0pt) node[below] {\small $\text{Im}g<0$};
		\coordinate (ww) at (-5.2,1.2);
		\fill (ww) circle (0pt) node[below] {\small $\text{Im}g<0$};	
		\coordinate (www) at (-5.2,-0.7);
		\fill (www) circle (0pt) node[below] {\small $\text{Im}g>0$};
	\end{tikzpicture}
	\caption{\footnotesize  In the purpler region, Im$g$>0 while in the white region, Im$g$<0. }
	\label{figdg2}
\end{figure}

Similarly as the above section, we define the following contour relying on $\xi, c$:
\begin{align*}
	\Sigma_{1}^\pm =&\left\lbrace -z_1+e^{(\pi\mp \psi) i}l,\ l\in(0,\frac{z_1-z_2}{2\cos\psi})\right\rbrace \cup\left\lbrace z_1+e^{\pm\psi i}l,\ l\in(0,\frac{z_1-z_2}{2\cos\psi})\right\rbrace \nonumber\\
	&\cup\left\lbrace z_2+e^{(\pi\mp \psi) i},\ l\in(0,\frac{z_1-z_2}{2\cos\psi})\right\rbrace \cup\left\lbrace -z_2+e^{\pm\psi i},\ l\in(0,\frac{z_1-z_2}{2\cos\psi})\right\rbrace ,\\
	\Sigma_{2}^\pm=&\left\lbrace -z_0+e^{\pm\psi i}l,\ l\in(0,\frac{z_0-1}{2\cos\psi})\right\rbrace \cup\left\lbrace z_0+e^{(\pi\mp\psi) i}l,\ l\in(0,\frac{z_0-1}{2\cos\psi})\right\rbrace\nonumber\\
	&\cup\left\lbrace 1+e^{\pm \psi i}l,\ l\in(0,\frac{z_0-1}{2\cos\psi})\right\rbrace \cup\left\lbrace -1+e^{(\pi \mp \psi) i}l,\ l\in(0,\frac{z_0-1}{2\cos\psi})\right\rbrace \nonumber\\
	&\cup\left\lbrace -z_2+e^{(\pi\mp\psi) i}\mathbb{R}^+\right\rbrace \cup\left\lbrace z_2+e^{\pm\psi i}\mathbb{R}^+\right\rbrace.
\end{align*}
 Here  $ \psi \leq \frac{\pi}{4}$ is a small enough positive constant such that $\Sigma_j^\pm,\ j=1,2$ are  contained in the region of Im$g>0$. And similarly as above equation, $\Omega_j^\pm$ is a closed region,  the edge of which is made up by $\Sigma_j^\pm,\ j=1,2$ and $\mathbb{R}$.
\begin{figure}[h]
	\centering
	\begin{tikzpicture}[node distance=2cm]
		\draw[lime!30, fill=lime!30](-5,0)--(-2,0)arc (165:180:2.2 and 6)--(-4.9,-1.55)--(-5,0);
		\draw[lime!30, fill=lime!30](5,0)--(2,0)arc (15:0:2.2 and 6)--(4.9,-1.55)--(5,0);
		\draw[lime!30, fill=lime!30](2,0)arc (-15:0:2.2 and 6)--(-2,1.55)--(-2,0)--(2,0);
		\draw[lime!30, fill=lime!30](-2,0)arc (195:180:2.2 and 6)--(2,1.55)--(2,0)--(-2,0);
		\draw[lime!30, fill=lime!30](-6.5,0)--(-5,0)arc (-15:0:2.2 and 6)--(-6.5,1.55)--(-6.5,0);
		\draw[lime!30, fill=lime!30](6.5,0)--(5,0)arc  (195:180:2.2 and 6)--(6.5,1.55)--(6.5,0);
		\draw[dashed](-6.5,0)--(6.5,0);	
		\coordinate (I) at (0,0);
		\fill (I) circle (0pt) node[below] {\footnotesize$0$};
		\coordinate (a) at (2,0);
		\fill (a) circle (1pt) node[below] {\footnotesize$z_0$};
		\coordinate (aa) at (-2,0);
		\fill (aa) circle (1pt) node[below] {\footnotesize$-z_0$};
		\coordinate (b) at (0.8,0);
		\fill (b) circle (1pt) node[below] {\footnotesize$1$};
		\coordinate (ba) at (-0.8,0);
		\fill (ba) circle (1pt) node[below] {\footnotesize$-1$};
		\coordinate (c) at (3.65,0);
		\fill (c) circle (1pt) node[below] {\footnotesize$c$};
		\coordinate (ca) at (-3.65,0);
		\fill (ca) circle (1pt) node[below] {\footnotesize$-c$};
		\coordinate (cr) at (5,0);
		\fill (cr) circle (1pt) node[below] {\footnotesize$z_2$};
		\coordinate (car) at (-5,0);
		\fill (car) circle (1pt) node[below] {\footnotesize$-z_2$};
			\coordinate (a) at (2.8,0);
		\fill (a) circle (1pt) node[below] {\footnotesize$z_1$};
		\coordinate (aa) at (-2.8,0);
		\fill (aa) circle (1pt) node[below] {\footnotesize$-z_1$};
		\draw(-5,0)--(-6.5,1.5)node[above]{\footnotesize$\Sigma_2^+$};
		\draw(5,0)--(6.5,1.5)node[above]{\footnotesize$\Sigma_2^+$};
		\draw(-5,0)--(-6.5,-1.5)node[below]{\footnotesize$\Sigma_2^-$};
		\draw(5,0)--(6.5,-1.5)node[below]{\footnotesize$\Sigma_2^-$};
		\draw(-2,0)--(-1.4,0.4)node[above]{\footnotesize$\Sigma_2^+$};
		\draw[-latex](-6,1)--(-5.5,0.5);
		\draw[-latex](-6,-1)--(-5.5,-0.5);
		\draw[-latex](5,0)--(5.5,0.5);
		\draw[-latex](5,0)--(5.5,-0.5);
		\draw(-0.8,0)--(-1.4,0.4);
		\draw(2,0)--(1.4,0.4)node[above]{\footnotesize$\Sigma_2^+$};
		\draw(0.8,0)--(1.4,0.4);
		\draw(-2,0)--(-1.4,-0.4)node[below]{\footnotesize$\Sigma_2^-$};
		\draw(-0.8,0)--(-1.4,-0.4);
		\draw(2,0)--(1.4,-0.4)node[below]{\footnotesize$\Sigma_2^-$};
		\draw(0.8,0)--(1.4,-0.4);
		\draw[-latex](-1.4,0.4)--(-1.1,0.2);
		\draw[-latex](-2,0)--(-1.7,0.2);
		\draw[-latex](-1.4,-0.4)--(-1.1,-0.2);
		\draw[-latex](-2,0)--(-1.7,-0.2);
		\draw[-latex](0.8,0)--(1.1,0.2);
		\draw[-latex](1.4,0.4)--(1.7,0.2);
		\draw[-latex](0.8,0)--(1.1,-0.2);
		\draw[-latex](1.4,-0.4)--(1.7,-0.2);
		\draw[-latex](-2.31,0)--(-2.3,0);
		\draw[-latex](2.3,0)--(2.31,0);
		\draw[-latex](-0.01,0)--(0.01,0);
		\draw(5,0)--(3.9,0.6)node[above]{\footnotesize$\Sigma_1^+$};
		\draw(-5,0)--(-3.9,0.6)node[above]{\footnotesize$\Sigma_1^+$};
		\draw(3.9,0.6)--(2.8,0);
		\draw(-3.9,0.6)--(-2.8,0);
		\draw[-latex](-5,0)--(-4.45,0.3);
		\draw[-latex](-5,0)--(-4.45,-0.3);
		\draw[-latex](3.9,0.6)--(4.45,0.3);
		\draw[-latex](3.9,-0.6)--(4.45,-0.3);
		\draw[-latex](-3.9,0.6)--(-3.35,0.3);
		\draw[-latex](-3.9,-0.6)--(-3.35,-0.3);
		\draw[-latex](2.8,0)--(3.35,0.3);
		\draw[-latex](2.8,0)--(3.35,-0.3);
		\draw(5,0)--(3.9,-0.6)node[below]{\footnotesize $\Sigma_1^-$};
		\draw(-5,0)--(-3.9,-0.6)node[below]{\footnotesize$\Sigma_1^-$};
		\draw(3.9,-0.6)--(2.8,0);
		\draw(-3.9,-0.6)--(-2.8,0);
		\node at (0,1) {\footnotesize $\overline{\Omega} $};
		\node at (6,0.5) {\tiny $\Omega_2^+$};
		\node at (-6,0.5) {\tiny $\Omega_2^+$};
		\node at (1.42,0.15) {\tiny $\Omega_2^+$};
		\node at (-1.42,0.15) {\tiny $\Omega_2^+$};
		\node at (1.42,-0.15) {\tiny $\Omega_2^-$};
		\node at (-1.42,-0.15) {\tiny $\Omega_2^-$};
		\node at (6,-0.5) {\tiny $\Omega_2^-$};
		\node at (-6,-0.5) {\tiny $\Omega_2^-$};
		\node at (4.1,-0.25) {\tiny $\Omega_1^-$};
		\node at (-4.1,-0.25) {\tiny $\Omega_1^-$};
		\node at (4.1,0.2) {\tiny $\Omega_1^+$};
		\node at (-4.1,0.2) {\tiny $\Omega_1^+$};
	\end{tikzpicture}
	\caption{\footnotesize  The region of $\Omega_j^\pm$, $j=1,2$.
  The same as Figure \ref{figdg2}, purple region means Im$g(z)>0$ while white region means Im$g(z)<0$.}
	\label{figfj3}
\end{figure}
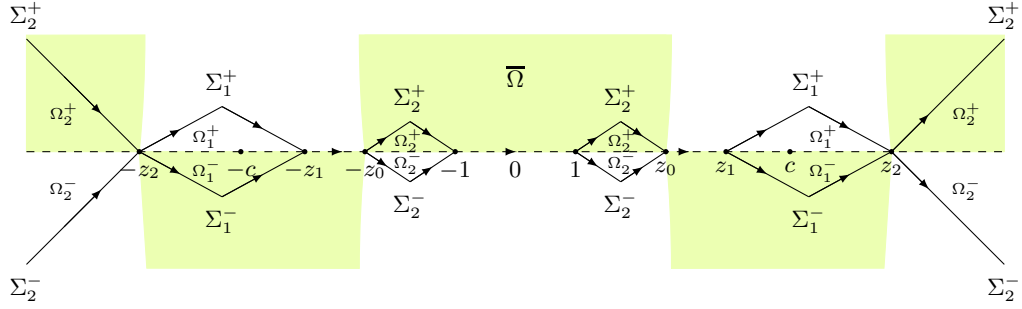

In this region of $\xi,c$, we introduce
a piecewise matrix interpolation function
\begin{equation}
	G(z) =\left\{ \begin{array}{ll}
		\left(\begin{array}{cc}
			1 &\frac{ r_2e^{-2itg}}{1-r_1r_2}\\
			0 & 1
		\end{array}\right),   &\text{as } z\in\Omega_1^+;\\[12pt]
		\left(\begin{array}{cc}
			1 &0\\
			\frac{r_1e^{2itg}}{1-r_1r_2} & 1
		\end{array}\right),   &\text{as } z\in\Omega_1^-;\\[12pt]
		\left(\begin{array}{cc}
			1 & 0\\
			-r_1e^{2itg} & 1
		\end{array}\right),   &\text{as } z\in\Omega_2^+;\\[12pt]
		\left(\begin{array}{cc}
			1 & -r_2e^{-2itg}\\
			0 & 1
		\end{array}\right),   &\text{as } z\in\Omega_2^-;\\
		I &\text{as } 	z \text{ in elsewhere},
	\end{array}\right..\label{funcG3}
\end{equation}
Same as above section,  $ G (z)$ bring a new singularity.
To deal with the jump on $\mathbb{R}$, we   introduce an auxiliary function $D(z)$, which admits the following jump condition
\begin{align*}
	&D_-(z)=D+(z)(1-r_1r_2),\hspace*{0.5cm}z\in [-z_2,z_2]\setminus[-c,c];\\
	&D_-(z)D_+(z)=i[r_2]_-,\hspace*{1.2cm}z\in[-c,c] \setminus[-z_0,z_0];\\
	&D_-(z)D_+(z)=1,\hspace*{2cm}z\in[-1,1].
\end{align*}
Define $Y_3(z)= (z^2-1)(z^2-c^2)Y(z),$ and
\begin{align}
	\log D(z)=&\frac{Y_3(z)}{2\pi i}\left(\int_{-c}^{-z_0}+\int_{z_0}^{c} \right) \dfrac{\log(i[r_2]_-(s))}{(s-z)[Y_3]_+(s)}ds\nonumber\\
	&-\frac{Y_3(z)}{2\pi i}\left( \int_{-z_2}^{-c}+\int_{c}^{z_2}\right) \dfrac{\log(1-r_1(s)r_2(s))}{(s-z)Y_3(s)}ds.
\end{align}
\begin{Proposition}\label{proD2}
	The scalar function $D(z)$ satisfies the following properties  \\
	(a) $D(z)$ is analytic on $\mathbb{C}\setminus\left( (-z_2,-z_0)\cup[-1,1]\cup(z_0,z_2)\right) $;\\
	(b) $D(z)$ has singularity at $z=\pm c$ with
	\begin{align}
		D(z)=\mathcal{O}(z-p)^{\mp 1/4},\hspace{0.3cm}z\to p\in \mathbb{C}^\pm\setminus\mathbb{R} ,\hspace{0.3cm}p=c,-c.
	\end{align}
	(c) As $z\to\infty\in\mathbb{C}\setminus\mathbb{R}$, $D(z)$ has limit  $D_\infty(z)$ with
	\begin{align}
		\log D_\infty(z)=&-\frac{1}{2\pi i}\left(\int_{-c}^{-z_0}+\int_{z_0}^{c}\right) \dfrac{\log(i[r_2]_-(s))}{[X]_+(s)}\left(z^2+zs+s^2-\frac{1+c^2+z_0^2}{2} \right) ds\nonumber\\
		&+\frac{1}{2\pi i}\left( \int_{-z_2}^{-c}+\int_{c}^{z_2}\right)\dfrac{\log(1-r_1(s)r_2(s))}{X(s)}\left(z^2+zs+s^2-\frac{1+c^2+z_0^2}{2} \right)ds. \nonumber
	\end{align}
	(d) As $z\to0\in\mathbb{C}^+$,
	\begin{align}
		D_\infty(z)=&D_\infty(0)\left(1 +D_\infty^{(1)}z\right) +\mathcal{O}(z^2),\\
		D(z)=&D(0)\left( 1+D^{(1)}z\right) +\mathcal{O}(z^2),
	\end{align}
	where
	{\small\begin{align}
	&	\log \left( D_\infty(0)\right) = \frac{1}{2\pi i}\left(\int_{-c}^{-z_0}+\int_{z_0}^{c} \right) \dfrac{\log(i[r_2]_-(s))}{[X]_+(s)}\left(\frac{1+c^2+z_0^2}{2} -s^2\right) ds\nonumber\\[5pt]
		&\qquad\qquad\qquad +\frac{1}{2\pi i}\left( \int_{-z_2}^{-c}+\int_{c}^{z_2}\right)\dfrac{\log(1-r_1(s)r_2(s))}{X(s)}\left(s^2-\frac{1+c^2+z_0^2}{2} \right)ds,\nonumber\\[5pt]
	&	D_\infty^{(1)}= -\frac{1}{2\pi i}\left(\int_{-c}^{-z_0}+\int_{z_0}^{c} \right) \dfrac{s\log(i[r_2]_-(s))}{[X]_+(s)}ds+\frac{1}{2\pi i}\left( \int_{-z_2}^{-c}+\int_{c}^{z_2}\right)\dfrac{s\log(1-r_1(s)r_2(s))}{X(s)}ds,\nonumber\\[5pt]
	&	\log \left( D(0)\right) =\frac{cz_0}{2\pi }\left(\int_{-c}^{-z_0}+\int_{z_0}^{c} \right) \dfrac{\log(i[r_2]_-(s))}{s[Y_3]_+(s)}ds -\frac{cz_0}{2\pi }\left( \int_{-z_2}^{-c}+\int_{c}^{z_2}\right)\dfrac{\log(1-r_1(s)r_2(s))}{sY_3(s)}ds,\nonumber\\[5pt]
&		D^{(1)}= \frac{cz_0}{2\pi }\left(\int_{-c}^{-z_0}+\int_{z_0}^{c} \right) \dfrac{\log(i[r_2]_-(s))}{s^2[Y_3]_+(s)}ds -\frac{cz_0}{2\pi }\left( \int_{-z_2}^{-c}+\int_{c}^{z_2}\right)\dfrac{\log(1-r_1(s)r_2(s))}{s^2Y_3(s)}ds.\nonumber
	\end{align}}
(e) As $z\to \pm z_2$, there is
\begin{align}
&	\log D(z)=\mp \frac{\log(1-r_1 r_2  )}{2\pi i}\log(z\mp z_2)+\mathcal{O}(1), \  z\to \pm z_2,\nonumber
\end{align}
where the logarithm function is analytic in   $U(\pm z_2)\setminus (\pm z_2, \pm z_2\mp \epsilon)$ respectively with some positive number $\epsilon$.
Further,  let $$\nu(z_2)= \pm \frac{1}{2\pi}\log(1-r_1r_2 ),  $$
Then we have
\begin{align}
&	D(z)= D_{\pm z_2}(z)(z\mp z_2)^{i\nu(\pm z_2)},\label{D+}
\end{align}
where $D_{\pm z_2}(z) $ are bounded analytic functions  in   $U(\pm z_2)\setminus (\pm z_2, \pm z_2\mp \epsilon)$ respectively. And have continued boundary extension. Further, $|D_{\pm z_2}(z)- D_{\pm z_2}(z_2)|/|z-z_2|$ is bounded on  $U(\pm z_2)$.
\end{Proposition}
Denote
\begin{align}
	\Sigma^{(2)}=\left(\cup_{j=1}^2\Sigma_j^\pm\right) \cup[-1,1]\cup[-c,-z_0]\cup[z_0,c].
\end{align}
Through $D(z)$ and $G(z)$, in this region of $\xi,c$, same as above subsection we give series of transformations:
\begin{align}
	N\to M^{(1)} \to M^{(2)}=D_\infty^{-\sigma_3}e^{itg(\infty)\sigma_3}Ne^{it(p-_-g)\sigma_3}D^{\sigma_3},
\end{align}
which then satisfies the following RH problem.
\begin{RHP}\label{RHP8}
	Find a matrix-valued function  $  M^{(2)}(z )$ which satisfies
	
	$\blacktriangleright$ Analyticity: $M^{(2)}(z)$ is meromorphic in $\mathbb{C}\setminus \Sigma^{(2)}$;

	$\blacktriangleright$ Symmetry: $M^{(2)}(z)=\sigma_2M^{(2)}(-z)\sigma_2$=$\sigma_1\overline{M^{(2)}(\bar{z})}\sigma_1$;
	
	$\blacktriangleright$ Jump condition: $M^{(2)}$ has continuous boundary values $M^{(2)}_\pm(z)$ on $\Sigma^{(2)}$ and
	\begin{equation}
		M^{(2)}_+(z)=M^{(2)}_-(z)V^{(2)}(z),\hspace{0.5cm}z \in \Sigma^{(2)},
	\end{equation}
	where
	\begin{equation}
		V^{(2)}(z)=\left\{ \begin{array}{ll}
			\left(\begin{array}{cc}
				1 &\frac{ -r_2D^{-2}e^{-2itg}}{1-r_1r_2}\\
				0 & 1
			\end{array}\right),&\text{as } z\in\Sigma_1^+,\\[14pt]
			\left(\begin{array}{cc}
				1 &0\\
				\frac{r_1D^2e^{2itg}}{1-r_1r_2} & 1
			\end{array}\right),&\text{as } z\in\Sigma_1^-,\\[14pt]
			\left(\begin{array}{cc}
				1 & 0\\
				r_1D^2e^{2itg} & 1
			\end{array}\right),   &\text{as } z\in\Sigma_2^+;\\[12pt]
			\left(\begin{array}{cc}
				1 & -r_2D^{-2}e^{-2itg}\\
				0 & 1
			\end{array}\right),   &\text{as } z\in\Sigma_2^-;\\[12pt]
			\left(\begin{array}{cc}
				0 & ie^{-itB_2}\\
				ie^{itB_2} & 0
			\end{array}\right),   &\text{as } z\in[-c,-z_1] ,\\[12pt]
			\left(\begin{array}{cc}
				0 & ie^{-itB_2}\\
				ie^{itB_2} & \frac{D_-}{D_+}e^{itB_2}e^{-2itg_+}
			\end{array}\right),   &\text{as } z\in[-z_1,-z_0] ,\\[12pt]
			\left(\begin{array}{cc}
				0 & ie^{-itB_1}\\
				ie^{itB_1} & 0
			\end{array}\right),   &\text{as } z\in[-1,1],\\[12pt]	
			\left(\begin{array}{cc}
				0 & i\\
				i & \frac{D_-}{D_+}e^{-2itg_+}
			\end{array}\right),   &\text{as } z\in[z_0,z_1] ,\\[12pt]
			\left(\begin{array}{cc}
				0 & i\\
				i & 0
			\end{array}\right),   &\text{as } z\in[z_1,c],	
		\end{array}\right.;
	\end{equation}
	
	$\blacktriangleright$ Asymptotic behaviors: $M^{(2)}(z) = I+\mathcal{O}(z^{-1}),\hspace{0.5cm}z \rightarrow \infty; $

	$\blacktriangleright$ Singularity: $M^{(2)}(z)$ has singularity at $z=\pm c$ with:
	\begin{align}
		&M^{(2)}(z)\sim \mathcal{O}(z\mp c)^{-1/4} ,\ z\to \pm c,\ \pm 1\text{ in }\mathbb{C}\setminus\mathbb{R}.
	\end{align}
\end{RHP}

Except the  cut  away from $\mathbb{R}$, the jump $ {V}^{(2)}(z)$ exponentially approaches the identity matrix as $t\to\infty$. So we expect to only consider the jump on $\mathbb{R}$. However, different from above section, in this region, $g$ has another pair of stationary phase points on $\mathbb{R}$. So in this case, we
denote $ U(\xi)$ as the union set of neighborhood of $\pm z_0$ and $\pm z_2$:
\begin{equation}
	U(\xi)=U(\pm z_0)\cup U(\pm z_2),\ U(\pm z_j)= \left\lbrace z:|z\mp z_j|\leq \varrho \right\rbrace,\ j=0,2.
\end{equation}
Here, $\varrho$ is a small positive constant such that $\varrho<\min\left\lbrace \frac{z_0-1}{3}, \frac{z_1-z_0}{3},\frac{z_2-c}{3},\epsilon \right\rbrace $.
Thus, the jump matrix $V^{(2)}(z)$    uniformly goes to  $I$  on     $\Sigma^{(2)}\setminus  U(\xi)$.
So outside the $U(\xi)$ there is only exponentially small error (in $t$) by completely ignoring the jump condition of  $M^{(2)}(z)$.
And this proposition enlightens  us to construct the solution $M^{(2)}(z)$ as follow:
\begin{equation}
	M^{(2)}(z)= \left\{\begin{array}{ll}
		E(z;\xi,c)M^{mod}(z;\xi,c), & z\notin U(\xi),\\[4pt]
		E(z;\xi,c)M^{0,\pm}(z;\xi,c),  &z\in U(\pm z_0),\\[4pt]
		E(z;\xi,c)M^{mod}(z;\xi,c)M^{2,\pm}(z;\xi,c),  &z\in U(\pm z_2).\\
	\end{array}\right. \nonumber
\end{equation}
Here, same as  $M^{mod}(z)$ is the model  RH problem  on the Riemann surface, which solution is given by theta function in Subsection \ref{secmod}. The difference is $M^{j,\pm}(z)$ are local model of $\pm z_j$, $j=0,2$. When $j=0$, same as above subsection, its solution can be expressed in terms of Airy functions shown in Subsection \ref{seclo}. But when $j=2$, its solution can be expressed in terms of parabolic
cylinder shown in Subsection \ref{seclo2}. And  $E(z;\xi,c)$ is the error function, which will be discussed in subsection \ref{secer2}.
\subsubsection{Localized RH problem near phase points }\label{seclo2}
\quad
As $t\to +\infty$, we consider to reduce the  RHP \ref{RHP8} to a model RH problem  whose solution can be given explicitly in terms of parabolic cylinder
functions on every contour $\Sigma^{(0)}_\pm=\Sigma^{(2)}\cap U(\pm z_2)$ respectively. And we only give the details of $\Sigma^{(0)}_+$, the model of other  critical point can be  constructed similar. We denote $\hat{\Sigma}^{(0)}_+$ as the contour $\{z=z_2+le^{\pm\varphi i},\ l\in\mathbb{R}\}$ oriented from $\Sigma^{(0)}_1$, and $\hat{\Sigma}_{j}$ is the extension of $\Sigma_{j}$  respectively. And for $z$ near $z_2$, note that $g''(z_2)>0$, so we rewrite phase function $g$ as
\begin{align}
	g(z)=g(z_2)+(z-z_2)^2\frac{g''(z_2)}{2}+\mathcal{O}((z-z_2)^3).
\end{align}

\begin{figure}
	\centering
	\subfigure[]{
		\begin{tikzpicture}[node distance=2cm]
			\draw[->](0,0)--(2.5,1.4)node[right]{$\hat{\Sigma}_{3}$};
			\draw(0,0)--(-2.5,1.4)node[left]{$\hat{\Sigma}_{1}$};
			\draw(0,0)--(-2.5,-1.4)node[left]{$\hat{\Sigma}_{2}$};
			\draw[->](0,0)--(2.5,-1.4)node[right]{$\hat{\Sigma}_{4}$};
			\draw[dashed](-3.8,0)--(3.8,0)node[right]{\scriptsize Re$z$};
			\draw[->](-2.5,-1.4)--(-1.25,-0.7);
			\draw[->](-2.5,1.4)--(-1.25,0.7);
			\coordinate (A) at (-1.2,0.5);
			\coordinate (B) at (-1.2,-0.5);
			\coordinate (G) at (1.4,0.5);
			\coordinate (H) at (1.4,-0.5);
			\coordinate (I) at (0,0);
			\fill (A) circle (0pt) node[left] {\scriptsize$\left(\begin{array}{cc}
					1 &\frac{ -r_2D^{-2}e^{-2itg}}{1-r_1r_2}\\
					0 & 1
				\end{array}\right)$};
			\fill (B) circle (0pt) node[left] {\scriptsize$\left(\begin{array}{cc}
					1 &0\\
					\frac{r_1D^2e^{2itg}}{1-r_1r_2} & 1
				\end{array}\right)$};
			\fill (G) circle (0pt) node[right] {\scriptsize$\left(\begin{array}{cc}
					1 & 0\\
					r_1D^2e^{2itg} & 1
				\end{array}\right)$};
			\fill (H) circle (0pt) node[right] {\scriptsize$\left(\begin{array}{cc}
					1 & -r_2D^{-2}e^{-2itg}\\
					0 & 1
				\end{array}\right)$};
			\fill (I) circle (1pt) node[below] {$z_2$};
		\end{tikzpicture}
	}

	\caption{\footnotesize The contour $\hat{\Sigma}^{(0)}_+$ and the jump matrix on it near the point $z_2$.}
	\label{figS0}
\end{figure}
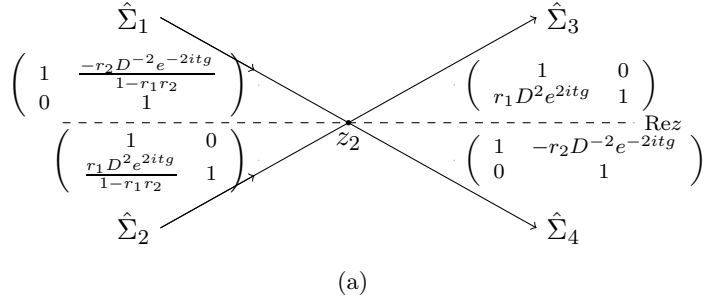
Consider following local RH problem:
\begin{RHP}\label{RHPlo1}
	Find a matrix-valued function  $ M^{2,+}(z)$ with following properties:
	
	$\blacktriangleright$ Analyticity: $M^{2,+}(z)$ is analytical  in $\mathbb{C}\setminus \hat{\Sigma}^{(0)}_+ $;

	$\blacktriangleright$ Jump condition: $M^{2,+}(z)$ has continuous boundary values $M^{2,+}_\pm$ on $\hat{\Sigma}^{(0)}_+$ and
	\begin{equation}
		M^{2,+}_+(z)=M^{2,+}_-(z)V^{2,+}(z),\hspace{0.2cm}z \in \hat{\Sigma}^{(0)}_+,
	\end{equation}
	where the jump matrix $V^{2,+}(z)$ is given  in Figure \ref{figS0}.

	$\blacktriangleright$ Asymptotic behaviors: $	M^{2,+}(z) =  I+\mathcal{O}(z^{-1}),\hspace{0.5cm}z \rightarrow \infty.$
\end{RHP}	

RHP \ref{RHPlo1} does not possess the symmetry condition shared by preceding RH  problem, because it is a local model and will only be used for bounded values of $z$.
In order to motivate the model, let $\zeta = \zeta(z)$ denote the rescaled local variable
\begin{align}
	\zeta(z)=2t^{1/2}\sqrt{g''(z_2)}(z-z_2).
\end{align}
 This map is a conformal bijection maps $U(z_2)$ to an expanding neighborhood of $\zeta= 0$. We choose the branch which maps the upper half plane to the lower half plane. Moreover, we denote:
 \begin{align}
 	r_{z_2}^\pm=r_1(\pm z_2)D^2_{\pm z_2}(\pm z_2)e^{2itg(\pm z_2)}(4tg''(\pm z_2))^{i\nu(\pm z_2)}.
 \end{align}
where, $D_{z_2}$ and $\nu(z_2)$ are defined in (\ref{D+}).

Through this change of variable,  the jump $V^{2,+}(z)$ approximates to  the jump of a parabolic cylinder model problem as follow:
\begin{RHP}\label{RHPpc}
	Find a matrix-valued function  $ M^{pc}(\zeta;\xi)$ with following properties:
	
	$\blacktriangleright$ Analyticity: $M^{pc}(\zeta;\xi)$ is analytical  in $\mathbb{C}\setminus \Sigma^{pc} $ with $\Sigma^{pc}=\left\lbrace\mathbb{R}e^{\varphi i} \right\rbrace \cup \left\lbrace\mathbb{R}e^{(\pi-\varphi) i} \right\rbrace$;

	$\blacktriangleright$ Jump condition: $M^{pc}$ has continuous boundary values $M^{pc}_\pm$ on $\Sigma^{pc}$ and
	\begin{equation}
		M^{pc}_+(\zeta;\xi)=M^{pc}_-(\zeta;\xi)V^{pc}(\zeta),\hspace{0.5cm}\zeta \in \Sigma^{\zeta},
	\end{equation}
	where
	\begin{align}
		V^{pc}(\zeta;\xi)=\left\{\begin{array}{ll}
			\left(\begin{array}{cc}
				1 & 0\\
				r_{z_2}^+\zeta^{2i\nu(z_2)}e^{\frac{i}{2}\zeta^2} & 1
			\end{array}\right),  & \zeta\in\mathbb{R}^+e^{\varphi i},\\[10pt]
			\left(\begin{array}{cc}
				1& -\bar{r}_{z_2}\zeta^{-2i\nu(z_2)}e^{-\frac{i}{2}\zeta^2}\\
				0&1
			\end{array}\right),   & \zeta\in \mathbb{R}^+e^{-\varphi i},\\[10pt]
			\left(\begin{array}{cc}
				1 & 0\\
				\frac{r_{z_2}^+}{1-|r_{z_2}^+|^2}\zeta^{2i\nu(z_2)}e^{\frac{i}{2}\zeta^2} & 1
			\end{array}\right),   & \zeta\in \mathbb{R}^+e^{(-\pi+\varphi) i},\\[10pt]
			\left(\begin{array}{cc}
				1 & -\frac{\bar{r}_{z_2}}{1-|r_{z_2}^+|^2}\zeta^{-2i\nu(z_2)}e^{-\frac{i}{2}\zeta^2}\\
				0 & 1
			\end{array}\right),   & \zeta\in \mathbb{R}^+e^{(\pi-\varphi) i}.
		\end{array}\right.
	\end{align}

	$\blacktriangleright$ Asymptotic behaviors: $M^{pc}(\zeta;\xi) =  I+M^{pc}_1\zeta^{-1}+\mathcal{O}(\zeta^{-2}),\hspace{0.5cm}\zeta \rightarrow \infty.$

\end{RHP}	

In  a   similar derivation   to the RHP\ref{RHPpc0},  we  can get
\begin{align}
	M^{2,\pm }(z)=I+\frac{t^{-1/2}}{z\mp z_2}\frac{i}{2\sqrt{g''(\pm z_2)}} \left(\begin{array}{cc}
		0 & \tilde{\beta}^\pm_{12}\\
		\tilde{\beta}^\pm_{21} & 0
	\end{array}\right)+\mathcal{O}(t^{-1}).\label{asyMpc}
\end{align}
where
\begin{align}
	&\tilde{\beta}^\pm_{12}=\frac{\sqrt{2\pi}e^{-\frac{1}{2}\pi\nu(\pm z_2)}e^{\frac{\pi i}{4} i}}{r_{z_2}^+\Gamma(-i\nu(\pm z_2))},\hspace{0.2cm}
\tilde{\beta}^\pm_{21}\tilde{\beta}^\pm_{12}=-\nu(\pm z_2).
\end{align}
We finally  obtain
\begin{Proposition}\label{asympc} $M^{2,\pm}(z)$ admits the following asymptotic expansion
	\begin{align}
		M^{2,\pm}(z)=I+t^{-1/2} \frac{A_\pm(\xi)}{z\mp z_2} +\mathcal{O}(t^{-1}), \ \ t\to+\infty,
	\end{align}
	where
	\begin{align}
		A_\pm(\xi)=\frac{i}{2\sqrt{g''(z_2)}} \left(\begin{array}{cc}
			0 & \tilde{\beta}^\pm_{12}\\
			\tilde{\beta}^\pm_{21} & 0
		\end{array}\right).\label{pcA}
	\end{align}
\end{Proposition}

\subsubsection{The small norm RH problem  for error function }\label{secer2}
\quad In this subsection,  we consider the error matrix-function $E(z;\xi,c)$ in this region.

\begin{RHP}\label{RHPE2}
	Find a matrix-valued function $E(z;\xi,c)$  with following properties:
	
	$\blacktriangleright$ Analyticity: $E(z;\xi,c)$ is analytical  in $\mathbb{C}\setminus  \Sigma^{E} $, where
	$$\Sigma^{E}= \partial U_{(\xi)}\cup
	\left[ \Sigma^{(2)}\setminus \left( U_{(\xi)}\cup[-c,-z_1]\cup[-1,1]\cup[z_1,c]\right) \right] ;$$

	$\blacktriangleright$ Asymptotic behaviors: $E(z;\xi,c) \sim I+\mathcal{O}(z^{-1}),\hspace{0.2cm}|z| \rightarrow \infty;$

	$\blacktriangleright$ Jump condition: $E(z;\xi,c)$ is continuous   $E_\pm(z;\xi,c)$ on $\Sigma^{E}$ satisfying
	$$E_+(z;\xi,c)=E_-(z;\xi,c)V^{E}(z),$$
	where the jump matrix $V^{E}(z)$ is given by
	\begin{equation}
		V^{E}(z)=\left\{\begin{array}{llll}
			M^{mod}(z)V^{(2)}(z)M^{mod}(z)^{-1}, & z\in \Sigma^{E}\setminus \partial U_{(\xi)},\\[4pt]
			M^{lo,\pm}(z)M^{mod}(z)^{-1},  & z\in \partial U_{\pm z_0},
			\\[4pt]
				M^{mod}(z)M^{2,\pm}(z)M^{mod}(z)^{-1},  & z\in \partial U_{\pm z_2},
		\end{array}\right. \label{deVE2}
	\end{equation}
	which is  shown in  Figure \ref{figE2}.
\end{RHP}

\begin{figure}[H]
	\centering
	\begin{tikzpicture}
		\draw[dashed](0,-1.5)--(0,1.5)node[above]{ Im$z$};
		\coordinate (I) at (0,0);
		\fill (I) circle (0pt) node[below] {$0$};
		\coordinate (a) at (2,0);
		\fill (a) circle (1pt) node[below] {$z_0$};
		\coordinate (aa) at (-2,0);
		\fill (aa) circle (1pt) node[below] {$-z_0$};
		\coordinate (b) at (0.8,0);
		\fill (b) circle (1pt) node[below] {$1$};
		\coordinate (ba) at (-0.8,0);
		\fill (ba) circle (1pt) node[below] {$-1$};
		\coordinate (c) at (3.65,0);
		\fill (c) circle (1pt) node[below] {$c$};
		\coordinate (ca) at (-3.65,0);
		\fill (ca) circle (1pt) node[below] {$-c$};
		\coordinate (cr) at (5,0);
		\fill (cr) circle (1pt) node[below] {$z_2$};
		\coordinate (car) at (-5,0);
		\fill (car) circle (1pt) node[below] {$-z_2$};
		\coordinate (a) at (2.8,0);
		\fill (a) circle (1pt) node[below] {$z_1$};
		\coordinate (aa) at (-2.8,0);
		\fill (aa) circle (1pt) node[below] {$-z_1$};
		\draw(-5,0)--(-6.5,1.5);
		\draw(5,0)--(6.5,1.5);
		\draw(-5,0)--(-6.5,-1.5);
		\draw(5,0)--(6.5,-1.5);
		\draw(-2,0)--(-1.4,0.4);
		\draw[-latex](-6,1)--(-5.5,0.5);
		\draw[-latex](-6,-1)--(-5.5,-0.5);
		\draw[-latex](5,0)--(5.5,0.5);
		\draw[-latex](5,0)--(5.5,-0.5);
		\draw(-0.8,0)--(-1.4,0.4);
		\draw(2,0)--(1.4,0.4);
		\draw(0.8,0)--(1.4,0.4);
		\draw(-2,0)--(-1.4,-0.4);
		\draw(-0.8,0)--(-1.4,-0.4);
		\draw(2,0)--(1.4,-0.4);
		\draw(0.8,0)--(1.4,-0.4);
		\draw[-latex](-1.4,0.4)--(-1.1,0.2);
		\draw[-latex](-2,0)--(-1.7,0.2);
		\draw[-latex](-1.4,-0.4)--(-1.1,-0.2);
		\draw[-latex](-2,0)--(-1.7,-0.2);
		\draw[-latex](0.8,0)--(1.1,0.2);
		\draw[-latex](1.4,0.4)--(1.7,0.2);
		\draw[-latex](0.8,0)--(1.1,-0.2);
		\draw[-latex](1.4,-0.4)--(1.7,-0.2);
		\draw(5,0)--(3.9,0.6);
		\draw(-5,0)--(-3.9,0.6);
		\draw(3.9,0.6)--(2.8,0);
		\draw(-3.9,0.6)--(-2.8,0);
		\draw[-latex](-5,0)--(-4.45,0.3);
		\draw[-latex](-5,0)--(-4.45,-0.3);
		\draw[-latex](3.9,0.6)--(4.45,0.3);
		\draw[-latex](3.9,-0.6)--(4.45,-0.3);
		\draw[-latex](-3.9,0.6)--(-3.35,0.3);
		\draw[-latex](-3.9,-0.6)--(-3.35,-0.3);
		\draw[-latex](2.8,0)--(3.35,0.3);
		\draw[-latex](2.8,0)--(3.35,-0.3);
		\draw(5,0)--(3.9,-0.6)node[below]{$\Sigma_2$};
		\draw(-5,0)--(-3.9,-0.6)node[below]{$\Sigma_2$};
		\draw(3.9,-0.6)--(2.8,0);
		\draw(-3.9,-0.6)--(-2.8,0);
		\draw[thick,red,fill=white](2,0) circle (0.15);
		\draw[thick,red,fill=white](-2,0) circle (0.15);
		\draw[-latex,red](2,0.15)--(2.05,0.15);
		\draw[-latex,red](-2,0.15)--(-1.93,0.15);
		\draw[thick,red,fill=white](5,0) circle (0.15);
		\draw[thick,red,fill=white](-5,0) circle (0.15);
		\draw[-latex,red](5,0.15)--(5.05,0.15);
		\draw[-latex,red](-5,0.15)--(-4.93,0.15);
		\draw[dashed](-6,0)--(6,0)node[right]{ Re$z$};
		\draw[dashed](-6.5,0)--(6.5,0);	
	\end{tikzpicture}
	\caption{  The jump contour $\Sigma^{E}$ for the $E(z;\xi,c)$. The red circles are $U(\xi)$. }
	\label{figE2}
\end{figure}
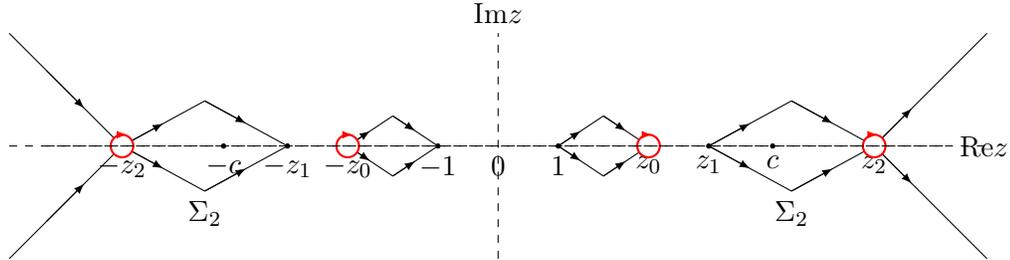

Similar with the discussion in Section \ref{erroranalysis},  the RHP  \ref{RHPE2} satisfies \begin{align}
	| V^{E}(z)-I|&=   \mathcal{O}(t^{-1/2}).\label{VE2}
\end{align} and admits  a unique  solution, which
can be given by
\begin{equation}
	E(z;\xi,c)=I+\frac{1}{2\pi i}\int_{\Sigma^{ E }}\dfrac{\left( I+\varpi(s)\right) (V^{E}(s)-I)}{s-z}ds,
\end{equation}
where the $\varpi\in L^\infty(\Sigma^{ E })$ is the unique solution of following equation
\begin{equation}
	(1-C_E)\varpi=C_E\left(I \right).
\end{equation}

In order to reconstruct the solution $u(y,t)$ of (\ref{mch}), we need the asymptotic behavior of $E(z;\xi,c)$ as $z\to 0\in \mathbb{C}^+$ and the long time asymptotic behavior of $E(0)$.
\begin{Proposition}\label{asyE2}
	As $z\to 0\in \mathbb{C}^+$, we have
	\begin{align}
		E(z;\xi,c)=E(0)+E_1z+\mathcal{O}(z^2),
	\end{align}
	where
	\begin{align}
		E(0)=I+\frac{1}{2\pi i}\int_{\Sigma^{E}}\dfrac{\left( I+\varpi(s)\right) (V^{E}-I)}{s}ds,
	\end{align}
	with long time asymptotic behavior
	\begin{equation}
		E(0)=I+t^{-1/2}H^{(0)}+\mathcal{O}(t^{-2}).\label{E0t2}
	\end{equation}
	And
	\begin{align}
		H^{(0)}=	H^{(0)}(\xi,c)=&\sum_{ p=\pm z_2 }\frac{M^{mod}(p)A_\pm(\xi)M^{mod}(p)^{-1}}{p}.
	\end{align}
	Here, $\tilde{B}_\pm$, $\tilde{A}_\pm$ is shown in (\ref{expf}).
	And
	\begin{equation}
		E_1=\frac{1}{2\pi i}\int_{\Sigma^{E}}\dfrac{\left( I+\varpi(s)\right) (V^{E}-I)}{s^2}ds,
	\end{equation}
	satisfying long time asymptotic behavior condition
	\begin{equation}
		E_1=t^{-1/2}H^{(1)}+\mathcal{O}(t^{-1}),\label{E1t2}
	\end{equation}
	where
	\begin{align}
		H^{(1)}=H^{(0)}(\xi,c)=&\sum_{ p=\pm z_2 }\frac{M^{mod}(p)A_\pm(\xi)M^{mod}(p)^{-1}}{p^2}.
	\end{align}
\end{Proposition}
\begin{proof}
	Substitute the long time asymptotic behavior of $V^{E}$, $\varpi(s)$ and Proposition \ref{asympc} into $2\pi i(E(0)-I)$:
	\begin{align}
		&\int_{\Sigma^{E}}\dfrac{\left( I+\varpi(s)\right) (V^{E}-I)}{s}ds\nonumber\\
		&=\int_{\partial U(\pm z_2)}\dfrac{ 	M^{mod}(s)(M^{2,\pm}(s)-I)M^{mod}(s)^{-1}}{s}ds\nonumber\\
		&+\mathcal{O}(t^{-1})\nonumber\\
		&=t^{-1/2}\int_{\partial U(\pm z_2)}\dfrac{ M^{mod}(s)A_\pm(\xi)M^{mod}(s)^{-1}}{s(z\mp z_2)}ds+\mathcal{O}(t^{-1}).
	\end{align}
	Then by residue theorem we finally arrive at the result.
\end{proof}
\section{The second-type genus-2 elliptic wave  region }\label{sec8}

  The  Region IV is corresponding to the case  $\frac{3}{4}<\xi<\xi_m$, $c>2$. Here, as denote in above section,  $\xi_m$ is the
 critical condition that stationary phase point $z_2$ merge $c$. Similarly, we need to construct new g-functions defined on genus 2
 Riemann surface which has real  branch points $\pm 1$, $\pm z_1$ and $\pm z_2$ with $z_1<z_2$. Different from above section, in Region  III, $z_1,$ $z_2$ both the stationary phase points of $g$. And the  canonical homology basis $\left\lbrace a_j,b_j \right\rbrace_{j=1}^2 $ is shown in Figure \ref{figab}.  In this two different cases, $g$ has different property. So after we proving the basic property of $g$, we will discuss separately.
 \begin{figure}[h]
 	\centering
 	\begin{tikzpicture}[node distance=2cm]
 		\draw[dashed](-6,0)--(6,0);
 		\draw(3.2,0)arc (0:180:3.2 and 1);
 		\draw[dashed](-3.2,0)arc (180:360:3.2 and 1);
 		\draw(1.1,0)arc (180:0:1 and 0.5);
 		\draw [->](0.01,1)--(0,1);
 		\draw[dashed](3.1,0)arc (0:-180:1 and 0.5);
 		\draw [->](2.06,0.5)--(2.05,0.5);
 		\draw [](4,0) ellipse (1 and 0.3);
 		\draw [->](4.01,0.3)--(4,0.3);
 		\draw [](0,0) ellipse (1.2 and 0.3);
 		\draw [->](0.01,0.3)--(0,0.3);	
 		\coordinate (a1) at (0,0.3);
 		\fill (a1) circle (0pt) node[above] {$a_2$};
 		\coordinate (a2) at (4,0.3);
 		\fill (a2) circle (0pt) node[above] {$a_2$};
 		\coordinate (b1) at (0,1);
 		\fill (b1) circle (0pt) node[above] {$b_2$};
 		\coordinate (r) at (2.1,0.5);
 		\fill (r) circle (0pt) node[below] {$b_1$};	
 		\coordinate (I) at (0,0);
 		\fill (I) circle (0pt) node[below] {$0$};
 		\coordinate (a) at (3,0);
 		\fill (a) circle (1pt) node[below] {$z_1$};
 		\coordinate (aa) at (-3,0);
 		\fill (aa) circle (1pt) node[below] {$-z_1$};
 		\coordinate (b) at (1.2,0);
 		\fill (b) circle (1pt) node[below] {$1$};
 		\coordinate (ba) at (-1.2,0);
 		\fill (ba) circle (1pt) node[below] {$-1$};
 		\coordinate (c) at (5,0);
 		\fill (c) circle (1pt) node[below] {$z_2$};
 		\coordinate (ca) at (-5,0);
 		\fill (ca) circle (1pt) node[below] {$-z_2$};
 		\draw[very thick](-5,0)--(-3,0);
 		\draw[very thick](5,0)--(3,0);
 		\draw[very thick](-1.2,0)--(1.2,0);
 	\end{tikzpicture}
 	\caption{\footnotesize  The  canonical homology basis $\left\lbrace a_j,b_j \right\rbrace_{j=1}^2 $ of the genius 2 Riemann surface. }
 	\label{figab2}
 \end{figure}
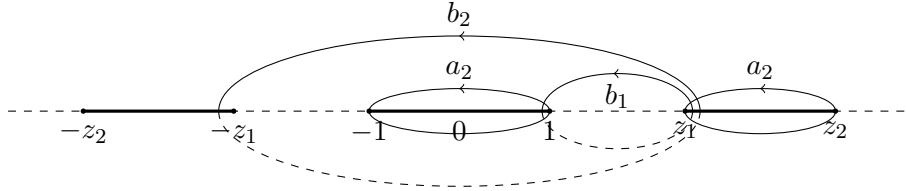

 \subsection{Constructing the $g$-function  }
 To construct the g-function, we first introduce 
 \begin{align}
 	Z(z)=\left[\dfrac{(z^2-z_1^2)(z^2-z_2^2)}{(z^2-1)} \right]^{1/2},
 \end{align}
Here the branch of the square root is such that $Z(z)\in i\mathbb{R}^+$ for $z\in[z_1,z_2]_+$. And $z_1$, $z_2$ admit:
\begin{align}
	\frac{1-\xi}{2}=\frac{1}{z_1z_2}\left(1-\frac{1}{z_1^2} -\frac{1}{z_2^2}\right) .\label{key}
\end{align}
 d$g$ is the  derivative of g-function:
\begin{align}
	\text{d}g=\dfrac{Z(z)}{z^3}\left[ \frac{1-\xi}{2}z^2-\frac{2}{z_1z_2}\right] \text{d}z.
\end{align}
$\text{d}\hat{g}$ is a meromorphic differential defined on the 2-genus Riemann surface, with d$g$ on the upper sheet and $-$d$g$ on the lower sheet.
Similarly,  the g-function is given by
\begin{align}
	g(z)=\int_{z_2}^{z}\text{d}g,\ z\in\mathbb{C}\setminus\Sigma^{mod}.
\end{align}

\begin{Proposition}	\label{prog2}
	There exist  a pair of real number $z_1=z_1(\xi,c)$, $z_2=z_2(\xi,c)$  in $(1,c)$ such that the function $g(z)$ defined above has the following properties:\\
	(a) The $a$-period of $g(z)$ is zero and the $b$-period of $g(z)$ is in $\mathbb{R}$;\\
	(b) the sign of Im$g$ has the same property in Figure \ref{figdg3};\\
	(c) $g(z)$ satisfies the following jump conditions across $[-z_2,z_2]$:
	\begin{align}
		&g_-(z)+g_+(z)=0,\hspace{0.5cm}z\in(z_1,z_2),\\
		&g_-(z)-g_+(z)=0,\hspace{0.5cm}z\in(1,z_1)\cup(-z_1,-1),\\
		&g_-(z)+g_+(z)=B_1,\hspace{0.5cm}z\in(-1,1),\\
		&g_-(z)+g_+(z)=B_2,\hspace{0.5cm}z\in(-z_2,-z_1),
	\end{align}
	here, $B_j=B_j(\xi)=\frac{1}{2}\oint_{b_j}dg$ is real;\\
	(d) $g(z)$ has another phase point $z_0=z_0(\xi)\in(z_1,z_2)$ which is the solution of equation $\frac{\xi-1}{2}z^2-\frac{2}{z_1z_2}=0$.\;\\
	(e) As $\xi\to\xi_m$, we have $z_2\to c$, while as $\xi\to\frac{3}{4}$, $z_1, z_2\to2$.
\end{Proposition}
\begin{figure}[h]
	\centering
	\begin{tikzpicture}[node distance=2cm]
		\draw[lime!30, fill=lime!30](-4.5,0)--(-2,0)arc (165:180:2.2 and 6)--(-4.4,-1.55)--(-4.5,0);
		\draw[lime!30, fill=lime!30](4.5,0)--(2,0)arc (15:0:2.2 and 6)--(4.4,-1.55)--(4.5,0);
		\draw[lime!30, fill=lime!30](2,0)arc (-15:0:2.2 and 6)--(-2,1.55)--(-2,0)--(2,0);
		\draw[lime!30, fill=lime!30](-2,0)arc (195:180:2.2 and 6)--(2,1.55)--(2,0)--(-2,0);
		\draw[lime!30, fill=lime!30](-6,0)--(-4.5,0)arc (-15:0:2.2 and 6)--(-6,1.55)--(-6,0);
		\draw[lime!30, fill=lime!30](6,0)--(4.5,0)arc  (195:180:2.2 and 6)--(6,1.55)--(6,0);
		\draw[dashed](-6,0)--(6,0);
		\draw(2,0)arc (-15:0:2.2 and 6);
		\draw(2,0)arc (15:0:2.2 and 6);
		\draw[](-2,0)arc (165:180:2.2 and 6);
		\draw[](-2,0)arc (195:180:2.2 and 6);
		\draw[](4.5,0)arc (165:180:2.2 and 6);
		\draw[](4.5,0)arc (195:180:2.2 and 6);
		\draw[](-4.5,0)arc (-15:0:2.2 and 6);
		\draw[](-4.5,0)arc (15:0:2.2 and 6);
		\coordinate (a1) at (0,1.2);
		\fill (a1) circle (0pt) node[below] {\footnotesize $\text{Im}g>0$};
		\coordinate (e1) at (-3,-0.7);
		\fill (e1) circle (0pt) node[below] {\footnotesize$\text{Im}g>0$};
		\coordinate (o1) at (3,-0.7);
		\fill (o1) circle (0pt) node[below] {\footnotesize$\text{Im}g>0$};	
		\coordinate (w1) at (0,-0.7);
		\fill (w1) circle (0pt) node[below] {\footnotesize$\text{Im}g<0$};
		\coordinate (e2) at (-3,1.2);
		\fill (e2) circle (0pt) node[below] {\footnotesize$\text{Im}g<0$};
		\coordinate (o2) at (3,1.2);
		\fill (o2) circle (0pt) node[below] {\footnotesize$\text{Im}g<0$};				
		\coordinate (I) at (0,0);
		\fill (I) circle (0pt) node[below] {\footnotesize$0$};
		\coordinate (a) at (2,0);
		\fill (a) circle (1pt) node[below] {\footnotesize$z_1$};
		\coordinate (aa) at (-2,0);
		\fill (aa) circle (1pt) node[below] {\footnotesize$-z_1$};
		\coordinate (b) at (1.2,0);
		\fill (b) circle (1pt) node[below] {\footnotesize$1$};
		\coordinate (ba) at (-1.2,0);
		\fill (ba) circle (1pt) node[below] {\footnotesize$-1$};
		\coordinate (c) at (3.5,0);
		\fill (c) circle (1pt) node[below] {\footnotesize$z_0$};
		\coordinate (ca) at (-3.5,0);
		\fill (ca) circle (1pt) node[below] {\footnotesize$-z_0$};
		\coordinate (cr) at (4.5,0);
		\fill (cr) circle (1pt) node[below] {\footnotesize$z_2$};
		\coordinate (car) at (-4.5,0);
		\fill (car) circle (1pt) node[below] {\footnotesize$-z_2$};
		\draw[very thick](-4.5,0)--(-2,0);
		\draw[very thick](4.5,0)--(2,0);
		\draw[very thick](-1.2,0)--(1.2,0);
		\coordinate (ww) at (5.2,1.2);
		\fill (ww) circle (0pt) node[below] {\footnotesize$\text{Im}g>0$};	
		\coordinate (www) at (5.2,-0.7);
		\fill (www) circle (0pt) node[below] {\footnotesize$\text{Im}g<0$};
		\coordinate (ww) at (-5.2,1.2);
		\fill (ww) circle (0pt) node[below] {\footnotesize$\text{Im}g<0$};	
		\coordinate (www) at (-5.2,-0.7);
		\fill (www) circle (0pt) node[below] {\footnotesize$\text{Im}g>0$};
	\end{tikzpicture}
	\caption{\footnotesize  In the purpler region, Im$g$>0 while in the white region, Im$g$<0. }
	\label{figdg3}
\end{figure}
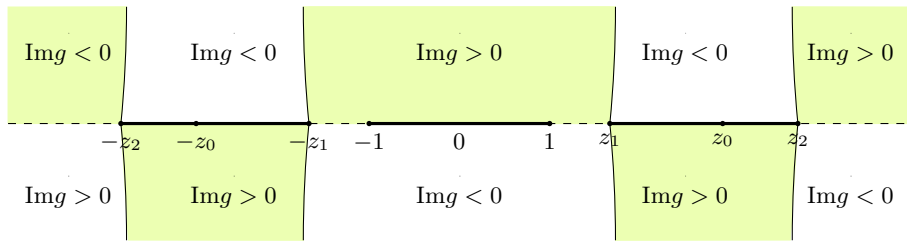
\begin{proof}
 Denote  $\eta=-\frac{2}{z_1z_2}$. Thus, (\ref{key}) gives
 \begin{align*}
 	z_1^2+z_2^2=\frac{4}{\eta^2}\left( 1+\frac{1-\xi}{\eta}\right) .
 \end{align*}
Then the $a_2$-period of $g$ equals to zero if and only if $F(\eta,\xi)=0$ with
\begin{align*}
	F(\eta,\xi)=\int_{z_1}^{z_2}\dfrac{Z(z)}{z^3}\left[ \frac{1-\xi}{2}z^2-\frac{2}{z_1z_2}\right] \text{d}z.
\end{align*}
When $\xi=\frac{3}{4}$, $F(\eta,\xi)=0$ has solution $(-\frac{1}{2},\frac{3}{4})$, and on the other end $\xi=\xi_m$, $F(\eta,\xi)=0$ has solution as shown in Proposition \ref{prog}: $(-\frac{2}{cz_0(\xi_m)},\xi_m)$. And
\begin{align*}
	\frac{\partial  F(\eta,\xi)}{\partial \eta}=\int_{z_1}^{z_2}\dfrac{z}{\sqrt{(z^2-1)(z^2-z_1^2)(z^2-z_2^2)}}\left[ 1+(1-\xi)\left(\frac{2}{\eta^3}+\frac{2}{\eta^4}(1-\xi) \right) \right] \text{d}z.
\end{align*}
Consider the function $f(x)=x^4+2(1-\xi)x+3(1-\xi)^2$. Note that,  (\ref{key}) implies $-\eta>1-\xi$, so simple calculation gives that $f(\eta)>f(\xi-1)>0$. So $	\frac{\partial  F(\eta,\xi)}{\partial \eta}\neq0$, which give the existence of solution $\eta$.
\end{proof}

\subsection{Opening the jump contour  }

Similarly as the above section, we define the following contour relying on $\xi, c$
\begin{align*}
	\Sigma^\pm=&\left\lbrace -z_1+e^{\pm\psi i}l,\ l\in(0,\frac{z_1-1}{2\cos\psi})\right\rbrace \cup\left\lbrace z_1+e^{(\pi\mp\psi) i}l,\ l\in(0,\frac{z_1-1}{2\cos\psi})\right\rbrace\nonumber\\
	&\cup\left\lbrace 1+e^{\pm\psi i}l,\ l\in(0,\frac{z_1-1}{2\cos\psi})\right\rbrace \cup\left\lbrace -1+e^{(\pi\mp\psi) i}l,\ l\in(0,\frac{z_1-1}{2\cos\psi})\right\rbrace \nonumber\\
	&\cup\left\lbrace -z_2+e^{(\pi\mp\psi) i}\mathbb{R}^+\right\rbrace \cup\left\lbrace z_2+e^{\pm\psi i}\mathbb{R}^+\right\rbrace.
\end{align*} 
    The   $\frac{\pi}{4}\geq\psi$ is a small enough positive constant such that $\Sigma^\pm $ 
are  contained in the region of Im$g>0$. And similarly as above equation, $\Omega^\pm$ is a closed region,  the edge of which is made up by $\Sigma^\pm$ and $\mathbb{R}$.
\begin{figure}[h]
	\centering
	\begin{tikzpicture}[node distance=2cm]
		\draw[lime!30, fill=lime!30](-5,0)--(-2,0)arc (165:180:2.2 and 6)--(-4.9,-1.55)--(-5,0);
		\draw[lime!30, fill=lime!30](5,0)--(2,0)arc (15:0:2.2 and 6)--(4.9,-1.55)--(5,0);
		\draw[lime!30, fill=lime!30](2,0)arc (-15:0:2.2 and 6)--(-2,1.55)--(-2,0)--(2,0);
		\draw[lime!30, fill=lime!30](-2,0)arc (195:180:2.2 and 6)--(2,1.55)--(2,0)--(-2,0);
		\draw[lime!30, fill=lime!30](-6.5,0)--(-5,0)arc (-15:0:2.2 and 6)--(-6.5,1.55)--(-6.5,0);
		\draw[lime!30, fill=lime!30](6.5,0)--(5,0)arc  (195:180:2.2 and 6)--(6.5,1.55)--(6.5,0);
		\draw[dashed](-6.5,0)--(6.5,0);	
		\coordinate (I) at (0,0);
		\fill (I) circle (0pt) node[below] {$0$};
		\coordinate (a) at (2,0);
		\fill (a) circle (1pt) node[below] {$z_1$};
		\coordinate (aa) at (-2,0);
		\fill (aa) circle (1pt) node[below] {$-z_1$};
		\coordinate (b) at (0.8,0);
		\fill (b) circle (1pt) node[below] {$1$};
		\coordinate (ba) at (-0.8,0);
		\fill (ba) circle (1pt) node[below] {$-1$};
		\coordinate (cr) at (5,0);
		\fill (cr) circle (1pt) node[below] {$z_2$};
		\coordinate (car) at (-5,0);
		\fill (car) circle (1pt) node[below] {$-z_2$};
		\draw(-5,0)--(-6.5,1.5)node[above]{\small $\Sigma^+$};
		\draw(5,0)--(6.5,1.5)node[above]{\small $\Sigma^+$};
		\draw(-5,0)--(-6.5,-1.5)node[below]{\small $\Sigma^-$};
		\draw(5,0)--(6.5,-1.5)node[below]{\small $\Sigma^-$};
		\draw(-2,0)--(-1.4,0.4)node[above]{\small $\Sigma^+$};
		\draw[-latex](-6,1)--(-5.5,0.5);
		\draw[-latex](-6,-1)--(-5.5,-0.5);
		\draw[-latex](5,0)--(5.5,0.5);
		\draw[-latex](5,0)--(5.5,-0.5);
		\draw(-0.8,0)--(-1.4,0.4);
		\draw(2,0)--(1.4,0.4)node[above]{\small $\Sigma^+$};
		\draw(0.8,0)--(1.4,0.4);
		\draw(-2,0)--(-1.4,-0.4)node[below]{\small $\Sigma^-$};
		\draw(-0.8,0)--(-1.4,-0.4);
		\draw(2,0)--(1.4,-0.4)node[below]{\small $\Sigma^-$};
		\draw(0.8,0)--(1.4,-0.4);
		\draw[-latex](-1.4,0.4)--(-1.1,0.2);
		\draw[-latex](-2,0)--(-1.7,0.2);
		\draw[-latex](-1.4,-0.4)--(-1.1,-0.2);
		\draw[-latex](-2,0)--(-1.7,-0.2);
		\draw[-latex](0.8,0)--(1.1,0.2);
		\draw[-latex](1.4,0.4)--(1.7,0.2);
		\draw[-latex](0.8,0)--(1.1,-0.2);
		\draw[-latex](1.4,-0.4)--(1.7,-0.2);
		\draw[-latex](-4.01,0)--(-4,0);
		\draw[-latex](4,0)--(4.01,0);
		\draw[-latex](-0.01,0)--(0.01,0);
		\node at (6,0.5) {\footnotesize $\Omega^+$};
		\node at (-6,0.5) {\footnotesize $\Omega^+$};
		\node at (1.42,0.15) {\footnotesize $\Omega^+$};
		\node at (-1.42,0.15) {\footnotesize $\Omega^+$};
		\node at (1.42,-0.15) {\footnotesize $\Omega^-$};
		\node at (-1.42,-0.15) {\footnotesize $\Omega^-$};
		\node at (6,-0.5) {\footnotesize $\Omega^-$};
		\node at (-6,-0.5) {\footnotesize $\Omega^-$};
	\end{tikzpicture}
	\caption{\footnotesize  The region of $\Omega^\pm$.  The same as Figure \ref{figdg2}, purple region means Im$g(z)>0$ while white region means Im$g(z)<0$.}
	\label{figfj4}
\end{figure}
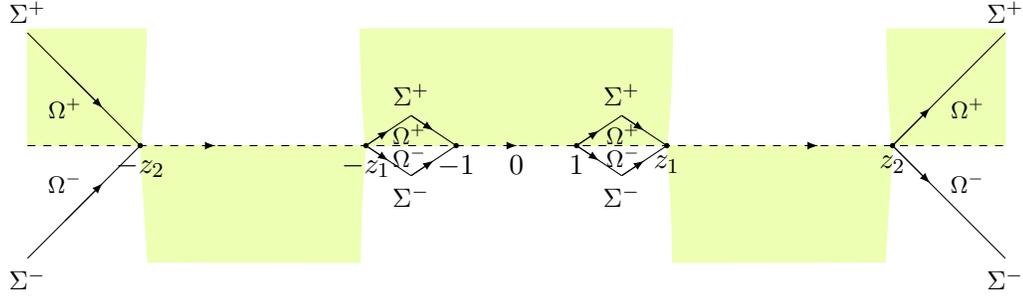

In this region of $\xi,c$, we introduce
a piecewise matrix interpolation function
\begin{equation}
	G(z)= G (z;\xi,c)=\left\{ \begin{array}{ll}
		\left(\begin{array}{cc}
			1 & 0\\
			-r_1e^{2itg} & 1
		\end{array}\right),   &\text{as } z\in\Omega^+;\\[12pt]
		\left(\begin{array}{cc}
			1 & -r_2e^{-2itg}\\
			0 & 1
		\end{array}\right),   &\text{as } z\in\Omega^-;\\
		I &\text{as } 	z \text{ in elsewhere},
	\end{array}\right. \label{funcG4}
\end{equation}
Same as above section,  $ G (z)$ bring a new singularity.
To deal with the jump on $\mathbb{R}$, we give a introduction of an auxiliary function $D(z)$, which admits the following jump condition:
\begin{align*}
	&D_-(z)D_+(z)=i[r_2]_-,\hspace*{0.5cm}z\in[-z_2,z_2] \setminus[-z_1,z_1];\\
	&D_-(z)D_+(z)=1,\hspace*{0.5cm}z\in[-1,1].
\end{align*}
Define
\begin{align}
	Z_3(z)=&(z^2-1)Z(z),\\
	\log D(z)=&\frac{Z_3(z)}{2\pi i}\left(\int_{-z_2}^{-z_1}+\int_{z_2}^{z_1} \right) \dfrac{\log(i[r_2]_-(s))}{(s-z)[Z_3]_+(s)}ds.
\end{align}
\begin{Proposition}\label{proD3}
	The scalar function $D(z)$ satisfies the following properties :\\
	(a) $D(z)$ is analytic on $\mathbb{C}\setminus\left( (-z_2,-z_1)\cup[-1,1]\cup(z_1,z_2)\right) $;\\
	(b) $D(z)$ has singularity at $z=\pm c$ with:
	\begin{align}
		D(z)=\mathcal{O}(z-p)^{\mp 1/4},\hspace{0.3cm}z\to p\in \mathbb{C}^\pm\setminus\mathbb{R} ,\hspace{0.3cm}p=c,-c.
	\end{align}
	(c) As $z\to\infty\in\mathbb{C}\setminus\mathbb{R}$, $D(z)$ has limit  $D_\infty(z)$ with
	\begin{align}
		\log D_\infty(z)=&-\frac{1}{2\pi i}\left(\int_{-z_2}^{-z_1}+\int_{z_2}^{z_1} \right) \dfrac{\log(i[r_2]_-(s))}{[X]_+(s)}\left(z^2+zs+s^2-\frac{1+z_2^2+z_1^2}{2} \right) ds.
	\end{align}

	(d) As $z\to0\in\mathbb{C}^+$,
	\begin{align}
		D_\infty(z)=&D_\infty(0)\left(1 +D_\infty^{(1)}z\right) +\mathcal{O}(z^2),\nonumber\\
		D(z)=&D(0)\left( 1+D^{(1)}z\right) +\mathcal{O}(z^2),\nonumber
	\end{align}
	where
{\small 	\begin{align}
		&\log \left( D_\infty(0)\right) = \frac{1}{2\pi i}\left(\int_{-z_2}^{-z_1}+\int_{z_2}^{z_1} \right) \dfrac{\log(i[r_2]_-(s))}{[X]_+(s)}\left(\frac{1+z_2^2+z_1^2}{2} -s^2\right) ds,\\
	&	D_\infty^{(1)}= -\frac{1}{2\pi i}\left(\int_{-z_2}^{-z_1}+\int_{z_2}^{z_1} \right) \dfrac{s\log(i[r_2]_-(s))}{[X]_+(s)}ds,\\
	&	\log \left( D(0)\right) = \frac{cz_1}{2\pi }\left(\int_{-z_2}^{-z_1}+\int_{z_2}^{z_1} \right) \dfrac{\log(i[r_2]_-(s))}{s[Z_3]_+(s)}ds,\\
	&	D^{(1)}= \frac{cz_1}{2\pi }\left(\int_{-z_2}^{-z_1}+\int_{z_2}^{z_1} \right) \dfrac{\log(i[r_2]_-(s))}{s^2[Z_3]_+(s)}ds.
	\end{align}}
\end{Proposition}
Denote
\begin{align}
	\Sigma^{(2)}=\left(\cup_{j=1}^2\Sigma_j^\pm\right) \cup[-1,1]\cup[-z_2,-z_1]\cup[z_1,z_2].
\end{align}
Through $D(z)$ and $G(z)$, in this region of $\xi,c$, same as above subsection we give series of transformations:
Thus, we can give the same transformation in this case with
\begin{align*}
	N(z)\to M^{(1)}(z)\to  	M^{(2)}(z)=\left\{\begin{array}{ll}
		E(z;\xi,c)M^{mod}(z;\xi,c) & z\notin U(\xi)\\
		E(z;\xi,c)M^{j,pm}(z;\xi,c)  &z\in U(\pm z_j),\ j=1,2,\\
	\end{array}\right.
\end{align*}
where   $M^{mod}(z)$ is the model  RH problem  on the Riemann surface, which solution is given by theta function in Subsection \ref{secmod2}. $M^{j,pm}(z)$ are local model of $\pm z_j$ which solution can be expressed in terms of Airy functions similarly in Subsection \ref{seclo}. And  $E(z;\xi,c)$ is the error function, which has jump contour in Figure \ref{figfE}, and it is similarly in subsection \ref{secer}. We obtain  its  asymptotic property   directly as follow:
\begin{figure}[h]
	\centering
	\begin{tikzpicture}[node distance=2cm]
		\coordinate (I) at (0,0);
		\fill (I) circle (0pt) node[below] {$0$};
		\draw(-5,0)--(-6.5,1.5);
		\draw(5,0)--(6.5,1.5);
		\draw(-5,0)--(-6.5,-1.5);
		\draw(5,0)--(6.5,-1.5);
		\draw(-2,0)--(-1.4,0.4);
		\draw[-latex](-6,1)--(-5.5,0.5);
		\draw[-latex](-6,-1)--(-5.5,-0.5);
		\draw[-latex](5,0)--(5.5,0.5);
		\draw[-latex](5,0)--(5.5,-0.5);
		\draw(-0.8,0)--(-1.4,0.4);
		\draw(2,0)--(1.4,0.4);
		\draw(0.8,0)--(1.4,0.4);
		\draw(-2,0)--(-1.4,-0.4);
		\draw(-0.8,0)--(-1.4,-0.4);
		\draw(2,0)--(1.4,-0.4);
		\draw(0.8,0)--(1.4,-0.4);
		\draw[-latex](-1.4,0.4)--(-1.1,0.2);
		\draw[-latex](-2,0)--(-1.7,0.2);
		\draw[-latex](-1.4,-0.4)--(-1.1,-0.2);
		\draw[-latex](-2,0)--(-1.7,-0.2);
		\draw[-latex](0.8,0)--(1.1,0.2);
		\draw[-latex](1.4,0.4)--(1.7,0.2);
		\draw[-latex](0.8,0)--(1.1,-0.2);
		\draw[-latex](1.4,-0.4)--(1.7,-0.2);
		\draw[thick,red,fill=white](5,0) circle (0.2);
		\draw[thick,red,fill=white](-5,0) circle (0.2);
		\draw[-latex,red](5,0.2)--(5.05,0.2);
		\draw[-latex,red](-5,0.2)--(-4.93,0.2);
		\draw[thick,red,fill=white](2,0) circle (0.2);
		\draw[thick,red,fill=white](-2,0) circle (0.2);
		\draw[-latex,red](2,0.2)--(2.05,0.2);
		\draw[-latex,red](-2,0.2)--(-1.93,0.2);
			\draw[dashed](-6.5,0)--(6.5,0);	
					\coordinate (a) at (2,0);
			\fill (a) circle (1pt) node[below] {$z_1$};
			\coordinate (aa) at (-2,0);
			\fill (aa) circle (1pt) node[below] {$-z_1$};
			\coordinate (b) at (0.8,0);
			\fill (b) circle (1pt) node[below] {$1$};
			\coordinate (ba) at (-0.8,0);
			\fill (ba) circle (1pt) node[below] {$-1$};
			\coordinate (cr) at (5,0);
			\fill (cr) circle (1pt) node[below] {$z_2$};
			\coordinate (car) at (-5,0);
			\fill (car) circle (1pt) node[below] {$-z_2$};
	\end{tikzpicture}
	\caption{\footnotesize  The jump contour $\Sigma^{E}$ for the $E(z;\xi,c)$. The red circles are $U(\xi)$.}
	\label{figfE}
\end{figure}
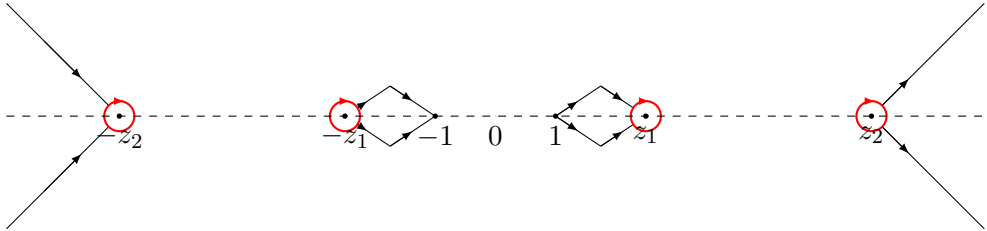

\begin{Proposition}\label{asyE3}
	As $z\to 0\in \mathbb{C}^+$, we have
	\begin{align}
		E(z;\xi,c)=E(0)+E_1z+\mathcal{O}(z^2),
	\end{align}
	where
	\begin{align}
		E(0)=I+\frac{1}{2\pi i}\int_{\Sigma^{E}}\dfrac{\left( I+\varpi(s)\right) (V^{E}-I)}{s}ds,
	\end{align}
	with long time asymptotic behavior
	\begin{equation}
		E(0)=I+t^{-1}H^{(0)}+\mathcal{O}(t^{-2}).\label{E0t}
	\end{equation}
	And
	{\small \begin{align}
		H^{(0)}=H^{(0)}(\xi,c)=&\sum_{ p=\pm z_j,j=1,2 }\frac{\text{d}}{\text{d}z}\left(\frac{1 }{z}F^{\pm}\left(\begin{array}{cc}
			0&-\frac{5i}{48}[\tilde{A}^j_\pm]^{-\frac{4}{3}}\\
			0&0\\
		\end{array}\right)(F^\pm)^{-1}\right)(p) \nonumber\\
		&+\sum_{ p=\pm z_j,j=1,2}\frac{1 }{p}\left( F^{\pm}\left(\begin{array}{cc}
			0&\frac{5}{36}\tilde{B}^j_\pm[\tilde{A}_\pm^j]^{-\frac{7}{3}}\\
			-\frac{7i}{48}[\tilde{A}_\pm^j]^{-\frac{2}{3}}&0\\
		\end{array}\right)(F^\pm)^{-1}\right)(p),\nonumber
	\end{align}}
	where   $\tilde{B}^j_\pm$ and  $\tilde{A}^j_\pm$  are  shown in (\ref{expf2}).
Further, $E_1$ admits   the following asymptotic expansion
	\begin{equation}
		E_1=t^{-1}H^{(1)}+\mathcal{O}(t^{-2}),
	\end{equation}
	where
	{\footnotesize\begin{align}
		H^{(1)}=H^{(1)}(\xi,c)=&\sum_{ p=\pm z_j,j=1,2 }\frac{\text{d}}{\text{d}z}\left(\frac{1 }{z^2}F^{\pm}\left(\begin{array}{cc}
			0&-\frac{5i}{48}[\tilde{A}_\pm^j]^{-\frac{4}{3}}\\
			0&0\\
		\end{array}\right)(F^\pm)^{-1}\right)(p) \nonumber\\
		&+\sum_{ p=\pm z_j,j=1,2}\frac{1 }{p^2}\left( F^{\pm}\left(\begin{array}{cc}
			0&\frac{5}{36}\tilde{B}^j_\pm[\tilde{A}_\pm^j]^{-\frac{7}{3}}\\
			-\frac{7i}{48}[\tilde{A}^j_\pm]^{-\frac{2}{3}}&0\\
		\end{array}\right)(F^\pm)^{-1}\right)(p). \nonumber
	\end{align}}

\end{Proposition}
Here   we  denote
\begin{align}
	F^\pm=M^{mod}H(z;\pm z_0)^{-1}N^{-1}f_\pm^{\frac{\sigma_3}{2}}=[F^\pm_{ij}]_{2\times2},
\end{align}
with $(F^\pm)^{-1}\triangleq[F^{\pm,*}_{ij}]_{2\times2}$.
\begin{align}
	\frac{3}{2}itg(z)=\tilde{A}_{+}^1\sqrt{z-z_1}^3+\tilde{B}_{+}^1\sqrt{z-z_1}^4+\mathcal{O}(\sqrt{z-z_1}^5),
\end{align},where the definition of the square root is mapping positive real number to positive real number,  and
$\tilde{A}^1_-,\tilde{B}_-^1,\tilde{A}^2_+,\tilde{B}_+^2,\tilde{A}^2_-,\tilde{B}_-^2$ are as the  similar definition on $-z_1,z_2,-z_2$ respectively.
{\footnotesize\begin{align}
	\tilde{A}^1_+=&\frac{t}{2}\left( \frac{2z_1(z_2^2-z_1^2)}{(z_1^2-1)}\right) ^{\frac{1}{2}}\left( \frac{1-\xi}{z_1}-\frac{4}{z_1^4z_2}\right),	\label{expf2}\\[4pt]
	\tilde{B}_+^1=&\frac{3t}{10}\left( \frac{2z_1(z_2^2-z_1^2)}{(z_1^2-1)}\right) ^{\frac{1}{2}}
	\left[  -\frac{1-\xi}{z_1^2}+\frac{12}{z_1^5z_2}+ \left(\frac{1-\xi}{z_1}-\frac{4}{z_1^4z_2}\right)\left(\frac{1}{4z_1}+\frac{z_1}{z_1^2-z_2^2}-\frac{2z_1}{z_1^2-1}\right)\right] ,\nonumber\\[4pt]
	\tilde{A}^2_+=&\frac{it}{2}\left( \frac{2z_2(z_2^2-z_1^2)}{(z_2^2-1)}\right) ^{\frac{1}{2}}\left( \frac{1-\xi}{z_2}-\frac{4}{z_1z_2^4}\right) ,\nonumber\\[4pt]
	\tilde{B}_+^2=&\frac{3it}{10}\left( \frac{2z_1(z_2^2-z_1^2)}{(z_2^2-1)}\right) ^{\frac{1}{2}}
	\left[  -\frac{1-\xi}{z_2^2}+\frac{12}{z_1z_2^5}+ \left(\frac{1-\xi}{z_2}-\frac{4}{z_1z_2^4}\right)\left(\frac{1}{4z_2}+\frac{z_2}{z_2^2-z_1^2}-\frac{2z_2}{z_2^2-1}\right)\right] ,\nonumber\\[4pt]
	\tilde{A}^1_-=&- i \tilde{A}^1_+, \ \ \  \tilde{B}_-^1=i\tilde{B}_+^1, \ \ \tilde{A}^2_-=   i \tilde{A}^2_+, \ \ \ \tilde{B}_-^2=-i\tilde{B}_+^2.\nonumber
\end{align}}
\subsection{Model RH problem on Riemann surface}\label{secmod2}
\quad Similarly to Subsection \ref{secmod}, we arrive at the following model RH  problem
\begin{RHP}\label{RHP9}
	Find a matrix-valued function $M^{mod}(z)$  with following identities:
	
	$\blacktriangleright$ Analyticity: $M^{mod}(z)$ is analytical  in $\mathbb{C}\setminus  \Sigma^{cut} $, with $$\Sigma^{cut}=[-z_2,-z_1]\cup[-1,1]\cup[z_1,z_2];$$

	$\blacktriangleright$ Asymptotic behaviors: $M^{mod}(z) \sim I+\mathcal{O}(z^{-1}),\hspace{0.5cm}|z| \rightarrow \infty;$

	$\blacktriangleright$ Jump condition: $M^{mod}(z)$  satisfies the jump relation
	$$M^{mod}_+(z)=M^{mod}_-(z)V^{mod}(z), \ \ \ z\in \Sigma^{cut},$$
	where the jump matrix $V^{mod}(z) $ is given by
	\begin{equation}
		\hat{V}^{mod}(z)=\left\{ \begin{array}{ll}
			\left(\begin{array}{cc}
				0 & ie^{-itB_2}\\
				ie^{itB_2} & 0
			\end{array}\right),   &\text{as } z\in[-z_2,-z_1] ,\\[12pt]
			\left(\begin{array}{cc}
				0 & ie^{-itB_1}\\
				ie^{itB_1} & 0
			\end{array}\right),   &\text{as } z\in[-1,1],\\[12pt]	
			\left(\begin{array}{cc}
				0 & i\\
				i & 0
			\end{array}\right),   &\text{as } z\in[z_1,z_2];	
		\end{array}\right.
	\end{equation}
	
	$\blacktriangleright$ Singularity: $M^{mod}(z)$ has singularity at $z=\pm z_2$ with:
	\begin{align}
		&M^{mod}(z)\sim \mathcal{O}(z\mp p)^{-1/4} ,\ z\to p=\pm z_2,\ \pm 1,\ \pm z_1\text{ in }\mathbb{C}\setminus\mathbb{R}.
	\end{align}
\end{RHP}
$M^{mod}$can be derived by the so-called $\vartheta$ function on the Riemann surface of genus 2. To construct the model RH problem $M^{mod}$ we further need to let
\begin{align}
	\kappa(z)&=\left[ \frac{(z-z_2)(z-1)(z+z_1)}{(z-z_1)(z+1)(z+z_2)}\right] ^{\frac{1}{4}}, z\in\mathbb{C}\setminus\Sigma^{cut},\\
	\kappa(z)&=1+\mathcal{O}(z^{-1}),z\to\infty,\\
	\mathcal{N}(z)&=\frac{1}{2}\left(\begin{array}{cc}
		\kappa+\kappa^{-1}&\kappa-\kappa^{-1}\\
		\kappa-\kappa^{-1}&\kappa+\kappa^{-1}\\
	\end{array}\right).
\end{align}
Then  there is a constant $\mathcal{K}\in \mathbb{C}^2$ satisfies for arbitrary divisor $\mathcal{P}_0$
\begin{align*}
	\vartheta(\mathcal{A}(P)-\mathcal{A}(\mathcal{P}_0)-\mathcal{K}):= \vartheta(\mathcal{A}(P)-K)
\end{align*}
has $g=2$ zeros $P_1,P_2$ on $\mathcal{M}$ with $P_1+P_2=\mathcal{P}_0$.
Observing if $P_1,P_2$ is zeros of $\mathcal{N}_{11},\mathcal{N}_{22}$ only if $P'_1,P'_2$ is zeros of $\mathcal{N}_{12},\mathcal{N}_{21}$ where $P'_1,P'_2$ are the same point of $P_1,P_2$ on the other sheet. Therefore, $\vartheta(\mathcal{A}(P)-K),\
\vartheta(\mathcal{A}(P)+K)$
have the same zeros as $\mathcal{N}_{11},\ \mathcal{N}_{22}$ and $\mathcal{N}_{12},\ \mathcal{N}_{21}$ have respectively at the time of $\mathcal{P}_0=P_1+P_2$.
We then can show that the    RHP \ref{RHP9} admits the solution
\begin{align}
	M^{mod}(z)&= F(\infty) \left(\begin{array}{cc}
	\mathcal{N}_{11}\frac{\vartheta(\mathcal{A}(z)-K+C)}{\vartheta(\mathcal{A}(z)-K)}&
	\mathcal{N}_{12}\frac{\vartheta(-\mathcal{A}(z)-K+C)}{\vartheta(-\mathcal{A}(z)-K)}\\[10pt]
	\mathcal{N}_{21}\frac{\vartheta(\mathcal{A}(z)+K+C)}{\vartheta(\mathcal{A}(z)+K)}&
	\mathcal{N}_{22}\frac{\vartheta(-\mathcal{A}(z)+K+C)}{\vartheta(-\mathcal{A}(z)+K)}\\
\end{array}\right),\label{Mmod2}
\end{align}
where
$$ F(\infty)= \frac{1}{2}{\rm diag}  \left(  \frac{\vartheta(\mathcal{A}(\infty)-K)}{\vartheta(\mathcal{A}(\infty)-K+C)},
\frac{\vartheta(\mathcal{A}(\infty)-K)}{\vartheta(\mathcal{A}(\infty)-K-C)} \right).  $$
Noting that  as $z\to0\in\mathbb{C}^+$,
\begin{align}
		\mathcal{N}(z)=\frac{\sqrt{2}}{2}
	\left(\begin{array}{cc}
		1 & i\\
		i & 1\\
	\end{array}\right)+z\frac{\sqrt{2}}{4}\left(\frac{1}{z_1}-\frac{1}{z_2}-1 \right) \left(\begin{array}{cc}
		i & 1\\
		1 & i\\
	\end{array}\right)+\mathcal{O}(z^2),
\end{align}
then we have
\begin{align}
	M^{mod}(z)=M^{mod}(0)+M^{mod}_1z+\mathcal{O}(z^2),
\end{align}
where
\begin{align}
&M^{mod}_1=\frac{\sqrt{2}}{4}P(\infty)\left(\frac{1}{z_1}-\frac{1}{z_2}-1 \right) \left(\begin{array}{cc}
	i & 1\\
	1 & i\\
\end{array}\right) G(0)+  \frac{\sqrt{2}}{2}
	\left(\begin{array}{cc}
	1 & i\\
	i & 1\\
	\end{array}\right) G_z(0),\nonumber
\end{align}
and
$$G(z)=\left(\begin{array}{cc}
	 \frac{\vartheta(\mathcal{A}(z)-K+C)}{\vartheta(\mathcal{A}(z)-K)}&
	 \frac{\vartheta(-\mathcal{A}(z)-K+C)}{\vartheta(-\mathcal{A}(z)-K)}\\[10pt]
 \frac{\vartheta(\mathcal{A}(z)+K+C)}{\vartheta(\mathcal{A}(z)+K)}&
 \frac{\vartheta(-\mathcal{A}(z)+K+C)}{\vartheta(-\mathcal{A}(z)+K)}\\
\end{array}\right).  $$

\section{Long-time asymptotics   for the  mCH equation }\label{sec9}

\quad  In this section,  we give our main result  on  the long-time asymptotics of the mCH equation (\ref{mcho}),
which is discuss  as follows.

 For the region I,  we have  sequence of transformations
\begin{align}
	N(z)\to  M^{(1)}(z)\to M^{mod1}(z),
\end{align}
from which we obtain
{\small \begin{align*}
	N(z)=&E(z;\xi,c)M^{mod1}(z)\delta^{-\sigma_3}G(z)^{-1} =\sqrt{2}(I+H^{(0)}t^{-\frac{1}{2}})\left( \begin{array}{ll}
		0&i\\
		i&0
	\end{array}\right) \exp(-I_{\delta}^1\sigma_3)\\
&+z\sqrt{2}(I+H^{(0)}t^{-\frac{1}{2}})\left( \begin{array}{ll}
	0&i\\
	i&0
\end{array}\right) \exp(-I_{\delta}^1\sigma_3)I_{\delta}^2\sigma_3 +z\frac{i}{\sqrt{2}}(I+H^{(0)}t^{-\frac{1}{2}})\exp(-I_{\delta}^1\sigma_3)\\
&+z\sqrt{2}H^{(1)}\left( \begin{array}{ll}
	0&i\\
	i&0
\end{array}\right)\exp(-\I_{\delta}^1\sigma_3)+\mathcal{O}(z^2)+\mathcal{O}(t^{-1}),
\end{align*}}
where the $I^1_{\delta}$ and  $I^2_{\delta}$ are given by   (\ref{Ide0})-(\ref{Ide0}), $H^{(0)},H^{(1)}$ is in Prop. \ref{AsyE0I}.
We further calculate
{\small \begin{align*}
&	-i[N(0)]_{12}= \sqrt{2}e^{I_{\delta}^1}(1+H_{11}^{(0)}t^{-\frac{1}{2}})+\mathcal{O}(t^{-1}),\\
 &-i\lim_{z\to 0\in\mathbb{C}^+}\frac{[N(z)]_{11}-[N(0)]_{11}}{z}= \frac{e^{-I_{\delta}^1}}{\sqrt{2}}+\sqrt{2}H_{12}^{(1)}e^{I_{\delta}^1} +(\sqrt{2}H_{12}^{(0)}e^{-I_{\delta}^1}I_{\delta}^{2}+\frac{e^{-I_{\delta}^1}}{\sqrt{2}}H_{11}^{(0)})t^{-\frac{1}{2}}+\mathcal{O}(t^{-1}),\\
	&-i\lim_{z\to 0\in\mathbb{C}^+}\frac{[N(z)]_{22}-[N(0)]_{22}}{z}= \sqrt{2}H_{21}^{(1)}e^{-I_{\delta}^1}+\frac{e^{I_{\delta}^1}}{\sqrt{2}} +(\frac{e^{I_{\delta}^1}}{\sqrt{2}}H_{22}^{(0)}-\sqrt{2}H_{21}^{(0)}e^{I_{\delta}^1}I_{\delta}^{2})t^{-\frac{1}{2}}+\mathcal{O}(t^{-1}).
\end{align*}}
Substituting  above estimates  into (\ref{recons x}) and  (\ref{recons u}) leads to
\begin{align}
&	u(x,t)= u(y(x,t),t)=
\frac{1}{2}+2H_{12}^{(1)}e^{2I_{\delta}^1}-H_{21}^{(1)}e^{-2I_{\delta}^1} \nonumber\\
&+(-\frac{1}{2}H_{22}^{(0)}+H_{21}^{(0)}I_{\delta}^{2}+2H_{12}^{(0)}I_{\delta}^{2}+H_{11}^{(0)}(\frac{5}{2}+H_{21}^{(1)}e^{-2I_{\delta}^1}+2H_{12}^{(1)}e^{2I_{\delta}^1}))t^{-\frac{1}{2}} +\mathcal{O}(t^{-1}), \nonumber\\
&x=y-2\ln(\sqrt{2}(1+H^{(0)}_{22}t^{-\frac{1}{2}})e^{-I_{\delta}^1})+\mathcal{O}(t^{-1}).\nonumber
\end{align}

For the region  II,  we have done the sequence of transformations
\begin{align}
	N(z)\to M^{(1)}(z)\to M^{(2)}(z)\to M^{modc}(z),
\end{align}
from which we find that
\begin{align}
N(z)=&\delta(\infty)^{\sigma_3}E(z;\xi,c)M^{modc}(z)\delta(z)^{-\sigma_3}G(z)^{-1}e^{-it(p_--\theta_+)\sigma_3},\label{ope}
\end{align}
which definition is in Prop.\ref{prode1}, RHP \ref{modelc}, (\ref{asyEr1}), (\ref{funcG}), (\ref{the+0}) and (\ref{theta+}) respectively.
To  reconstruct   $u(x,t)$ by using (\ref{recons u}),    in above  equation (\ref{ope})  we  take    $z\to 0$  in $\mathbb{C}^+$,
  we  obtain that
\begin{align}
	N(z)=&\delta(\infty)^{\sigma_3}N^{modc}(z)e^{-I_\delta^1\sigma_3}(I-I_\delta^2\sigma_3z)e^{-it(p_--\theta_\pm)(0^+)\sigma_3}+\mathcal{O}(z^2)+\mathcal{O}(t^{-2})\nonumber\\
	=&\sqrt{2}i\delta(\infty)^{\sigma_3}\sigma_1e^{-I_\delta^1\sigma_3}e^{-it(p_--\theta_\pm)(0^+)\sigma_3}\nonumber\\
	&+z\delta(\infty)^{\sigma_3}\left( \frac{1}{c\sqrt{2}}e^{-I_\delta^1\sigma_3}-\sqrt{2}I_\delta^2e^{-I_\delta^1\sigma_3}\sigma_2\right) e^{-it(p_--\theta_\pm)(0^+)\sigma_3}+\mathcal{O}(z^2)+\mathcal{O}(t^{-2}),\nonumber
\end{align}
which implies that
\begin{align*}
&-i[	N(0)]_{12}=\sqrt{2}\delta(\infty)e^{I_\delta^1+it(p_--\theta_\pm)(0^+)},\\
&-i\lim_{z\to 0\in\mathbb{C}^+}\frac{[N(z)]_{11}-[N(0)]_{11}}{z}=\frac{-i\delta(\infty)}{c\sqrt{2}}e^{-I_\delta^1-it(p_--\theta_\pm)(0^+)},\\
&-i\lim_{z\to 0\in\mathbb{C}^+}\frac{[N(z)]_{22}-[N(0)]_{22}}{z}=\frac{-i\delta(\infty)}{c\sqrt{2}}e^{I_\delta^1+it(p_--\theta_\pm)(0^+)}.
\end{align*}
Substituting  above estimates  into (\ref{recons x}) and  (\ref{recons u}) leads to
\begin{align}
&u(x,t)= u(y(x,t),t)=
	\frac{-i}{c}\left( \delta(\infty)^2+\frac{1}{2}\right)  +\mathcal{O}(t^{-2}),\nonumber\\
&x(y,t) =y-2\ln\left(\sqrt{2}\delta(\infty)e^{I_\delta^1+it(p_--\theta_\pm)(0^+)}\right)  +\mathcal{O}(t^{-2}).\nonumber
\end{align}

For the regions III and IV,  their solving  processes  are   similar,  we take $\xi>1$ part of region   III  as an example.
The  the sequence of transformations is
\begin{align*}
	N(z)\to M^{(1)}(z )\to M^{(2)}(z ) =\left\{\begin{array}{ll}
		E(z;\xi,c)M^{mod}(z;\xi,c), & z\notin U(\xi),\\[4pt]
		E(z;\xi,c)M^{lo,+}(z;\xi,c),  &z\in U(z_0),\\[4pt]
		E(z;\xi,c)M^{lo,-}(z,\xi),  &z\in U(-z_0),  
	\end{array}\right.
\end{align*}
which gives
\begin{align}
	N(z)=&e^{-itg(\infty)\sigma_3}D_\infty(z)^{\sigma_3}E(z;\xi,c)M^{mod}(z)D(z)^{-\sigma_3} G (z)^{-1}e^{-it(p_--g)\sigma_3}.\nonumber
\end{align}
To take $z\to 0$  in $\mathbb{C}^+$, further using  the asymptotic expansion in (\ref{ginf}), Prop.\ref{proD}, Prop.\ref{prog}, (\ref{Mmod0}) and  (\ref{funcG2}),
  we obtain
\begin{align}
	N(z)=&e^{-itg(\infty)\sigma_3}D_\infty(0)^{\sigma_3}M^{mod}(0^+)D(0)^{-\sigma_3}e^{-it(p_--g)(0)\sigma_3}\nonumber\\
	&+t^{-1}e^{-itg(\infty)\sigma_3}D_\infty(0)^{\sigma_3}H^{(0)}M^{mod}(0^+)D(0)^{-\sigma_3}e^{-it(p_--g)(0)\sigma_3}\nonumber\\
	&+ze^{-itg(\infty)\sigma_3}D_\infty(0)^{\sigma_3}\left( D_\infty^{(1)}\sigma_3M^{mod}(0^+)D(0)^{-\sigma_3}+M^{mod}_1D(0)^{-\sigma_3}\right.\nonumber\\
	&\left.-M^{mod}(0^+)D(0)^{-\sigma_3}D^{(1)}\sigma_3\right)e^{-it(p_--g)(0)\sigma_3} \nonumber\\
	&+zt^{-1}e^{-itg(\infty)\sigma_3}D_\infty(0)^{\sigma_3}H^{(1)}M^{mod}(0^+)D(0)^{-\sigma_3}e^{-it(p_--g)(0)\sigma_3} +\mathcal{O}(z^2)+\mathcal{O}(t^{-2}).\nonumber
\end{align}
Substituting  above equation  into (\ref{recons x}) and  (\ref{recons u}) leads to
\begin{align}
	u(x,t)=&u(y(x,t),t)=u_{g,D,\xi}(y,t)+t^{-1}\mathcal{E}(\xi) +\mathcal{O}(t^{-2}),
\end{align}
where  we denote
{\small \begin{align}
	&u_{g,D,\xi}(y,t)= -e^{-2itg(\infty)}D_\infty(0)^2[M^{mod}(0^+)]_{12}\left[ \left(D_\infty^{(1)}-D^{(1)}\right)[M^{mod}(0^+)]_{11}+[M^{mod}_1]_{11}\right]  \nonumber\\
	&-e^{2itg(\infty)}D_\infty(0)^{-2}[M^{mod}(0^+)]_{12}^{-1} \left[ \left(D^{(1)}-D_\infty^{(1)}\right) [M^{mod}(0^+)]_{22}+[M^{mod}_1]_{22} \right],\label{1}
\end{align}}
and
{\small \begin{align}
	\mathcal{E}&(\xi)=
	- e^{-2itg(\infty)}D_\infty(0)^2 [M^{mod}(0^+)]_{12}\left[[H^{(1)}]_{11}[M^{mod}(0^+)]_{11}+[H^{(1)}]_{12}[M^{mod}(0^+)]_{21} \right] \nonumber\\
	-&e^{-2itg(\infty)}D_\infty(0)^2\left( [H^{(0)}]_{11}[M^{mod}(0^+)]_{12}+[H^{(0)}]_{12}[M^{mod}(0^+)]_{22}\right) \nonumber\\
	&\times\left[ \left(D_\infty^{(1)}-D^{(1)}\right)[M^{mod}(0^+)]_{11}+[M^{mod}_1]_{11}\right]   \nonumber\\
	 -&e^{2itg(\infty)}D_\infty(0)^{-2}\dfrac{[H^{(1)}]_{21}[M^{mod}(0^+)]_{12}+[H^{(1)}]_{22}[M^{mod}(0^+)]_{22}}{[M^{mod}(0^+)]_{12}}   \nonumber\\
	 +&e^{2itg(\infty)}D_\infty(0)^{-2}\dfrac{[H^{(0)}]_{11}[M^{mod}(0^+)]_{12}+[H^{(0)}]_{12}[M^{mod}(0^+)]_{22}}{[M^{mod}(0^+)]_{12}^2}  \nonumber\\
	 & \times\left[ \left(D^{(1)}-D_\infty^{(1)}\right) [M^{mod}(0^+)]_{22}+[M^{mod}_1]_{22} \right] .\label{11}
\end{align}}
Moreover, we have
\begin{align}
	x(y,t)=&y-2\ln\left(-ie^{-itg(\infty)+it(p_--g)(0)}D_\infty(0)D(0)[M^{mod}(0^+)]_{12}\right)\nonumber\\
	&+2i\dfrac{[H^{(0)}]_{11}[M^{mod}(0^+)]_{12}+[H^{(0)}]_{12}[M^{mod}(0^+)]_{22}}{[M^{mod}(0^+)]_{12}}t^{-1}  +\mathcal{O}(t^{-2}),\label{111}
\end{align}
where   $H^{(0)}$ and $H^{(1)}$ is in Prop.\ref{asyE} and Prop.\ref{asyE2} corresponding to different case of $\xi>1$ and $\xi<1$ part of region III.

Finally,  summing up above results  gives  our  main theorem in this paper.
\begin{theorem}\label{last}   Let $u(x,t)$ be the solution for  the initial-value problem (\ref{mcho}) and (\ref{initial}),
then there exist a large constant $T_1=T_1(\xi), \  \xi=\frac{y}{t}$,  such that   for all $T_1<t\to\infty$,
the   long time asymptotics of the mCH equation (\ref{mcho}) are
given as follows.

\begin{itemize}

\item[$\blacktriangleright$ ] {\rm\bf The  region I:}   $(i)  \ \xi<\frac{3}{4}; \   (ii)\  1< c\leq \lambda_1, \  \frac{3}{4}<\xi<1; \   (iii)\ 1< c\leq \lambda_1, \  1<\xi,$
which  is a slow decay step-like background constant region with genus-0. We have asymptotic expansion
\begin{align}
	&u(x,t)= u(y(x,t),t)=
	\frac{1}{2}+2H_{12}^{(1)}e^{2I_{\delta}^1}-H_{21}^{(1)}e^{-2I_{\delta}^1} \nonumber\\
	&+\left( -\frac{1}{2}H_{22}^{(0)}+H_{21}^{(0)}I_{\delta}^{2}+2H_{12}^{(0)}I_{\delta}^{2}+H_{11}^{(0)}(\frac{5}{2}+H_{21}^{(1)}e^{-2I_{\delta}^1}+2H_{12}^{(1)}e^{2I_{\delta}^1})\right) t^{-\frac{1}{2}} +\mathcal{O}(t^{-1}), \nonumber\\
& x=y-2\ln\left( \sqrt{2}(1+H^{(0)}_{22}t^{-\frac{1}{2}})e^{-I_{\delta}^1}\right)+\mathcal{O}(t^{-1}), \nonumber
\end{align}
where the $I^1_{\delta}$ and  $I^2_{\delta}$ are given by   (\ref{Ide0})-(\ref{Ide0}),
$H^{(0)}$ and $ H^{(1)}$ are given  in Proposition \ref{AsyE0I}.

\item[$\blacktriangleright$ ] {\rm\bf The  region II:} $\xi>1+2/c$, which a  fast decay step-like background constant region with genus-0. We have asymptotic expansion
\begin{align}
&	u(x,t)= u(y(x,t),t)=
	 -i c^{-1} \left( \delta(\infty)^2+1/2 \right)  +\mathcal{O}(t^{-2}),\nonumber\\
&	x(y,t) =y-2\ln\left(\delta(\infty)e^{I_\delta^1+ia(y,t)}\right)  +\mathcal{O}(t^{-2}),  \nonumber
\end{align}
where  $a(y,t)=-\frac{i}{2}(  c+1)y+\frac{it}{2}\left(  c^{-2} + c\right) $, $I_\delta^1$ and $\delta(\infty)$ are given   in Prop. \ref{prode1}.

\item[$\blacktriangleright$ ] {\rm\bf The  region III:}  Genus-2 elliptic wave region.

(i) \  $1<\xi<1+\frac{2}{c}$, $c>\sqrt{2} $;  (ii) \ $1+\frac{2}{c^4}(2-c^2)<\xi<1+\frac{2}{c}$, $ \sqrt{2}<c<2$, we have asymptotic expansion
\begin{align*}
&u(y(x,t),t)=u_{g,D,\xi}(y,t)+t^{-1}\mathcal{E}(\xi) +\mathcal{O}(t^{-2}),\nonumber\\
&x(y,t)= y-2\ln\left(-ie^{-itg(\infty)+it(p_--g)(0)}D_\infty(0)D(0)[M^{mod}(0^+)]_{12}\right)\nonumber\\
	&+2i\dfrac{[H^{(0)}]_{11}[M^{mod}(0^+)]_{12}+[H^{(0)}]_{12}[M^{mod}(0^+)]_{22}}{[M^{mod}(0^+)]_{12}}t^{-1}  +\mathcal{O}(t^{-2}),
\end{align*}
where  $u_{g,D,\xi}(y,t)$, $\mathcal{E}(\xi)$, $g(\infty)$, $g(z)$, $D_\infty(0)$, $D(0)$,  $M^{mod}$,  $H^{(0)}$ and $H^{(1)}$ are
show in (\ref{1}), (\ref{11}),  (\ref{ginf}), Prop.\ref{prog}, Prop.\ref{proD},   (\ref{Mmod0}) and Prop.\ref{asyE},  respectively.

(iii) \  $\xi_m <  \xi<1$, $c>2$,  we have asymptotic expansion
\begin{align*}
	&u(y(x,t),t)=u_{g,D,\xi}(y,t)+t^{-1/2}\mathcal{E}(\xi) +\mathcal{O}(t^{-1}),\nonumber\\
	&x(y,t)= y-2\ln\left(-ie^{-itg(\infty)+it(p_--g)(0)}D_\infty(0)D(0)[M^{mod}(0^+)]_{12}\right)\nonumber\\
&{\small+2i\dfrac{[H^{(0)}]_{11}[M^{mod}(0^+)]_{12}+[H^{(0)}]_{12}[M^{mod}(0^+)]_{22}}{[M^{mod}(0^+)]_{12}}t^{-1/2}  +\mathcal{O}(t^{-1}),}
\end{align*}
 where  $u_{g,D,\xi}(y,t)$, $\mathcal{E}(\xi)$, $g(\infty)$, $g(z)$, $D_\infty(0)$, $D(0)$,  $M^{mod}$,  $H^{(0)}$ and $H^{(1)}$ are show in (\ref{1}), (\ref{11}),  (\ref{ginf}),  (\ref{Mmod0}), Prop.\ref{prog}, Prop.\ref{proD2}  and Prop.\ref{asyE2},  respectively. Although has same sign, $\mathcal{E}(\xi)$, $H^{(0)}$ and $H^{(1)}$ represent the   contribution  of the pairs of stationary phase point out of cut via parabolic cylinder model.

\item[$\blacktriangleright$ ] {\rm\bf The  region IV:} \  $\frac{3}{4}<\xi<\xi_m$, $2<c$, which is a genus-2 elliptic wave region. We have asymptotic expansion
\begin{align*}
	&u(y(x,t),t)=u_{g,D,\xi}(y,t)+t^{-1}\mathcal{E}(\xi) +\mathcal{O}(t^{-2}),\nonumber\\
	&x(y,t)=y-2\ln\left(-ie^{-itg(\infty)+it(p_--g)(0)}D_\infty(0)D(0)[M^{mod}(0^+)]_{12}\right)\nonumber\\
	&+2i\dfrac{[H^{(0)}]_{11}[M^{mod}(0^+)]_{12}+[H^{(0)}]_{12}[M^{mod}(0^+)]_{22}}{[M^{mod}(0^+)]_{12}}t^{-1}  +\mathcal{O}(t^{-2}).
\end{align*}
where  $u_{g,D,\xi}(y,t)$, $\mathcal{E}(\xi)$, $g(\infty)$, $D_\infty(0)$, $D(0)$,  $g(z)$, $M^{mod}$,  $H^{(0)}$ and $H^{(1)}$ are show in (\ref{1}), (\ref{11}),    Prop. \ref{proD3}, Prop. \ref{prog2}  and Prop. \ref{asyE3},  respectively. Although has same sign, $\mathcal{E}(\xi)$, $H^{(0)}$ and $H^{(1)}$ represent the  common contribution  of two local Airy Model of two pairs of stationary phase points.
\end{itemize}
\end{theorem}

\appendix
\section{Appendix. The RH model for  Airy function}\label{appairy}
\quad  In this appendix,  we recall
 the standard model RH problem  of Airy function that is used in our paper. Let $\Gamma=\Gamma_1\cup\Gamma_2\cup\Gamma_3\cup\Gamma_4
\subset\mathbb{C}$ denote the rays
\begin{eqnarray}
	\Gamma_1:=\{ze^{\frac{2i\pi}{3}}|z\in\mathbb{R}^+\},	\Gamma_2:=\{-z|z\in\mathbb{R}^+\},\\
	\Gamma_3:=\{ze^{\frac{4i\pi}{3}}|z\in\mathbb{R}^+\},\hspace{0.5cm}
	\Gamma_4:=\{z|z\in\mathbb{R}^+\}.
\end{eqnarray}
  The corresponding   open sectors are given as  follows
\begin{eqnarray}
	S_1=\{z|\arg z\in(0, {2\pi}/{3})\},
	S_2=\{z|\arg z\in( {2\pi}/{3},\pi)\},\\
	S_3=\{z|\arg z\in(\pi, {4\pi}/{3})\},
	S_4=\{z|\arg z\in( {4\pi}/{3},2\pi)\}.
\end{eqnarray}
Let $\chi=e^{\frac{2i\pi}{3}}$ and the function $m^{Ai}(z)$ for $z\in\mathbb{C}\setminus\Gamma$ by
\begin{align}
		m^{Ai}(z)=\mathcal{A}(z)\left\{\begin{array}{llll}
		e^{\frac{2}{3}z^{\frac{3}{2}}\sigma_3}, &z\in S_1,\\[8pt]
			\left(\begin{array}{cc}
				1 & 0\\
				-1 & 1
			\end{array}\right)e^{\frac{2}{3}z^{\frac{3}{2}}\sigma_3},  &z\in S_2,\\[10pt]
			\left(\begin{array}{cc}
				1 & 0\\
				1 & 1
			\end{array}\right)e^{\frac{2}{3}z^{\frac{3}{2}}\sigma_3},  &z\in S_3,\\[10pt]
			e^{\frac{2}{3}z^{\frac{3}{2}}\sigma_3},  &z\in S_4.\\
		\end{array}\right.
	\end{align}
\begin{align}
	\mathcal{A}(z)=\left\{\begin{array}{ll}
		\left(\begin{array}{cc}
			\text{Ai}(z) & \text{Ai}(\chi^2z)\\
			\text{Ai}'(z) & \chi^2\text{Ai}'(\chi^2z)
		\end{array}\right)e^{-\frac{\pi}{6}\sigma_3},  &\text{Im}z>0,\\[10pt]
		\left(\begin{array}{cc}
		\text{Ai}(z) & -\chi^2\text{Ai}(\chi z)\\
		\text{Ai}'(z) & -\text{Ai}'(\chi z)
	\end{array}\right)e^{-\frac{\pi}{6}\sigma_3},  &\text{Im}z<0,\\[10pt]
	\end{array}\right.
\end{align}
$\text{Ai}(z)$ is Airy function who is analytic in $\mathbb{C}$. \begin{lemma}
	$m^{Ai}:\mathbb{C}\setminus\Gamma\to\mathbb{C}^{2\times2}$ is a matrix valued analytic function and
satisfies the jump condition
 $$m^{Ai}_{+}(z)=m^{Ai}_-(z)v^{Ai}(z),$$
 where
\begin{align}
	v^{Ai}({z})=\left\{\begin{array}{lll}
			\left(\begin{array}{cc}
				1 & 0\\
				-e^{\frac{4}{3}z^{\frac{3}{2}}} & 1
			\end{array}\right),  &z\in \Gamma_1\cup\Gamma_3,\\[10pt]
			\left(\begin{array}{cc}
				0 & -1\\
				1 & 0
			\end{array}\right),  &z\in \Gamma_2,\\[10pt]
			\left(\begin{array}{cc}
				1 & -e^{\frac{4}{3}z^{\frac{3}{2}}}\\
				0& 1
			\end{array}\right),  &z\in \Gamma_4.\\
		\end{array}\right.
\end{align}
The asymptotic behavior of $m^{Ai}(z)$ as $z\to \infty$ can be shown as
\begin{equation}
	N^{-1}z^{\frac{\sigma_3}{4}}m^{Ai}(z)=I+\sum_{j=1}^{\infty}\frac{m_j^{Ai}}{z^{\frac{3j}{2}}}, z\to \infty,\label{asymlo}
\end{equation}
where
\begin{equation}
m^{Ai}_j=\frac{e^{\frac{i\pi}{4}}}{\sqrt{2}}N^{-1}\left(\begin{array}{cc}
	1&0\\
	0&-i\\
\end{array}\right)(\frac{3}{2})^{j}
\left(\begin{array}{cc}
(-1)^ju_j&u_j\\
-(-1)^jv_j&v_j\\
\end{array}\right)e^{-\frac{i\pi}{4}\sigma_3},\label{asyairy}
\end{equation}
and
\begin{equation*}
N=\frac{1}{\sqrt{2}}\left(\begin{array}{cc}
		1&i\\
		i&1\\
	\end{array}\right), \ \	u_j=\frac{(2j+1)(2j+3)...(6j-1)}{(216)^jj!},\ v_j=\frac{6j+1}{1-6j}u_j.
\end{equation*}
\end{lemma}

\section{Appendix. The existence of $\xi_m$ in elliptic region }\label{xim}

All discussion in the follows
 is under the case
 $$\frac{3}{4}<\xi<1, \ \  \eta-(\eta(\eta-4))^{1/2}<c^2,\  \ \eta:=  (1-\xi)^{-1}.$$
In a compact region
$$\Omega=\left\{(x,y)\in\mathbb{R}^2:\ \frac{1}{x^2}+\frac{1}{y^2}+\frac{xy}{\eta^2}\leq1,0<x\leq y\right\},$$
we consider  the function
{\small \begin{align}
		F(x,y)=\int_{x}^{y}\frac{1}{z^3}\left[\dfrac{z^2-x^2}{(z^2-1)(y^2-z^2)} \right]^{1/2}\left[ \frac{\xi-1}{2}z^4+\frac{y}{x}\left( 1+\frac{1}{y^2}-\frac{1}{x^2}\right) z^2-\frac{2y}{x}\right] \text{d}z \nonumber
\end{align}}
which  is  continuous  and   differential, moreover satisfies that
$$F(x,y)\equiv   0, \ \  (x,y) \in L := \left\{ (x,y):  x=y,\frac{2}{x^2}+\frac{x^2}{2\eta}< 1 \right\}.$$
Therefore two  positive phase points $z_1$ and $z_2$ can be defined as
\begin{align}
& z_1^2z_2^2 =\frac{4\eta y}{x}, \ \ 	z_1^2+z_2^2 =\frac{2\eta y}{x}(1+\frac{1}{y^2}-\frac{1}{x^2}).\nonumber
\end{align}
Then  $\max\lbrace z_1,z_2\rbrace\geq y$  leads to the following  two cases
\begin{align}
	&\frac{xy^3}{4\eta}\leq 1  \ \  {\rm or} \ \ \frac{xy^3}{4\eta}>1, \  \frac{1}{x^2}+\frac{1}{y^2}+\frac{xy}{2\eta}\leq1.
\end{align}
 While $\min\lbrace z_1,z_2\rbrace>x$ implies $\frac{x^5}{4\eta y}<1.$
Especially for  $(x,y)\in \Omega$,    $\max\lbrace z_1,z_2\rbrace\geq y$ and $\max\lbrace z_1,z_2\rbrace= y$ if and only if $(x,y)\in \partial\Omega\setminus L $.

Without loss of generality, we let $0<z_1<z_2$.
Direct  calculation  shows that
{\small \begin{align}
	F_x&=\frac{1-\xi}{2x^3}(x^2-z_1^2)(x^2-z_2^2)\int_{x}^{y} z \frac{z}{\sqrt{(z^2-1)(z^2-x^2)(y^2-z^2)}}\text{d}z,\nonumber\\
	F_y&=-\frac{1}{x}(1-\frac{1}{x^2}-\frac{1}{y^2}-\frac{(1-\xi)xy}{2})\int_{x}^{y}\frac{z^2-x^2}{y^2-z^2}\frac{z}{\sqrt{(z^2-1)(z^2-x^2)(y^2-z^2)}}\text{d}z.\nonumber
\end{align}}
Obviously, $F(x,y)=0$ in the interior of $\Omega$  implies $x<z_1<y$ then $F_x>0$. Because $F_x\neq0$ and $F_y \neq 0$ in the interior of $\Omega$, the curve $F(x,y)=0$ in the interior of must approach to the boundary of $\Omega$.

Meanwhile, we have
 $$F_x(x,x)=F_y(x,x)= -\frac{1}{\sqrt{x^2-1}x^3}\left( \frac{\xi-1}{2}x^4+ x^2-2\right)<0,\ \ (x,y)\in \Omega_1.$$
 which implies $F(x,y)=0$ in the interior of $\Omega$ won't approach to $L$.
Moreover, the point  $\left(  ( \eta+( \eta (\eta-4))^{1/2})^{1/2}, ( \eta+( \eta (\eta-4))^{1/2})^{1/2}\right)\in \partial\Omega$
is away from $\lbrace(x,y)|\frac{x^5}{\eta y}<1\rbrace$, which implies the curve won't approach to  this point either.

Considering the boundary $\partial\Omega\setminus\lbrace(x,y)|x=y\rbrace$, in another word,
 $$\lbrace\frac{1}{x^2}+\frac{1}{y^2}+\frac{xy}{2\eta}=1|0<x\leq y\rbrace.$$
 It can be a curve with a positive parameter $w\in\lbrace w_0^{-1}<w<w_0|w_0+w_0^{-1}=\sqrt{\eta},w_0>1\rbrace$ as
{\small\begin{align}
	x=w\sqrt{\eta}\left(\sqrt{1- w  \eta^{-1/2} + w^{-1} \eta^{-1/2}} -\sqrt{1- w  \eta^{-1/2} - w^{-1} \eta^{-1/2}}\right),\\
	y=w\sqrt{\eta}\left(\sqrt{1- w  \eta^{-1/2} + w^{-1} \eta^{-1/2}}+\sqrt{1- w  \eta^{-1/2} - w^{-1} \eta^{-1/2}}\right).
\end{align}}
Directly calculating can derive
{\small \begin{align}\label{Fw}
	\frac{\partial F}{\partial w}=\frac{2\left(6\sqrt{\eta}w^2-4\eta w+\frac{2\sqrt{\eta}}{w^2}\right)}{\sqrt{1-\frac{w}{\sqrt{\eta}}+\frac{1}{w\sqrt{\eta}}}
\sqrt{1-\frac{w}{\sqrt{\eta}}-\frac{1}{w\sqrt{\eta}}}}\int_{x}^{y}\frac{(w-\sqrt{\eta})z^2+2\sqrt{\eta}}{z\sqrt{(z^2-1)(z^2-x^2)(y^2-z^2)}}\text{d}z.\nonumber
\end{align}}
The above  integral   is negative in $w\in\lbrace w_0^{-1}<w<w_0|w_0+w_0^{-1}=\sqrt{\eta},w_0>1\rbrace$, $F(w_0^{-1})=F(w_0)=0$ and
 the equation
 $$6\sqrt{\eta}w^2-4\eta w+\frac{2\sqrt{\eta}}{w^2}=0$$
admits  two solution   between $w_0^{-1}$ and $w_0$,
 which implies  there is one unique zeros of $F(w)$ in $w\in\lbrace w_0^{-1}<w<w_0|w_0+w_0^{-1}=\sqrt{\eta},w_0>1\rbrace$ written as $w=w_m$. And $\frac{\text{d}F}{dw}(w_m)\neq0$ implies there must exist one unique branch of curve $F(x,y)=0$  approaches to the point $(x(w_m),y(w_m))$. And another possible point is
$$\left(  ( \eta-( \eta (\eta-4))^{1/2})^{1/2}, ( \eta-( \eta (\eta-4))^{1/2})^{1/2}\right)=(x(w_0^{-1}),y(w_0^{-1})).$$
For $y $ decrease to  $( \eta-( \eta (\eta-4))^{1/2})^{1/2}=y(w_0^{-1})$,
$$F(x(w),y(w))<0, \ \   F_x (y(w))<0 $$
implies there exist $x=x(y)$ satisfies $F(x,y)=0$. And every $x$ satisfies $F(x,y)=0$ has $\frac{\text{d}F}{dx}>0$ implies the zeros is unique.
 Therefore, there must exist a unique curve satisfies $F(x,y)=0$ from $(x(w_0^{-1}),y(w_0^{-1}))$ to $(x(w_m),y(w_m))$ in  $\Omega$.
 Note $y(w_m)$ as $c_m$. In another word, for every $\xi<1$, $c\in ( ( \eta+( \eta (\eta-4))^{1/2})^{1/2},c_m(\xi))$,
 there exist a $z_0=z_0(c)$ satisfies $F(z_0,c)=0$ and $z_2(z_0,c_m)=c_m$ when $c$ comes to  $c_m$. Let $\xi_m(c)=c_m^{-1}(c)$. It
 is well-defined  because  $c_m=y(w_m)$  is continuous  with respect to  $\xi$, which implies $c_m(\xi)$ is a surjection to $c^2>4$. Then the existence of $\xi_m$ has been proved.\vspace{4mm}

\noindent\textbf{Acknowledgements}

This work is supported by  the National Science
Foundation of China (Grant No. 11671095,51879045).
\hspace*{\parindent}

\end{document}